\documentclass[12pt]{amsart}

\usepackage{amssymb}
\usepackage{amsmath}

\usepackage{color}
\usepackage{stmaryrd}

\usepackage{algorithm}
\usepackage{algorithmic}
\usepackage{graphicx,subfigure}
\usepackage[normalem]{ulem}
\usepackage{cancel}
\usepackage{marginnote}
\usepackage{fullpage, paralist,enumerate}

\newcommand{\nwc}{\newcommand*}
\newcommand{\commentout}[1]{}

\nwc{\red}{\color{red}}
\nwc{\blue}{\color{blue}}

\nwc{\nn}{\nonumber}
\newtheorem{thm}{Theorem}[section]
\newtheorem{lem}[thm]{Lemma}
\newtheorem{prop}[thm]{Proposition}

\newtheorem{ex}[thm]{Example}
\newtheorem{rmk}[thm]{Remark}
\newtheorem{defi}[thm]{Definition}

\nwc{\beq}{\begin{eqnarray}}
\nwc{\eeq}{\end{eqnarray}}
\nwc{\beqn}{\begin{eqnarray*}}
\nwc{\eeqn}{\end{eqnarray*}}

\nwc{\IE}{\mathbb{E}}
\nwc{\IR}{\mathbb{R}}
\nwc{\IN}{\mathbb{N}}
\nwc{\IC}{\mathbb{C}}
\nwc{\IZ}{\mathbb{Z}} 
\nwc{\II}{\mathbb{I}}

\nwc{\mb}{\mathbf}
\nwc{\ml}{\mathcal}

\newcommand{\lt}{\left}
\newcommand{\rt}{\right}
\nwc{\ep}{\varepsilon}
\nwc{\lamb}{\lambda}
\nwc{\sgn}{\mbox{\rm sgn}}
\nwc{\lan}{\langle}
\nwc{\ran}{\rangle}
\nwc{\lb}{\llbracket}
\nwc{\rb}{\rrbracket}
\nwc{\im}{{\rm i}}
\nwc{\bk}{{\mb k}}
\nwc{\bn}{{\mb n}}
\nwc{\mbm}{\mathbf{m}}
\nwc{\cM}{\mathcal{M}}
\nwc{\cT}{\mathcal{T}}
\nwc{\cS}{\mathcal S}
\nwc{\bom}{\mathbf{w}}
\nwc{\bM}{\mathbf{M}}
\nwc{\cL}{\mathcal{L}}
\nwc{\half}{{1\over 2}}
\nwc{\cN}{\mathcal{N}}
\nwc{\bN}{\mathbf{N}}
\nwc{\bt}{\mb t}
\nwc{\bkk}{\mb k}
\nwc{\bll}{\mb l}
\nwc{\bs}{\mb s}
\nwc{\bv}{\mb v}
\nwc{\bu}{\mb u}
\nwc{\cG}{{\ml G}}
\nwc{\cB}{{\ml B}}
\nwc{\cF}{{\ml F}}
\nwc{\prox}{\hbox{prox}}
\nwc{\cO}{{\ml O}}
\nwc{\pxp}{P_X^\perp}
\nwc{\px}{P_X}
\nwc{\py}{P_Y}
\nwc{\rx}{R_X}
\nwc{\bw}{\mathbf{w}}
\nwc{\mbe}{\mathbf{e}}
\nwc{\be}{\mathbf{e}}
\nwc{\br}{\mb r}
\nwc{\yeta}{\tilde{y}}
\nwc{\bomg}{\boldsymbol{\omega}}

\nwc{\x}{x_*}
\nwc{\y}{y_*}
\nwc{\muh}{\nu}
\nwc{\xh}{x}

\newcommand{\CC}{\mathbb{C}}
\newcommand{\ZZ}{\mathbb{Z}}
\newcommand{\EE}{\mathbb{E}}

\def\STFT{{\mathcal V}}

\newcommand{\xo}{x_{*}}
\newcommand{\eps}{\varepsilon}
\newcommand{\Tp}{\mathcal{T}^{\perp}}

\newcommand{\ignore}[1]{}
\newcommand{\supp}{\operatorname{supp}}
\newcommand{\diag}{\operatorname{diag}}
\newcommand{\dist}{\operatorname{dist}}
\newcommand{\trace}{\operatorname{Tr}}
\newcommand{\rank}{\operatorname{rank}}
\newcommand{\polylog}{\operatorname{polylog}}

\newcommand{\auto}{\mathfrak{A}}

\newcommand{\cA}{\mathcal{A}}
\newcommand{\cH}{\mathcal{H}}
\newcommand{\Hn}{{\cH_n}}
\newcommand{\range}{\mathcal{R}}

\usepackage{scalerel,stackengine}
\stackMath
\newcommand\widecheck[1]{%
\savestack{\tmpbox}{\stretchto{%
  \scaleto{%
    \scalerel*[\widthof{\ensuremath{#1}}]{\kern-.6pt\bigwedge\kern-.6pt}%
    {\rule[-\textheight/2]{1ex}{\textheight}}
  }{\textheight}%
}{0.5ex}}%
\stackon[1pt]{#1}{\scalebox{-1}{\tmpbox}}%
}

\title[The Numerics of Phase Retrieval]{The Numerics of Phase Retrieval}
\author[{Acta Numerica}]{%
Albert Fannjiang\address{Department of Mathematics, University of California  Davis, Davis CA {\tt fannjiang@math.ucdavis.edu}}
 \and Thomas Strohmer 
 \address{Center for Data Science and Artificial Intelligence Research, UC Davis {\tt strohmer@math.ucdavis.edu}}
}

\begin{document}

\label{firstpage}
\maketitle

\begin{abstract}
Phase retrieval, i.e., the problem of recovering a function from the squared magnitude of its Fourier transform, arises
in many applications such as X-ray crystallography, diffraction imaging, optics, quantum mechanics, and astronomy. 
This problem has confounded engineers, physicists, and mathematicians for many decades. 
Recently, phase retrieval has seen a resurgence in research activity, ignited by new imaging modalities and novel mathematical concepts.
As our scientific experiments produce larger and larger datasets and we aim for faster and faster throughput,  it becomes increasingly important to study the involved numerical algorithms in a systematic and principled manner. Indeed, the last decade has witnessed a surge in the systematic study of computational algorithms for phase retrieval.  In this paper we will review these recent advances from a numerical viewpoint.

\end{abstract}

\tableofcontents 

\section{Introduction}
\label{s:intro}

When algorithms fail  to produce correct results in real world applications,  we would like to know why they failed. Is it because of some mistakes in the experimental setup, corrupted measurements, calibration errors, incorrect modeling assumptions, or is it due to a deficiency of the algorithm itself? If it is the latter, can it be fixed by a better initialization, a more careful tuning of the parameters, or by choosing a different algorithm? Or is a more fundamental modification required, such as developing a different model, including additional prior information, taking more measurements, or a better compensation of calibration errors? As our scientific experiments produce larger and larger datasets and we aim for faster and faster throughput,  it becomes increasingly important to address the aforementioned challenges in a systematic and principled manner.
Thus, a rigorous and thorough study of computational algorithms both from a theoretical and numerical viewpoint is not a luxury, but  emerges as
an imperative ingredient towards effective data-driven discovery.

The last decade has witnessed a surge in the systematic study of numerical algorithms for the famous phase retrieval problem, i.e., the
problem of recovering a signal or image from the intensity measurements of its Fourier transform~\cite{Hur89,KST95}.
In many applications one would like to acquire information about an
object but  it is impossible or impractical to measure the
phase of a signal. We are then faced with the difficult task of reconstructing the object of interest
from these magnitude measurements.  Problems of this kind fall in the
realm of phase retrieval problems, and are notoriously difficult to solve numerically.  
In this paper we will review recent advances in the area of phase retrieval with a strong focus on numerical algorithms.

Historically, one of the first important applications of phase
retrieval is X-ray crystallography \cite{Mil90,Har93}, and today
this is still one of the most important applications.  In 1912, Max von Laue discovered the diffraction of X-rays by crystals.
In 1913, W.H~Bragg and his son W.L.~Bragg realized that one could determine crystal structure from X-ray diffraction patterns.
Max von Laue received the Nobel Prize in 1914 and the Braggs in 1915, marking the beginning of many more Nobel Prizes to be awarded for discoveries in the area of x-ray crystallography. Later, the Shake-and-Bake algorithm become of most successful direct methods for phasing single-crystal diffraction data  and opened a new era in research in  mapping the chemical structures of small molecules~\cite{hauptman1997}.

The phase retrieval problem permeates many other areas of imaging
science. For example, in 1980, David Sayre suggested to extend the approach of x-ray crystallography to non-crystalline specimens.  
This approach is today known under the name of Coherent Diffraction Imaging (CDI)~\cite{miao1999extending}. See~\cite{shechtman2015phase}
for a detailed discussion of the benefits and challenges of CDI.
Phase retrieval also arises in optics \cite{Wal63}, fiber optic communications~\cite{kumar2014fiber}, astronomical imaging \cite{DF87},
microscopy~\cite{MIS08}, speckle interferometry~\cite{DF87},
quantum physics~\cite{Rei44,Cor06}, and even in differential geometry \cite{BSV02}.  

In particular, X-ray tomography has become an invaluable tool in biomedical imaging to
generate quantitative 3D density maps of extended specimens at nanoscale~\cite{dierolf2010ptychographic}.  We refer to~\cite{Hur89,LBL02} for various instances of the phase problem and additional references. A review of phase retrieval in optical imaging can be found in~\cite{shechtman2015phase}. 

Uniqueness and stability properties from a mathematical viewpoint are reviewed in ~\cite{grohs2019}. We just note here that the very first mathematical findings regarding uniqueness related to the phase retrieval problem are Norbert Wiener's seminal results on spectral factorization~\cite{Wie32}.

Phase retrieval has seen a significant resurgence in activity in recent years. This resurgence is fueled by:
(i) the desire to image individual molecules and other nano-particles; (ii) new imaging capabilities 
{such as ptychography, single-molecule diffraction and  serial nanocrystallography,  as well as the availability of {X-ray free-electron lasers (XFELs) and} new X-ray synchrotron sources that provide extraordinary X-ray fluxes, see for example
\cite{chapman2011femtosecond,neutze2000potential,Mil06,Sca06,Bog08,MIS08,dierolf2010ptychographic,DM08}}; and (iii) the influx of novel mathematical concepts and ideas, spearheaded by~\cite{CSV2013,CESV2013} {as well as  deeper understanding of non-convex optimization methods such as Alternating Projections ~\cite{GS72} and Fienup's Hybrid-Input-Output (HIO) algorithm \cite{Fie82}}. 
These mathematical concepts include advanced methods from convex and non-convex optimization, techniques from random matrix theory, and  insights from algebraic geometry.

\bigskip

Let $x$ be a (possibly multidimensional) signal, then in its most basic form, the phase retrieval problem can be expressed as
\begin{equation}
 \text{Recover} \,\, x, \quad
\text{given} \quad |\hat{x}(\bomg)|^2 = \left| \int_{T} x({\bt}) e^{-2\pi \im \bt \cdot \bomg}\, d\bt \right|^2, \quad \bomg \in\Omega, \label{fourier1} 
\end{equation}
where $T$ and $\Omega$ are the domain of the signal $x$ and its Fourier transform $\hat{x}$, respectively (and the Fourier transform
in~\eqref{fourier1} should be understood as possibly multidimensional transform).

When we measure $|\hat{x}(\bomg)|^2$ instead of $\hat{x}(\bomg)$, we lose information about the phase of $x$. If we could somehow retrieve the phase of $x$, then it would be trivial to recover $x$---hence the term  {\em phase retrieval}.
Its origin comes from the fact that detectors can often times only record the squared modulus of the Fresnel or Fraunhofer diffraction pattern of the radiation that is scattered from an object. In such settings, one cannot measure the phase of the optical
wave reaching the detector and, therefore, much information about the scattered object or the optical field is
lost since, as is well known, the phase encodes a lot of the structural content of the image we wish to form.

Clearly, there are infinitely many signals that have the same Fourier magnitude. This includes simple modifications such as translations or reflections of a signal.  While in practice such trivial ambiguities are likely acceptable, there are infinitely many other signals sharing the same Fourier magnitude which do not arise from a simple transform of the original signal. Thus, to make the problem even 
theoretically solvable (ignoring for a moment the existence of efficient and stable numerical algorithms) additional information about the signal must be harnessed. To achieve this we can either assume prior knowledge on the structure of the underlying signal or we can somehow take additional (yet, still phaseless) measurements of $x$, or we pursue a combination of both approaches.

Phase retrieval problems are usually ill-posed and notoriously difficult to
solve. Theoretical conditions that guarantee uniqueness of the solution
for generic signals exist for certain cases. However, as mentioned in
\cite{LBL02} and \cite{unique}, these uniqueness results do not translate into
numerical computability of the signal from its intensity
measurements, or about the robustness and stability of commonly used
reconstruction algorithms. Indeed, many of the
existing numerical methods for phase retrieval rely on all
kinds of a priori information about the signal, and none of these
methods is proven to actually recover the signal.

{This is the main difference between inverse and optimization problems: the latter focuses on minimizing the loss function
while the former emphasizes minimization of reconstruction error of the unknown object. The bridge between the loss function
and the reconstruction error depends precisely on the measurement schemes which are domain-dependent. 
}

Practitioners, not surprisingly, care less about theoretical guarantees of phase retrieval algorithms as long as they perform reasonably well in practice. Yet, it is a fact that algorithms do not always succeed. And then we want to know what went wrong. Was it a fundamental misconception in the experimental setup? After all, Nature does not alway cooperate. Was is due to underestimating measurement noise or unaccounted-for calibration errors? How robust is the algorithm in presence of corrupted measurements or perturbations cause by lack of calibration?
How much parameter tuning is acceptable when we deal with large throughput of data? All these questions require a systematic empirical study  of algorithms combined with a careful theoretical numerical analysis. 
This paper provides a snapshot  from an algorithmic viewpoint of recent activities in the applied mathematics community in this field. 
{In addition to traditional convergence analysis, we give equal attention to the sampling schemes and the data structures. }

\subsection{Overview}

In Section~\ref{s:setup} we introduce the main setup, some mathematical notation, and introduce various measurement techniques arising
in phase retrieval, such as coded diffraction illumination and ptychography. Section~\ref{s:uniqueness} is devoted to questions of uniqueness and feasibility. We also analyze various noise models. Nonconvex optimization methods are covered in Section~\ref{s:nonconvex}. We first review and analyze iterative projection methods, such as alternating projections, averaged alternating reflections, and the Douglas-Rachford splitting. We also review issues of convergence.  We then analyze gradient descent methods and  the Alternating Direction Method of Multipliers in detail. We discuss convergence rates, fixed points, and robustness of these algorithms. The question of the right initialization method is the contents of Section~\ref{s:init}, as initialization plays a key role for the performance of many algorithms.
 In Section~\ref{s:convex} we introduce various convex optimization methods for phase retrieval, such as PhaseLift and convex methods without ``lifting''. We also discuss applications in quantum tomography and how to take advantage of signal sparsity.
Section~\ref{sec:blind-ptycho} focuses on blind ptychography. We describe connections to time-frequency analysis, discuss in detail ambiguities arising in blind ptychography and describe a range of blind reconstruction algorithms. Holographic coded diffraction imaging is the topic of Section~\ref{s:holo}.
We conclude in Section~\ref{s:outlook}.

\section{Phase retrieval and ptychography: basic setup}\label{s:setup}

\subsection{Mathematical formulation}

There are many ways in which one can pose the phase-retrieval problem,
for instance depending upon whether one assumes a continuous or
discrete-space model for the signal. In this paper, we consider discrete
length signals (one-dimensional or multi-dimensional) for simplicity,
and because numerical algorithms ultimately operate with digital data. Moreover, for the same reason we will often focus on finite-length signals.
We refer to~\cite{grohs2019} and the many references therein regarding the similarities and delicate differences arising between the discrete and the continuous setting.

To fix ideas, suppose our object of interest is represented by a discrete signal
$x(\bn), \bn = (n_1,n_2,\cdots,n_d)\in \IZ^d.$ 
Define the  Fourier transform of $\x$ as
\beqn
\sum_{\bn} \x(\bn)
   e^{- 2\pi\im \bn\cdot \bom},\quad \bom \in \Omega. 
\eeqn
We denote the Fourier transform operator by $F$ and $F^{-1}$ is its
inverse Fourier transform\footnote{Here, $F$ may correspond to a one- or multi-dimensional Fourier transform,  and operate in the continuous, discrete, or finite domain. The setup will become clear from the context.}. The phase retrieval problem consists in finding $x$ from the magnitude coefficients
$|(F x)[\bomg]|$, $\bomg \in \Omega$. Without further information about the
unknown signal $x$, this problem is in general ill-posed since there are many
different signals whose Fourier transforms have the same
magnitude. Clearly, if $x$ is a solution to the phase retrieval
problem, then (i) $c x$ for any scalar $c \in \CC$ obeying $|c| = 1$ is
also solution, (ii) the ``mirror function'' or time-reversed signal
$\bar{x}(-\bt)$ is also solution, and (iii) the shifted signal $x(\bt - \bs)$ is also
a solution. From a physical viewpoint these ``trivial associates'' of
$x$ are usually acceptable ambiguities. But in general infinitely many
solutions can be obtained from $\{|\hat{x}(\bomg)| : \bomg \in
\Omega\}$ beyond these trivial associates \cite{San85}. 

Most phase retrieval problems are formulated in 2-D, often with the ultimate goal to reconstruct--via tomography--a 3-D structure.
But phase retrieval problems also arise in 1-D (e.g.\ fiber optic communications) and potentially even 4-D (e.g.\ mapping the dynamics of biological structures).

Thus,  we formulate the phase retrieval problem in a more general way as follows:
Let $x \in \CC^n$ and  $a_k \in \CC^n$:
\begin{equation}\label{eq:data}
 \text{Recover $x$, \quad given} \,\, y_k = | \langle x, a_k\rangle |^2, \quad k = 1, \ldots, N.
\end{equation}
Here, $x$ and the $a_k$ can represent  multi-dimensional  signals. Here, we assume intensity measurements
but obviously the problem is equivalent from a theoretical viewpoint if we assume magnitude measurements
\beqn
b_k = | \langle x, a_k\rangle |, \quad k = 1, \ldots, N.
\eeqn
To ease the burden of notation,
when $x$ represents an image and the two-dimensionality of $x$ is essential for the presentation, we often will denote its dimension as $x\in \CC^{n\times n}$ (instead of the more cumbersome notation $x \in \CC^{\sqrt{n}\times \sqrt{n}}$),  in which case the total number of unknowns is $n^2$. In other cases, when the dimensionality of  $x$ is less relevant to the analysis, we will simply consider $x \in \CC^n$, where $x$ may be one- or multi-dimensional. The dimensionality will be clear from the context.

Also, the measurement vectors
$a_k$ can represent different measurement schemes (e.g.\ coded diffraction imaging, ptychography,...) with specific structural properties, that we will describe in more detail later.

We note that if $x$ is a solution to the phase retrieval problem, then $c x$
for any scalar $c \in \CC$ obeying $|c| = 1$ is also a solution. Thus,
without further information about $x$, all we can hope for is to recover
$x$ up to a {\em global phase}.  Thus, when we talk in this paper about exact recovery of $x$, we always mean recovery up to this global phase factor.

As mentioned before, the phase retrieval problem is notoriously ill-posed in its most classical form, where one tries to recover $x$
from intensities of its Fourier transform, $|\hat{x}|^2$. We will discuss questions about uniqueness in Section~\ref{s:uniqueness}, see also the reviews~\cite{grohs2019,jaganathan2015phase,bendory2017fourier}. To combat his ill-posedness, we have the options to include additional prior information about $x$
or acquire additional measurements about $x$, or a combination of both.  We will briefly the most common strategies below.

\subsection{Prior information}

A natural way to attack the ill-posedness of phase retrieval is to reduce the number of unknown parameters. The most common assumption is to invoke {\em support constraints} on the signal~\cite{Fie82,CMW07}. This is often justified since the object of interest may have clearly defined boundaries, outside of which one can assume that the signal is zero. The effectiveness of this constraint often hinges on the accuracy on the estimated support boundaries.
{\em Positivity} and {\em real-valuedness} are other frequent assumptions suitable in many settings, while {\em atomicity} is more limited to specific scenarios~\cite{Fie78,Fie82,Mar07,CMW07}. Another assumption that has gained popularity in recent years is {\em sparsity}~\cite{shechtman2015phase}. Under the sparsity assumption, the signal of interest has only relatively few non-zero coefficients in some (known) basis, but we do not know a priori the indices of these coefficients, thus we do not know the location of the support. This can be seen as a generalization of the usual support constraint. 

Oversampling in the Fourier domain has been proposed as a means to
mitigate the non-uniqueness of the phase retrieval problem in connection with prior signal information~\cite{MCS97}. 
While oversampling offers no benefit for most one-dimensional signals, the
situation is more favorable for multidimensional signals, where it has
been shown that twofold oversampling in each dimension almost always
yields uniqueness for finitely supported, real-valued and non-negative
signals~\cite{BS79,Hay82,San85}, see also~\cite{grohs2019}. 
As pointed out in \cite{LBL02}, these uniqueness results do not say
anything about how a signal can be recovered from its intensity
measurements, or about the robustness and stability of commonly used
reconstruction algorithms. We will discuss throughout the paper how to incorporate various kinds of prior information in the algorithm design.

\subsection{Measurement techniques}\label{s:measurements}

The  setup of classical X-ray crystallography (aside of oversampling) corresponds to the most basic measurement setup where the measurement vectors $a_k$ are the columns of the associated 2-D DFT matrix. This means if $x$ is an $n \times n$ image, we obtain $n^2$ Fourier-intensity samples, which is obviously not enough to recover $x$. Thus, besides oversampling, different strategies have been devised to obtain additional measurements about $x$. We briefly review these strategies and discuss many of them in more detail throughout the paper.

\subsubsection{Coded diffraction imaging}
The combination of X-ray diffraction, oversampling and phase retrieval has launched the field of
{\em Coherent Diffraction Imaging} or CDI~\cite{miao1999extending,Mar07}. A detailed description of CDI and phase retrieval can be found
in~\cite{shechtman2015phase}.  
As pointed out in~\cite{shechtman2015phase}, the lensless nature of CDI is actually an advantage when dealing with extremely intense
and destructive pulses, where one can only carry out a single pulse measurement with each object (say, a molecule) before
the object disintegrates. Lensless imaging is mainly used in short wavelength spectral regions such as extreme ultraviolet
(EUV) and X-ray, where high precision imaging optics are difficult to manufacture, expensive and
experience high losses. We discuss CDI in more detail in Section~\ref{ss:coded}, as well as throughout the paper.

\subsubsection{Multiple structured illuminations}  \label{ss:diffractions}

A by now very popular approach to increase the number of measurements is to collect
several diffraction patterns providing ``different views'' of the
sample or specimen, as illustrated in Figure~\ref{fig:mask}. The concept of using multiple measurements as an attempt to resolve the phase ambiguity for diffraction imaging is of course not new,
and was suggested in \cite{Mis73}. Since then, a variety of methods have been proposed to carry out these multiple measurements;
depending on the particular application, these may include the use of various gratings and/or of masks, the rotation of the axial
position of the sample, and the use of defocusing implemented in a spatial light modulator, see \cite{Dua11} for details and references. 

Inspired by work from compressive sensing and coded diffraction imaging, theoretical analysis clearly revealed the potential of combining randomness with multiple illuminations~\cite{CSV2013,unique}.
Despite the sometimes expressed skepticism towards the feasibility of random illuminations~\cite{luke2017phase}, this concept has a long history in optics and X-ray imaging, and great progress continues to be made~\cite{ptycho-rpi}, \cite{horisaki2016single},\cite{random-aperture}, \cite{seaberg2015coherent}, \cite{zhang2016phase}, \cite{marchesini2019shaping}, thereby exemplifying the exciting advanced that can be achieved by an efficient feedback loop between theory and practice. To quote from the source~\cite{marchesini2019shaping}: ``The ability to arbitrarily shape coherent x-ray wavefronts at new synchrotron and x-ray free electron facilities with these new optics will lead to advances in measurement capabilities and techniques that have been difficult to implement in the x-ray regime.''

\begin{figure}
\begin{center}
\includegraphics[width=100mm]{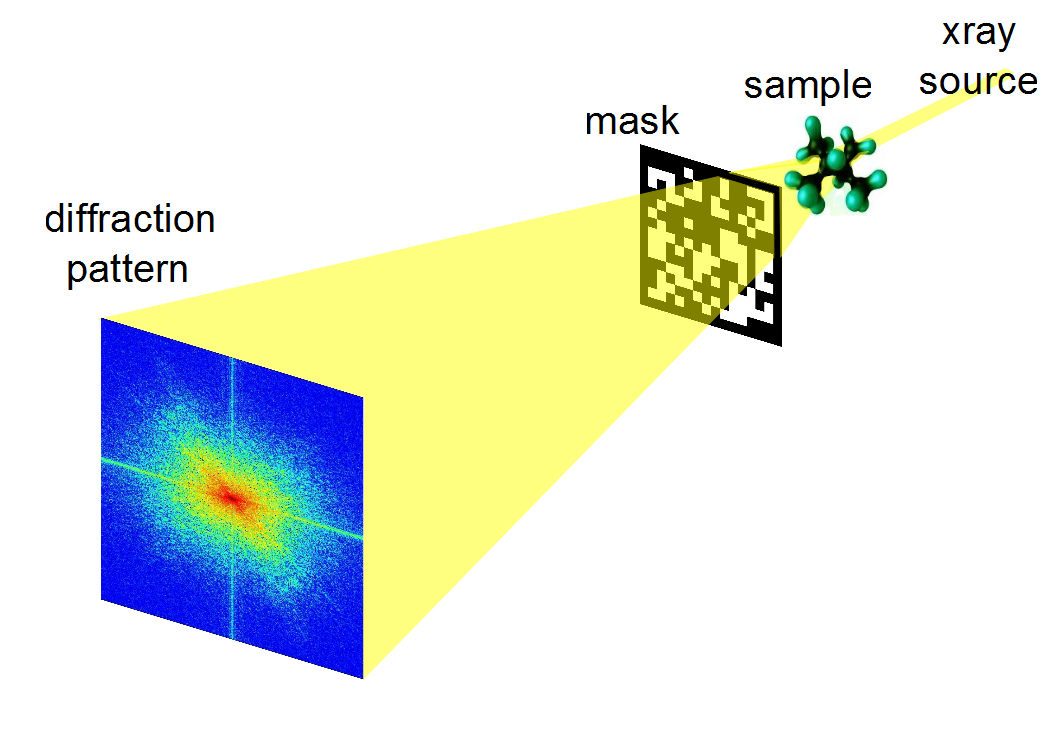}
\caption{A typical setup for structured illuminations in diffraction imaging
using a phase mask.}
\label{fig:mask}
\end{center}
\end{figure}

We can create  multiple illuminations in many ways. One possibility is to modify the phase front
after the sample by inserting a {\em mask} or a phase plate, see~\cite{Liu08} for example. A schematic layout is shown in Figure~\ref{fig:mask}.  
Another standard approach would be to change the profile or modulate the illuminating beam, which can easily be
accomplished by the use of {\em optical gratings}~\cite{LP97}. A simplified  representation would look similar to the scheme depicted in Figure~\ref{fig:mask}, with a grating instead of the mask (the grating could be placed before or after the sample).

{\em Ptychography} can be seen as an example of multiple illuminations.  But due to its specific structure, ptychography deserves to be treated separately. In ptychography, one records several diffraction patterns from overlapping areas of the sample, see \cite{Rod08,TDB09} and
references therein. We discuss ptychography in more detail in Section~\ref{ss:ptycho} and Section~\ref{s:ptycho}.
In \cite{Pfeiffer}, the sample is scanned by shifting the phase plate as in ptychography; the
difference is that one scans the known phase plate rather than the object being imaged.
{\em Oblique illuminations}  are another possibility to create multiple illuminations. Here one can use illuminating beams hitting the sample at user specified angle \cite{FHP10}.

In mathematical terms, the phase retrieval problem when using multiple structured illuminations
in the measurement process, can be expressed as follows.
\beqn
    &\text{Find}   & \quad x\\
    &\text{subject to} & \quad  y_{k,\ell} = |(FD_\ell x)_k|^2, \quad  k=1,\dots,n; \ell =1,\dots,L,
\eeqn
where $D_\ell$ is a diagonal matrix representing the $\ell$-th mask out of a total of $L$ different masks, and the total number of measurements is
given by $N=nL$.

\subsubsection{Holography}

Holographic techniques, going back to the seminal work of Dennis Gabor~\cite{gabor1948new}, are among the more popular methods that have
been proposed to measure the phase of the optical wave. The basic idea of holography is to include a reference in the illumination process. This prior information can be utilized to recover the phase of the signal.  While holographic techniques have been successfully applied in certain areas of optical imaging, they can be generally difficult to implement in
practice~\cite{Dua11}. In recent years we have seen significant  progress in this area~\cite{saliba2016novel,latychevskaia2012holography}.
We postpone a more detailed discussion of holographic methods to Section~\ref{s:holo}.

\subsection{Measurement of coded diffraction patterns}\label{ss:coded}

Due to the importance of coded diffraction patterns for phase retrieval, we describe this scheme  in more detail.
Let $\IZ_n^2=\lb 0,n-1\rb^2$ be the object domain containing the support of the discrete object $\x$ where $\lb k, l\rb$ denotes the integers between, and including,  $k\le l\in \IZ$.

 For any  vector $u$, define  its modulus vector $|u|$ as
 $
 |u|(j)= |u(j)|
 $
 and its phase vector $\sgn(u)$ as 
\[
\sgn(u)(j)=\lt\{\begin{matrix} 
e^{\im \alpha}&\mbox{if $u(j)=0$}\\
u(j)/|u(j)|&\mbox{else. }
\end{matrix}\rt.
\]
where $j$ is the index for the vector component. The choice of $\alpha\in \IR$ is arbitrary when $u(j)$ vanishes.  
However, for numerical implementation, $\alpha$ can be conveniently set to $0$.

In the noiseless case phase retrieval problem is to solve 
\beq
\label{pr}
b=|u|\quad \mbox{ with}\quad  u=A\x
\eeq
for unknown object $\x$ with given data $b$ and 
some measurement matrix $A$. 

 Let $\x(\bn), \bn = (n_1,n_2,\cdots,n_d)\in \IZ^d,$ be a discrete  object function   supported in 
\[
\cM = \{ 0\le m_1\le M_1, 0\le m_2\le M_2,\cdots, 0\leq m_d\leq M_d\}. 
\]

Define the $d$-dimensional {\em discrete-space Fourier transform} of $\x$ as
\beqn
\sum_{\bn\in \cM} \x(\bn)
   e^{-  2\pi\im \bn\cdot \bom},\quad \bom=(w_1,\cdots,w_d)\in [0,1]^d.   \eeqn
However, only the {\em intensities} of the Fourier transform, called the diffraction pattern,
are measured  
 \beq
   \sum_{\bn =-\bM}^{\bM}\sum_{\mbm\in \cM} \x(\mbm+\bn)\overline{\x(\mbm)}
   e^{-\im 2\pi \bn\cdot \bom},\quad \bM = (M_1,\cdots,M_d)\nn
   \eeq
   which is the Fourier transform of the autocorrelation
   \beqn
	  R(\bn)=\sum_{\mbm\in \cM} \x(\mbm+\bn)\overline{\x(\mbm)}.
	  \eeqn
Here and below the over-line  means
complex conjugacy. 

Note that
$R$ is defined on the enlarged  grid
 \begin{eqnarray*}
 \widetilde \cM = \{ (m_1,\cdots, m_d)\in \IZ^d: -M_1 \le m_1 \le M_1,\cdots, -M_d\le m_d\leq M_d \} 
 \end{eqnarray*}
whose cardinality is roughly $2^d$ times that of $\cM$.
Hence by sampling  the diffraction pattern
 on the grid 
\beqn
\cL = \Big\{(w_1,\cdots,w_d)\ | \ w_j = 0,\frac{1}{2 M_j + 1},\frac{2}{2M_j + 1},\cdots,\frac{2M_j}{2M_j + 1}\Big\}
\eeqn
we can recover the autocorrelation function by the inverse Fourier transform. This is the {\em standard oversampling} with which  the diffraction pattern and the autocorrelation function become equivalent via the Fourier transform.

A coded diffraction pattern is measured with a mask
whose effect is multiplicative and results in  
a {\em masked object}  $
\x(\bn) \mu(\bn)$ 
where $\mu(\bn)$ is an array of random variables representing the mask.   
In other words, a coded diffraction pattern is just the plain diffraction pattern of
a masked object. 

We will focus on the effect of {\em random phases} $\phi(\bn)$ in the mask function 
$
\mu(\bn)=|\mu|(\bn)e^{\im \phi(\bn)}
$
where  $\phi(\bn)$ are independent, continuous real-valued random variables and $|\mu|(\bn)\neq 0,\forall \bn\in \cM$ (i.e. the mask is transparent). 
The mask function by assumption is a finite set of  continuous random variables
and so is $\y=A \x$.  Therefore  $\y$ vanishes  nowhere almost surely, i.e.
\[
b_{\rm min}= \min_{j} b_j>0.
\]

For simplicity we assume $|\mu|(\bn)=1,\forall\bn$ which gives rise to
a {\em phase} mask and an {\em isometric}  propagation matrix 
\beq
\label{one}
\hbox{\rm (1-mask )}\quad A= c\Phi\,\, \diag\{\mu\},
\eeq
i.e. $A^*A=I$ (with a proper choice of the normalizing constant  $c$), where $\Phi$ is the {\em oversampled}  $d$-dimensional discrete Fourier transform (DFT). Specifically  $\Phi \in \IC^{|\tilde \cM|\times |\cM|}$ is the sub-column matrix of
the standard DFT on  the extended grid $\tilde \cM$ where $|\cM|$ is
the cardinality of $\cM$.

If the non-vanishing mask $\mu$ does not have a uniform transparency, i.e. $|\mu|(\bn)\neq 1, \forall \bn,$ then we can define  a new object vector $|\mu|\odot \x$ and a new
isometric propagation matrix
\[
A= c\Phi\,\, \diag\lt\{{\mu\over |\mu|}\rt\}
\]
with which to recover the new object first.

When two phase masks $\mu_1, \mu_2$ are deployed, 
the propagation matrix $A^*$ is the stacked coded DFTs, i.e.   
\beq \label{two}\hbox{(2-mask case)}\quad 
A=c \lt[\begin{matrix}
\Phi\,\, \diag\{\mu_1\}\\
\Phi\,\, \diag\{\mu_2\}
\end{matrix}\rt]. 
\eeq
 With proper normalization, $A$ is
isometric. 

All of the results with coded diffraction patterns present in this work apply to  $d\geq 2$. But the most relevant case is $d=2$
which is assumed hereafter. 
We can vectorize the object/masks by converting the $n\times n$ square grid into an ordered set of index. Let $N$ the total number of measured data. In other words, $A\in \IC^{N\times n^2}$ where  $N$ is about $4\times n^2$ and and $8\times n^2$, respectively, in the case of \eqref{one} and \eqref{two}.

\subsection{Ptychography}\label{s:ptycho}

Ptychography is a special case of coherent diffractive imaging that uses multiple
micro-diffraction patterns obtained through scanning across the unknown specimen
with a mask, making a  measurement for each location   via a localized illumination on the specimen~\cite{hoppe1969beugung,Rod08}.
This provides a much larger set of measurements, but at the cost of a longer, more involved experiment.
As such ptychography is a synthetic aperture technique and, along with advances in
detection and computation techniques, has enabled microscopies with enhanced resolution
and robustness without the need for lenses. Ptychography offers numerous benefits and thus attracted significant attention.
See~\cite{dierolf2010ptychographic,TDB09,Rod08,qian2014efficient,pfeiffer2018x,horstmeyer2016diffraction} for a small sample of different activities in this field.

A schematic drawing of a ptychography experiment in which a probe scans through a 2D object in an overlapping fashion
and produces a sequence of diffraction patterns of the scanned regions is depicted in Figure~\ref{fig:ptychosetup}.
Each image frame represents the magnitude of the Fourier transform
of ${\mu}(\bs) x(\bs + \bt)$, where ${\mu}(\bs)$ is a localized illumination (window) function or
a mask,  $x(\bs)$ is the unknown object of interest, and $t$ is a translational vector. Thus
the  measurements taken in ptychography can be expressed as
\begin{equation} \label{ptycho1}
|  F({\mu}(\bs) x(\bs+\bt)  |^2.
\end{equation}
Due to its specific underlying mathematical structure, ptychography deserves its own analysis.
A  detailed discussion of various reconstruction algorithms for ptychography can be found in~\cite{qian2014efficient}. For a convex approach using the PhaseLift idea see for instance~\cite{horstmeyer2015solving}. An intriguing algorithm that combines ideas from PhaseLift with the local nature of the measurements can be found in~\cite{iwen2016phase}.

\begin{figure}
\begin{center}
\includegraphics[width=100mm]{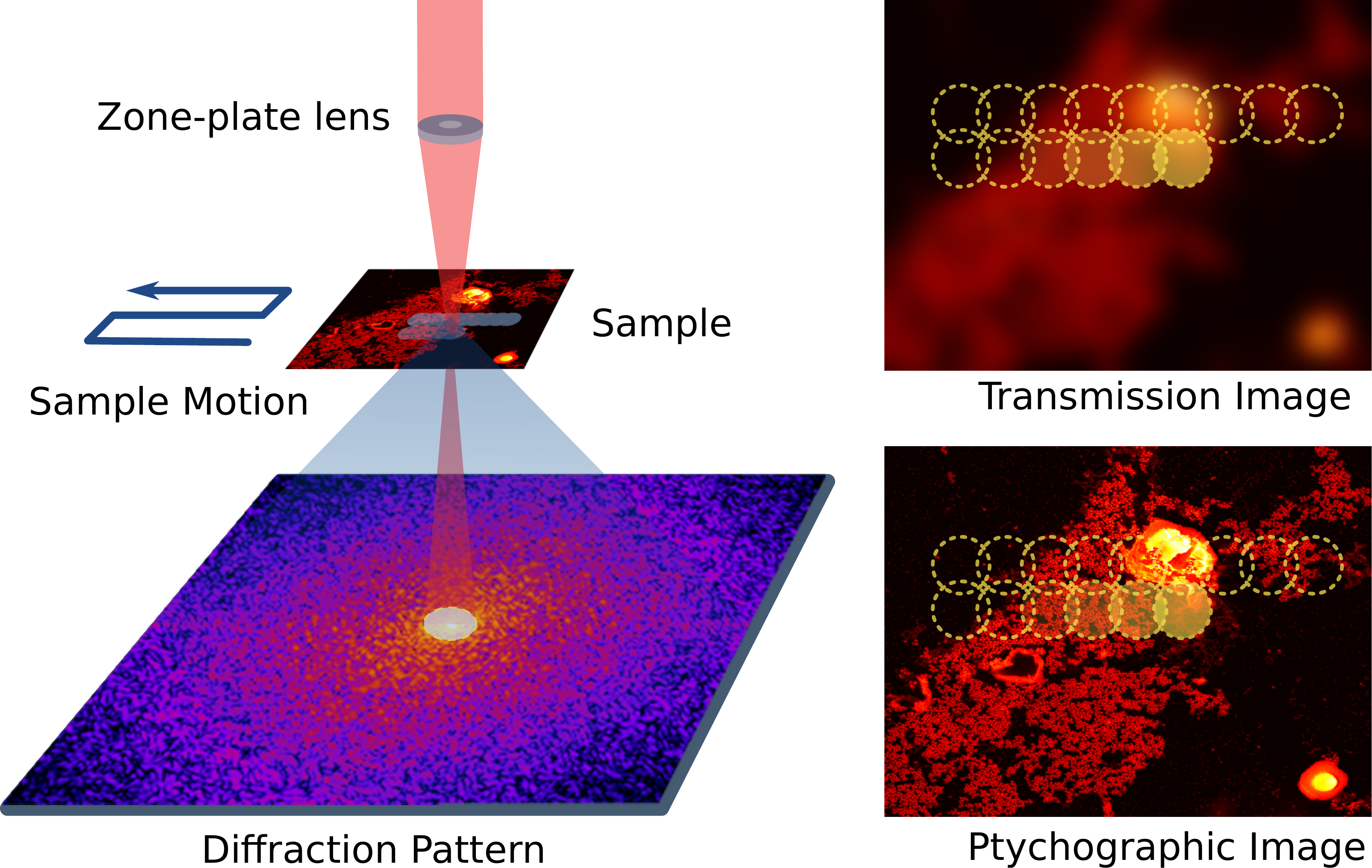}
\caption{A schematic drawing of a ptychography experiment in which a probe scans through a 2D object in an overlapping fashion
and produces a sequence of diffraction patterns of the scanned regions.   Image courtesy of~\protect\cite{qian2014efficient}. }
\label{fig:ptychosetup}
\end{center}
\end{figure}

\subsection{Ptychography and time-frequency analysis}

An inspection of the basic measurement mechanism of ptychography in~\eqref{ptycho1}  shows an interesting connection to time-frequency analysis~\cite{Grochenig2001}.  To see this, we recall the definition of the  {\em short-time Fourier transform} (STFT) and the
{\em Gabor transform}.  For $\bs,\omega \in \IR^d$ we define the {\em translation operator} $T_\bs$ and
the {\em modulation operator} $M_{\omega}$ by
\beqn
T_\bs {x}(\bt) = {x}(\bt-\bs), \qquad M_{\bomg} {x}(\bt) = e^{2\pi \im \bomg \cdot \bt}{x}(\bt),
\eeqn
{where  $x \in L^2(\IR^d)$. Let $\mu \in {\cS}(\IR^d)$}, where ${\cS}$
denotes the Schwartz space. The STFT  of ${x}$ with respect to the {\em window} ${\mu}$ is defined by
\begin{equation}
\STFT_{{\mu}}{x}(\bs,\bomg) = \int\limits_{\IR^{d}} {x}(\bt) {\mu}(\bs-\bt)  
e^{-2\pi \im \bomg \cdot \bt}dt = \langle {x}, M_\omega T_\bs {\mu} \rangle , \qquad (\bs,\bomg) \in \IR^{2d}.\nn
\end{equation}
A Gabor system consists of functions of the form
\begin{equation}
 e^{2\pi \im b\bll  \bt} {\mu}(\bt- a\bk ) = M_{b\bll } T_{a \bkk } {\mu} ,\,\,\, (\bkk,\bll) \in \ZZ^d \times \ZZ^d\nn
\end{equation}
where  $a, b >0$ are the time- and frequency-shift parameters~\cite{Grochenig2001}.  The associated Gabor transform
${G}: L^2(\IR) \mapsto \ell^2(\ZZ\times \ZZ)$ is defined as
\begin{equation}
\nn
{Gx} = \{\langle {x},  M_{b \bll } T_{a \bkk } {\mu}  \rangle\}_{(\bkk,\bll) \in \ZZ^d \times \ZZ^d}.
\end{equation}
${G}$ is  clearly just an STFT that has been sampled at the time-frequency lattice $a\ZZ \times b\ZZ$.
It is clear that the definitions of the STFT and Gabor transform above can be adapted in an obvious way for discrete or finite-dimensional functions.

Since ptychographic measurements take the form $\{| \langle {x}, M_{\bomg} T_{\bs} {\mu} \rangle |^2\}$ where $(\bs,\bomg)$ are indices of some time-frequency lattice, it is now evident that these measurements simply correspond to squared magnitudes of the STFT or (depending on the chosen time-frequency shift parameters) of the Gabor transform of the signal $x$ with respect to the mask ${\mu}$.
Thus, methods developed for the reconstruction of a function from magnitudes of its (sampled) STFT---see e.g.~\cite{eldar2014sparse,pfander2019robust,grohs2019}---become also relevant for ptychography.

Beyond ptychography, phase retrieval from the STFT magnitude has been used in in speech and audio processing~\cite{nawab1983signal,balan2010signal}.  It has also have found extensive applications in optics. As described in~\cite{jaganathan2015phase}, one example arises in frequency resolved optical gating (FROG) or XFROG, which is used for characterizing 
ultra-short laser pulses by optically producing the STFT magnitude of  the measured pulse.

\subsection{2D Ptychography} \label{ss:ptycho}

While the mathematical framework of ptychography can be formulated in any dimension, the two-dimensional case is the most relevant in practice.
In the ptychographic measurement, the $m\times m$ mask has a smaller size  than the $n\times n$ object, i.e. $m< n$, and is shifted around
to various positions for measurement of coded diffraction patterns so as to cover  the entire object. 

Let $\cM^{0}:=\IZ_m^2, m<n,$ be the initial mask area, i.e.  the support of the mask $\mu^{0}$ describing the illumination field. 
Let $\cT$ be the set of all shifts (i.e. the scan pattern), including $(0,0)$,  involved  in the ptychographic measurement. 
 Denote by $\mu^\bt$ the $\bt$-shifted mask for all $\bt\in \cT$ and $\cM^\bt$ the domain of
$\mu^\bt$. Let $\x^\bt$ the object restricted to $\cM^\bt$.
We  refer to each $\x^\bt$ as a part of $\x$ and write $\x=\vee_\bt \x^\bt$ {  where $\vee$ is the ``union" of functions consistent over their common support set}. In ptychography, the original object is broken up into a set of overlapping object parts, each of which produces a $\mu^\bt$-coded diffraction pattern.  
The totality of the coded diffraction patterns is called the ptychographic measurement data.   For convenience of analysis, we assume the value zero for $\mu^\bt, \x^\bt$ outside of $\cM^\bt$
 and the periodic boundary condition on $\IZ_n^2$ when $\mu^\bt$ crosses over the boundary of $\IZ_n^2$.

A basic scanning pattern is the 2D lattice with the basis $\{\bv_1, \bv_2\}$
\beq
\cT=\{\bt_{kl}\equiv k\bv_1+l\bv_2: k,l\in \IZ\},\quad \bv_1,\bv_2\in \IZ^2\nn
\eeq
acting on the object domain $\IZ_n^2$. 
Instead of $\bv_1$ and $\bv_2$ we can also take $ \bu_1=\ell_{11}\bv_1+\ell_{12}\bv_2$ and $\bu_2=\ell_{21} \bv_1+\ell_{22}\bv_2$ for integers $\ell_{ij}$ with $\ell_{11}\ell_{22}-\ell_{12}\ell_{21}=\pm 1$. This ensures that $\bv_1$ and $\bv_2$ themselves are integer linear combinations of $\bu_1,\bu_2$. Every lattice basis defines a fundamental parallelogram, which determines the lattice. 
There are five 2D lattice types, called period lattices, as given by the crystallographic restriction theorem. In contrast,
there are 14 lattice types in 3D, called Bravais lattices. 

Under  the periodic boundary condition the raster scan with the step size $\tau=n/q, q\in \IN,$ $\cT$ consists of $\bt_{kl}={\tau}(k,l)$, with $k,l\in \{0,1,\cdots, q-1\}$ (Figure \ref{fig:graph}(a)). The periodic boundary condition means that for $k=q-1$ or $l=q-1$  the shifted mask is wrapped around into the other end of the object domain.

 \begin{figure}[t]
\begin{center}
\subfigure[raster scan]{\includegraphics[width=3.5cm]{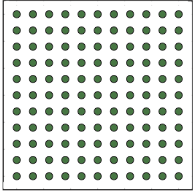}}\hspace{1cm}
\subfigure[]{\includegraphics[width=3cm]{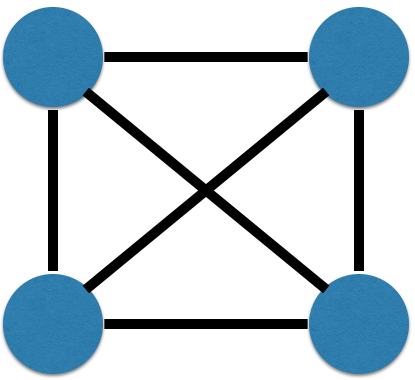}}\hspace{1cm}
\subfigure[]{\includegraphics[width=3cm]{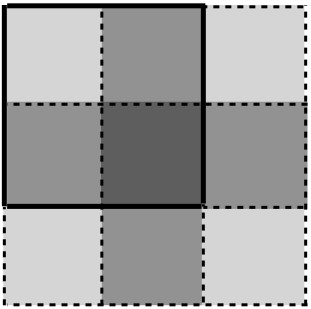}}
\caption{A complete undirected graph (a) representing four connected object parts (b) where the grey level indicates the number of coverages by the mask in four scan positions.
}
\label{fig:graph}
\end{center}
\end{figure}

A basic requirement is the strong connectivity property of the object with respect to the measurement scheme. 
It is useful to think of connectivity in graph-theoretical terms: Let the ptychographic experiment be represented by
a complete graph $\cG$ whose notes correspond to $\{\x^\bt:\bt\in \cT\}$ (see Figure \ref{fig:graph}(b)).

An edge between two nodes  corresponding to $\x^\bt$ and $\x^{\bt'}$ is $s$-connective if 
\beq
\label{r100}\label{s-conn}
|\cM^\bt\cap \cM^{\bt'}\cap\supp(\x)|\ge s\ge 2
\eeq
where $|\cdot|$ denotes the cardinality. In the case of full support (i.e. $\supp(\x)=\cM$),  \eqref{r100} becomes
$|\cM^\bt\cap \cM^{\bt'}|\ge s$.
An $s$-connective sub-graph $\cG_s$ of $\cG$ consists
of all the nodes of $\cG$ but only the $s$-connective edges.  Two nodes are adjacent (and neighbors) in $\cG_s$ iff they are $s$-connected. A chain in $\cG_s$ is a sequence of nodes  such that 
two successive nodes are adjacent. In a simple chain all the nodes are distinct. Then the object parts $\{ \x^\bt:\bt\in \cT\}$ are  $s$-connected if and only if $\cG_s$ is a connected graph, i.e. every two nodes is connected by a chain of $s$-connective edges. Loosely speaking, an object is strongly connected w.r.t. the ptychographic scheme if $s\gg 1$. 
We say that $\{ \x^\bt: \bt\in \cT\}$ are $s$-connected if there is an $s$-connected chain between any
two elements. 

Let us consider the simplest raster scan corresponding to  the  {\em square lattice} with $\bv_1=(\tau,0), \bv_2=(0,\tau)$ of  step size $\tau>0$,
i.e. 
\beq
\label{raster}
\bt_{kl}=\tau (k,l),\quad k,l=0,\dots, q-1. 
\eeq
For even coverage of the object, we assume that $\tau=n/q=m/p$ for some $p< q\in \IN$. 

  Denote the $\bt_{kl}$-shifted  masks and blocks by $\mu^{kl}$ and $\cM^{kl}$, respectively. Likewise, denote
by $\x^{kl}$ the object restricted to the shifted domain $\cM^{kl}$. 

\begin{figure}[t]
\begin{center}
\subfigure[Matrix $A_{\nu}$]{
\includegraphics[width=7cm]{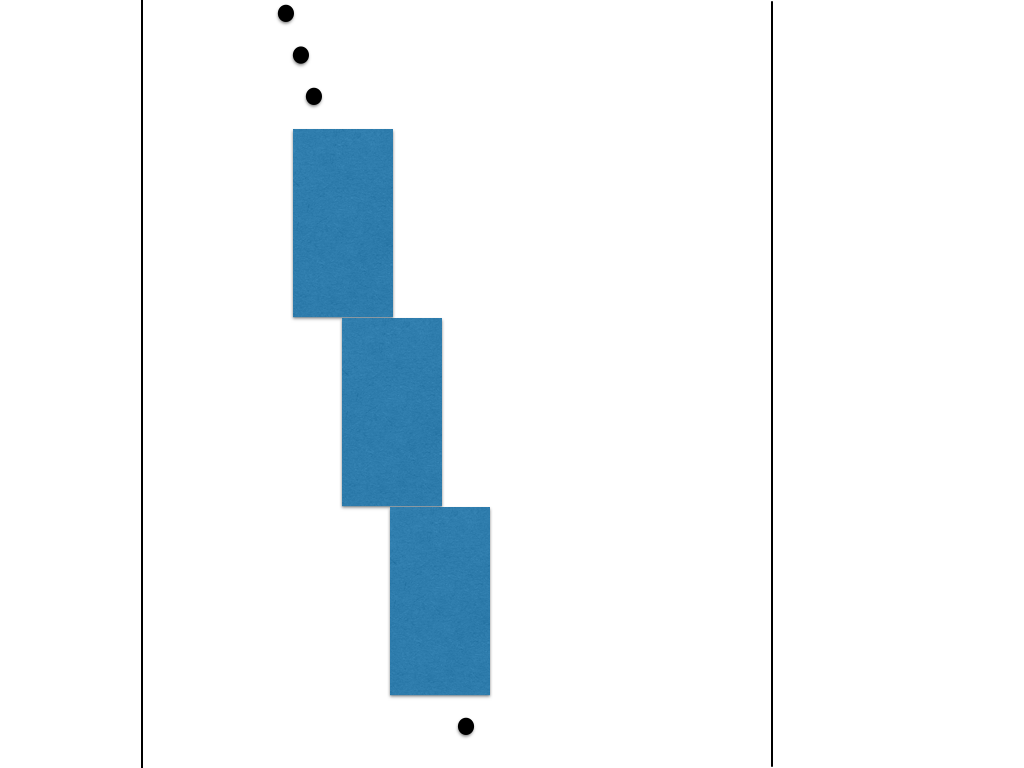}
}\hspace{-2cm}
\subfigure[Matrix $B_{\xh}$]{
\includegraphics[width=7cm]{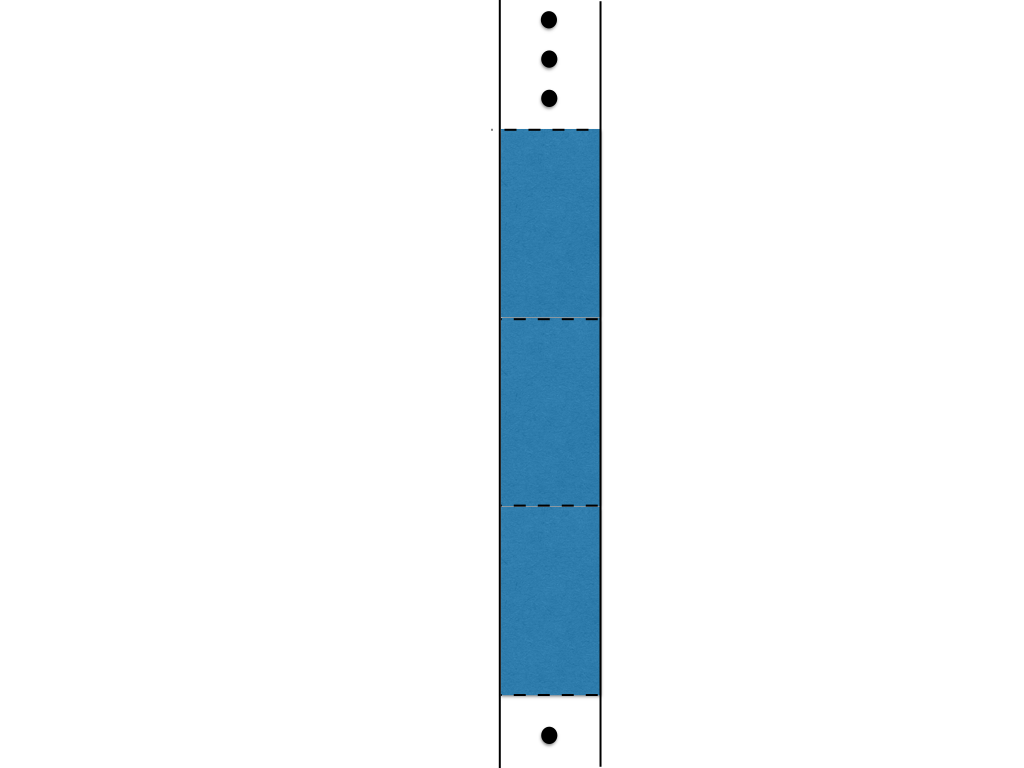}
}
\caption{(a) $A_{\nu}$ is a concatenation of shifted
blocks $\{\Phi\,\diag({\nu}^\bt):\,\,\bt\in \cT\}$; (b) $B_{\xh}$  is a concatenation of unshifted 
blocks $\{\Phi\,\diag( x^\bt):\,\,\bt\in \cT\}$.   In both cases, each block gives rise to
a coded diffraction pattern $|\Phi ({\nu}^\bt\odot  x^\bt)|$.
}
\label{fig10}
\end{center}
\end{figure}

Let $\cF(\nu,x)$ be the bilinear transformation representing the totality of the Fourier (magnitude and phase) data for any mask $\nu$ and object $x$.
From $\cF(\nu^0,x)$ we can define two measurement matrices. First, for a given $\nu^0\in \IC^{m^2}$, 
let $A_{\nu}$ be defined via the relation $A_{\nu} x:=\cF(\nu^0,x)$ for all $x\in \IC^{n^2}$; second, for a given $x\in \IC^{n^2}$,
let $B_{x}$ be defined via $B_{x} {\nu}=\cF(\nu^0,x)$ for all $\nu^0\in \IC^{m^2}$. 

More specifically, let $\Phi$ denote the over-sampled Fourier matrix. The measurement matrix $A_{\nu}$ is a concatenation of $\{\Phi \,\diag(\nu^\bt):\bt \in \cT\}$ (Figure \eqref{fig10}(a)). Likewise,
$B_{x}$ is $\{\Phi \,\diag(x^\bt):\bt \in \cT\}$ stacked on top of each other (Figure \eqref{fig10}(b)). 
Since
$\Phi$ has orthogonal columns, both $A_{\nu}$ and $B_{x}$ have orthogonal columns.
Both matrices will be relevant when we discuss blind ptychography which does not assume the prior knowledge of the mask 
in Section \ref{sec:blind-ptycho}.

\section{Uniqueness, ambiguities, noise}\label{s:uniqueness}

In this section we discuss various questions of uniqueness and feasibility related the phase retrieval problem. Since a detailed and thorough current review of uniqueness and feasibility can be found in~\cite{grohs2019}, we mainly focus on aspects not covered in that review. We will also discuss various noise models.

\subsection{Uniqueness and ambiguities with coded diffraction patterns}\label{ss:uniqueness}

{ {\bf Line object:}   $\x$ is a {\em line object}  if   the original object support   is part of a line segment.}  Otherwise, $\x$ is said to be
a nonlinear object. 

Phase retrieval solution is unique only up to a constant of modulus one
 no matter how many coded diffraction patterns are measured.
 Thus the proper error metric for an estimate $ x$ of the true solution $\x$  is given by
  \beqn
 \min_{\theta\in \IR}\|e^{-i\theta} \x- x \|=\min_{\theta\in \IR}\|e^{i\theta}  x - \x \|
 \eeqn
 where the optimal phase adjustment $\theta_*$ is given by
 \beqn
 \theta_*=\measuredangle ( { x^*\x}). 
 \eeqn

Now we recall the uniqueness theorem of phase retrieval with coded diffraction patterns.

\begin{thm} \label{unique1} \cite{unique} 
 Let $\x\in \IC^{n^2}$ be a nonlinear object and  $x$ a solution of  of the  phase retrieval problem. Suppose that the phase of the random mask(s) is independent continuous random variables on $(-\pi,\pi]$.  

\noindent {\bf (i) One-pattern case.} Suppose, in addition, that $\measuredangle \x(j) \in [-\alpha\pi,\beta\pi],\,\forall j$ with $\alpha+\beta\in (0,2)$ and that  the density function for $\phi(\bn)$ is a constant (i.e. $(2\pi)^{-1}$) for every $\bn$.
 
Then   $x=e^{i\theta} \x$ for some constant $\theta\in (-\pi,\pi]$ with a high probability which has a simple, lower bound 
\beq
\label{prob}
1-n^2\lt|{\beta+\alpha\over 2}\rt|^{\llfloor S/2\rrfloor} 
\eeq
where $S$ is the number of nonzero components in $\x$  and $\llfloor S/2\rrfloor$ the greatest integer less than or equal to $S/2$. 

\noindent  {\bf (ii) Two-pattern case.}  Then   $x=e^{i\theta} \x$ for some constant $\theta\in (-\pi,\pi]$ with  probability one. 
\end{thm}
The proof of Theorem \ref{unique1} is given in \cite{unique} where
more general uniqueness theorems can be found.  
 It is noteworthy that
the probability bound for uniqueness \eqref{prob} improves exponentially  with higher sparsity of the object. 

We have the analogous uniqueness theorem for ptychography.

\begin{thm}\label{unique2}\cite{blind-ptycho}
 Let $\x\in \IC^{n^2}$ be a nonlinear object and  $x$ a solution of  of the  phase retrieval problem. Suppose that the phase of the random mask(s) is independent continuous random variables on $(-\pi,\pi]$.  

If the connectivity condition \eqref{s-conn} holds, then $\x$ is the unique ptychogaphic solution up to a constant phase factor.

\end{thm}

\subsection{Ambiguities with one diffraction  pattern}\label{sec:amb1}
 \begin{figure}[t]
   \centering
                \subfigure[Original object]{  \includegraphics[width = 4cm]{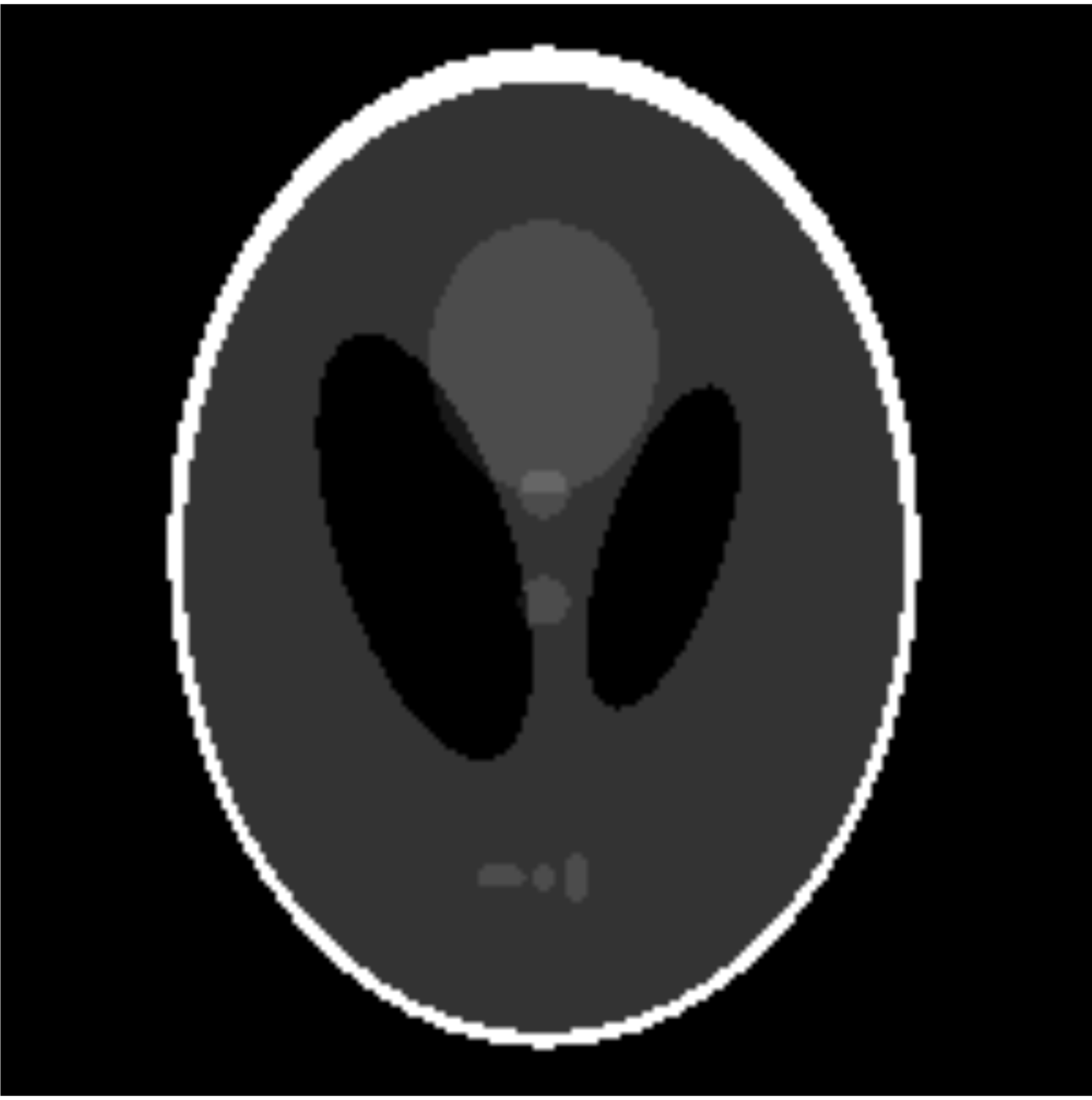}}
  \subfigure[AP]{
    \includegraphics[width = 4cm]{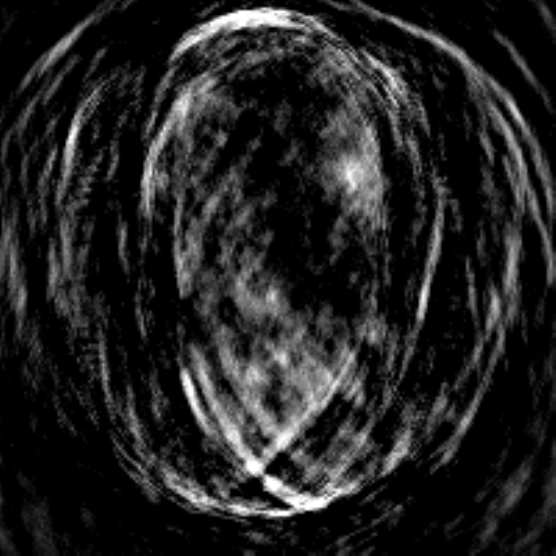}}
          \subfigure[AAR]{
    \includegraphics[width = 4cm]{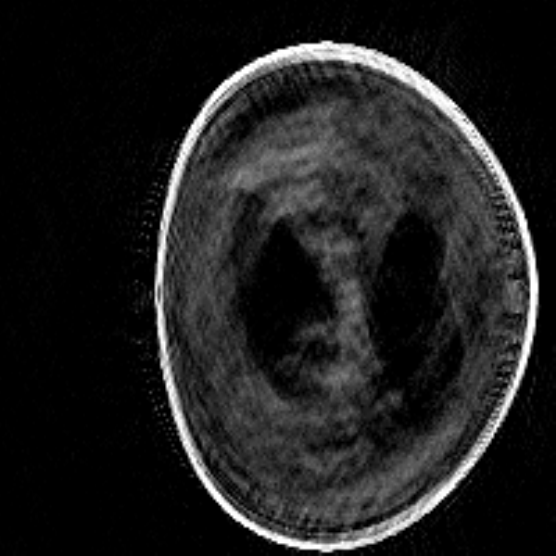}}
 \caption{
 (b) AP and (c) AAR reconstruction of the nonnegative real-valued phantom  with a plain uniform mask.} 
 \label{fig:ui}
\end{figure}

By the methods in \cite{unique}, it can be shown that 
an object estimate $\xh$ produces the same coded diffraction pattern as $\x$ 
if and only if 
\beq
\label{12}
\xh(\bn)&=&\lt\{\begin{matrix} e^{\im \theta} \x(\bn + \mbm)\mu(\bn + \mbm)/\mu(\bn) \\
e^{\im \theta} \overline{\x}(\bN-\bn + \mbm)\overline{\mu} (\bN-\bn + \mbm)/\mu(\bn),
\end{matrix}\rt.
\eeq
for  some $\mbm\in \IZ^2, \theta\in \IR$ {almost surely}. The ``if" part of the above statement is straightforward to check. {The ``only if" part 
is a useful result of using a random mask in measurement.}
Therefore, in addition to the trivial phase factor, there are translational (related to $\mbm$), conjugate-inversion (related to
$\overline{\x}(\bN-\cdot)$) as well as modulation ambiguity (related to $\mu(\bn+ \mbm)/\mu(\bn)$ or 
$\overline{\mu} (\bN+ \mbm-\bn)/\mu(\bn)$). {Among these, the conjugate-inversion (a.k.a. the twin image) is more prevalent as
it can not be eliminated by a tight support constraint. }

If, however, we have the prior knowledge that $\x$ is real-valued, then  none of the ambiguities in \eqref{12} can happen since
the right hand side of \eqref{12} has a nonzero imaginary part almost surely for any $\theta,\mbm$. 

On the other hand, if the mask is uniform (i.e. $\mu=$ constant), then \eqref{12} becomes
\beq
\label{12'}
\xh(\bn)&=&\lt\{\begin{matrix} e^{\im \theta} \x(\bn + \mbm) \\
e^{\im \theta} \overline{\x}(\bN-\bn + \mbm),
\end{matrix}\rt.
\eeq
for  some $\mbm\in \IZ^2, \theta\in \IR$. So even with the real-valued prior, all the ambiguities in \eqref{12'} are present, including
 translation, conjugate-inversion and constant phase factor. In addition, there may be other ambiguities not listed in \eqref{12'}. 
 
These ambiguities result in poor reconstruction as shown in Figure \ref{fig:ui} for the nonnegative real-valued phantom in Figure \ref{fig:ui}(a) with a plain, uniform mask by two widely used algorithms,
 Alternating Projections (AP) and Averaged Alternating Reflections (AAR), both of which are discussed in Section \ref{sec:solution}.

The phantom and its complex-valued variant, randomly phased phantom (RPP) used in Figure \ref{fig1} have the 
distinguished feature that their support is not the whole $n\times n$ grid but surrounded by an extensive area of dark pixels, thus making
the translation ambiguity in \eqref{12'} show up. This is particularly apparent in Figure \ref{fig:ui}(c). 
In general, when the unknown object has the full $n\times n$ support, phase retrieval becomes somewhat easier because 
the translation ambiguity is absent regardless of the mask used. 

\subsubsection{Twin-like ambiguity with a Fresnel mask} 

 \begin{figure}
\centering

\subfigure[$q=2$]{\includegraphics[width=4cm]{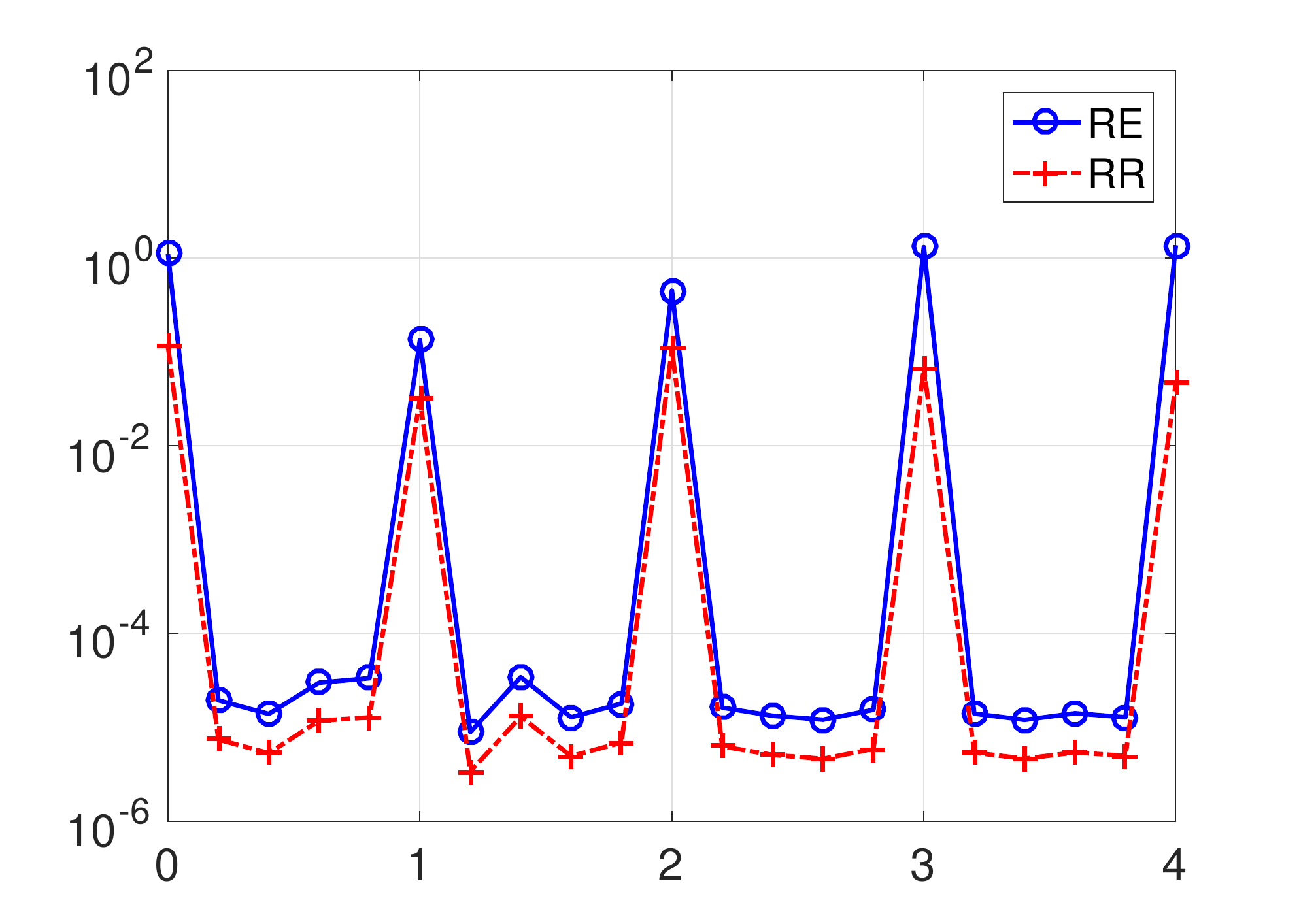}}
 \subfigure[$q=4$]{\includegraphics[width=4cm]{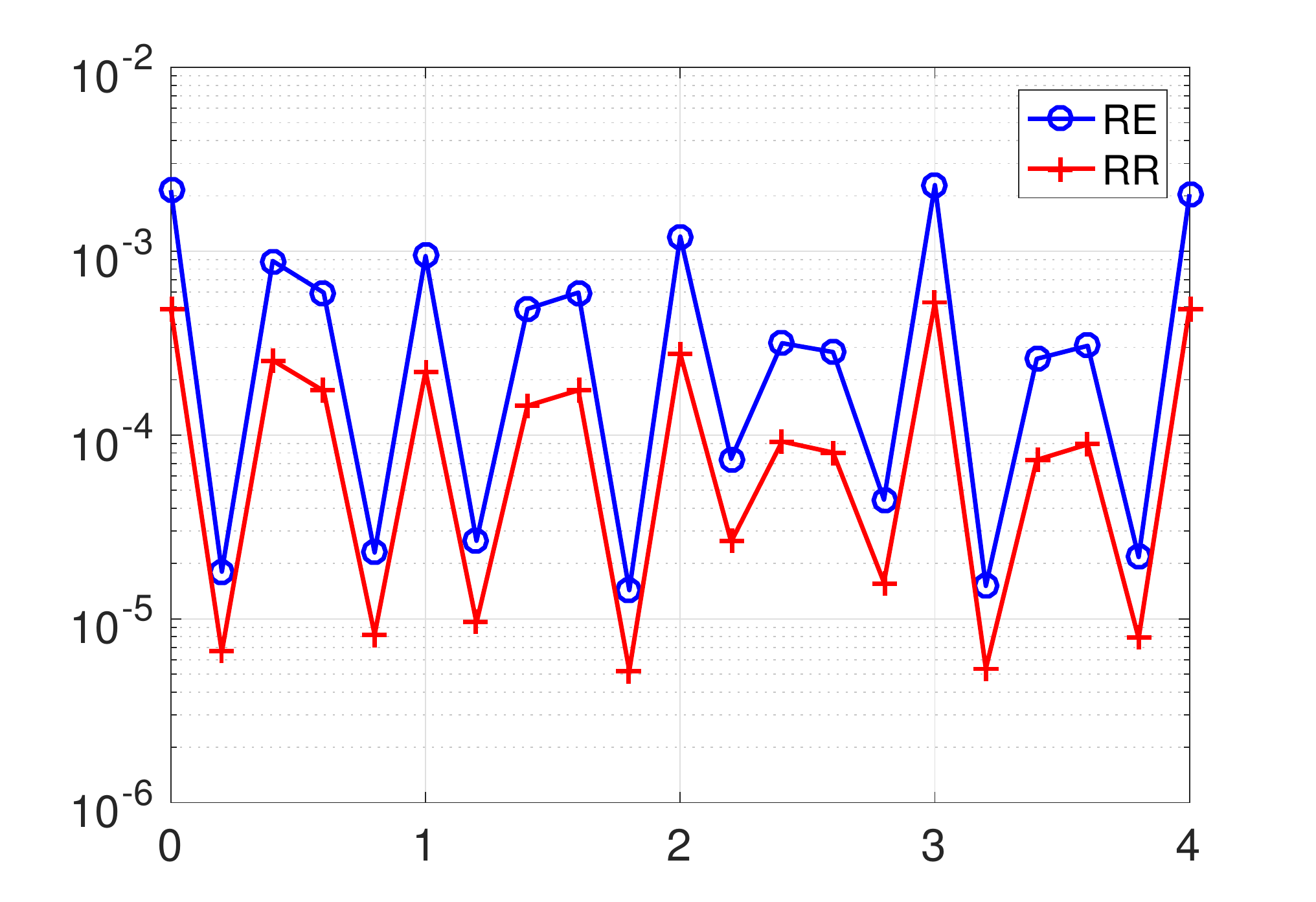}}
  \subfigure[$q=6$]{\includegraphics[width=4cm]{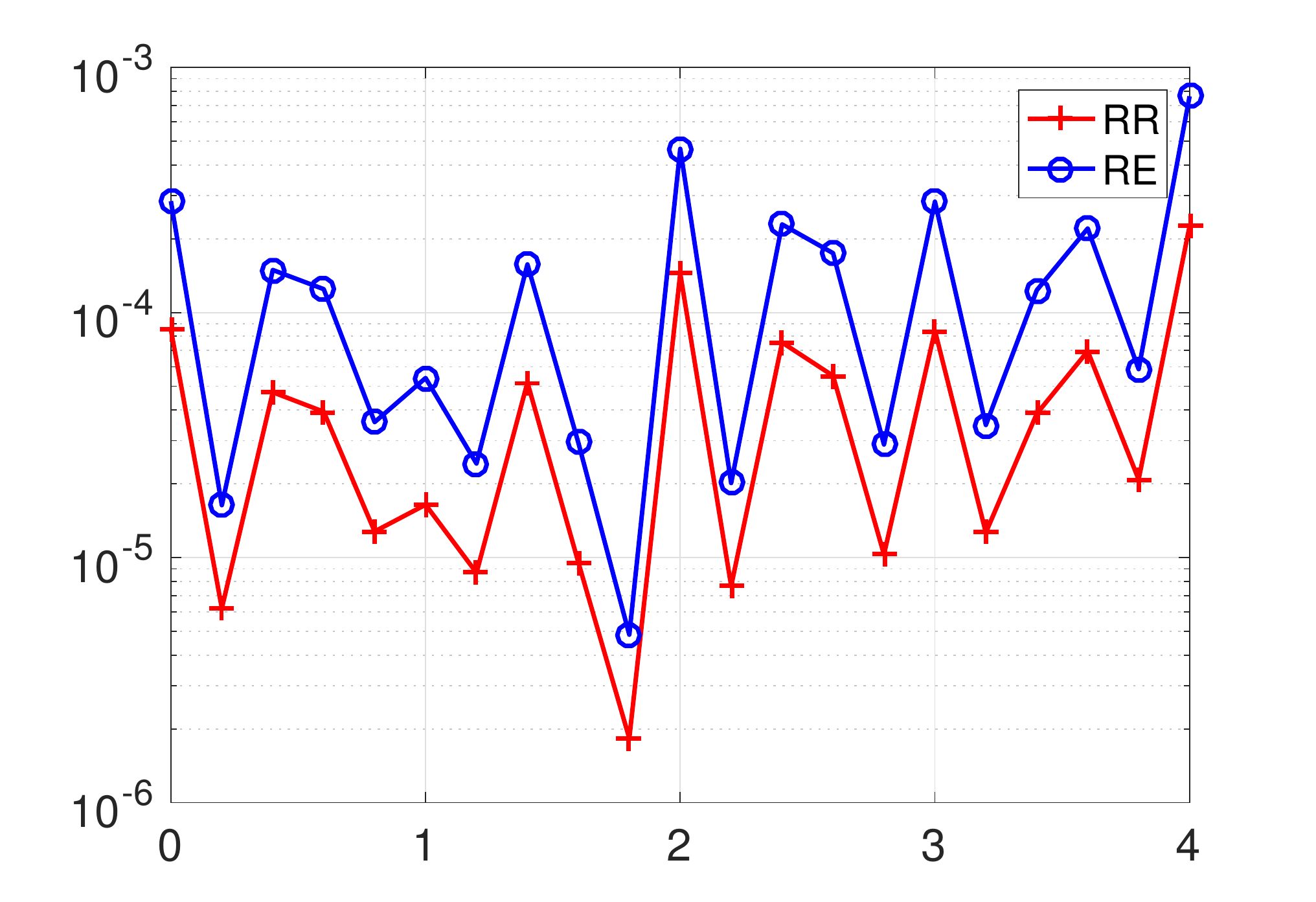}} 
  \caption{Relative error (RE) and relative residual (RR) on the semi-log scale  versus the parameter ${f}$ of the Fresnel mask for 
  the test object RPP.  }
  \label{fig1}
  \end{figure}

The next example shows that a commonly used mask can harbor a twin-like image as ambiguity
and the significance of using ``random" mask for phase retrieval. 

Consider
the Fresnel mask function which up to a shift  has the form 
\beq\label{par}
\mu^0(k_1,k_2):=\exp\lt\{\im \pi {f}(k_1^2+k_2^2)/m\rt\},\quad k_1, k_2=1,\cdots, m
\eeq
where ${f}\in \IR$ are adjustable parameters (see Figure \ref{fig:different-q}(c) for the phase pattern of \eqref{par}).

We construct a twin-like ambiguity for the Fresnel mask with $f\in \IZ$ and $q=2$. Similar twin-like ambiguities can be constructed for
general $q$. 

For constructing the twin-like ambiguity we shall write the object vector $\x$ as $n\times n$ matrix.
Let  ${\check{ \xi}}$ be the conjugate inversion of any $\xi\in \IC^{n\times n}$, i.e.
\beqn
{\check{\xi}}_{ij}=\overline{\xi}_{n+1-i,n+1-j}.
\eeqn
\begin{prop}\label{twin} \cite{DRAP-ptycho}
Let $f\in \IZ$ and  $\mu\in \IC^{m\times m}$ be the Fresnel mask \eqref{par}. For an even integer $n$, the matrix  \beq \nn
\overline{{\check{\mu}}}\odot \mu:=h=
\left(
\begin{array}{ccc}
  h_1&  h_2   \\
  h_3 &   h_4   
\end{array}
\right),\quad h_j\in \IC^{m/2\times m/2},\quad j=1,2,3,4,\eeq
satisfies the symmetry
\beq
h_1=h_4=\sigma h_2=\sigma h_3,\; \sigma=(-1)^{{f}(1+m/2)}.\nn
\eeq
Moreover, for $q=2$ (hence $m=n$ and $\tau=m/2$), then
$x={\check{x}_*} \odot \overline{h}$ and $\x$
 produce  the same ptychographic  data set  
with the Fresnel mask $\mu$. 
\end{prop}

To demonstrate the danger of using a ``regularly"  structured mask, we plot the relative error (RE) and relative residual (RR) of
reconstruction (200 AAR iterations followed by 100 AP iterations) in Figure \ref{fig1}.
The test object is randomly phased phantom (RPP)  whose modulus is exactly the nonnegative phantom (Figure \ref{fig:ui}(a))
but whose phase is randomly and uniformly distributed in $ [-\pi, \pi]$.  The scan scheme is the raster scan with $\tau=m/2$, i.e. 50\% overlap ratio between adjacent masks. Both RE and RR spike at integer-valued $f$ and the spill-over effect gets worse as $q$ increases.

\subsection{Phase retrieval as feasibility}\label{sec:solution}

For two dimensional, complex-valued objects, let 
$\IC^{n^2}$ be the object space where $n$ is the number of pixels in each dimension. 
Sometimes, it may be more convenient to think of the object space as $\IC^{n\times n}$. 
Let $N$ be the total number of data. The data manifold 
\beq
Y:=  \{u\in \IC^N: |u|=b\}\nn
\eeq
 is { an} $N$-dimensional real torus.  
For phase retrieval it is necessary that  $N>2n^2$ \cite{Balan1}. Without loss of generality  we assume that $A$ has a full rank.

Due to the rectangular nature (more rows than columns) of the measurement matrix $A$, 
it is more convenient to work with the transform domain $\IC^N$. Let $X:=A\IC^{n^2}$, i.e. the range of $A$.

The problem of phase retrieval and ptychography can be formulated as the feasibility problem
\beq
\hbox{Find}\quad  u\in  X \cap Y, \nn
 \eeq
 in the transform domain instead of 
the object domain. Let  $P_X$ and $P_Y$ be the projection onto $X$ and $Y$, respectively.

Let us clarify  the meaning of solution in  the transform domain since $A$ is overdetermining. 
Let $\odot$ denotes the component-wise (Hadamard) product and we can write
\beq
P_X u=AA^+ u,&& P_Y u=b\odot\sgn(u)\nn
\eeq
where the pseudo-inverse 
\[
A^+= (A^*A)^{-1} A^*
\]
becomes $A^*$ if { $A$ isometric which we assume henceforth}.

We refer to $u=e^{\im \alpha}A\x, \alpha\in \IR$, as the {\em true} solution (in the transform domain),
up to a constant phase factor $e^{\im \alpha}$. 
We say that $u$ is a {\em generalized solution} (in the transform domain) if 
\[
|\tilde u|=b,\quad \tilde u:=P_X u. 
\]
In other words,  $u$ is said to be a generalized solution if $A^+u$ is a phase retrieval solution. 
Typically a generalized solution $u$  is neither a feasible solution (since $|u|$ may not equal $b$)  nor unique (since $A$ is overdetermining) and $u+z$ is also a generalized solution
if $P_X z=0$.  

We call  $u$ a {\em regular} solution if $u$ is a  generalized solution and
$P_X u=u$.
Let $\tilde u= P_X u$ for a generalized solution $u$. 
Since  $P_X \tilde u=\tilde u$ and $|\tilde u|=b$, $\tilde u$ is a regular solution.
Moreover, since $P_X R_X u=P_X u$ and $R_X R_X u=u$, $u$ is a generalized solution if and only if  $R_X u$ is a generalized solution.

The goal of the inverse problem \eqref{pr} is the unique determination of $\x$, up to a constant phase factor, from the given data $b$. 
In other words,  uniqueness holds if, and only if, all regular solutions $\tilde u$ in the transform domain have the form
\beqn
\tilde u=e^{\im \alpha} A \x
\eeqn
or equivalently, any generalized solution $u$ is an element of 
the $(2N-2n^2)$ real-dimensional vector space
 \beq
\label{gen}\label{unique}
\{e^{\im \alpha}A\x+z: P_Xz=0,\,\,  z\in \IC^N,\,\,\alpha\in\IR\}.  
\eeq
In the transform domain, the uniqueness is
characterized by the uniqueness of the regular solution, up to a constant phase factor. 
Geometrically, uniqueness  means that the intersection $X\cap Y$ is a circle (parametrized $e^{\im \alpha}$ times $A\x$).

\subsection{Noise models and log-likelihood functions}

In the noisy case, it is more convenient to work with the optimization framework instead of the feasibility framework.
When the noise statistics is known, it is natural to consider the maximum likelihood estimation (MLE) framework.
In MLE,  the negative log-likelihood function is the natural choice for the loss function.  

\subsubsection{Poisson noise}
For the Poisson noise,  the negative log-likelihood function is \cite{ML12},\cite{Poisson2}
\beq
L(u) &=&\sum_i |u(i)|^2-b^2(i)\ln |u(i)|^2. \label{Poisson}
\eeq
A disadvantage of working with the Poisson loss function \eqref{Poisson} is the occurrence of
divergent derivative where $u(i)$ vanishes but $b(i)$ does not. This roughness can be softened as
follows. 

At the high signal-to-noise (SNR) limit, the Poisson distribution
\[
P(n)={\lambda^ne^{-\lambda}\over n!}
\]
has the asymptotic limit 
\beq
\label{PG1}
P(n)\sim {e^{-(n-\lambda)^2/(2\lamb)}\over \sqrt{2\pi\lamb}}. 
\eeq
Namely  in the low noise limit the Poisson noise is equivalent to the Gaussian noise of the mean $|A\x|^2$ and
 the variance equal to the intensity of the diffraction pattern.  
The overall SNR can be tuned by varying the  signal energy $\|A\x\|^2$. 

The negative log-likelihood function for the right hand side of  \eqref{PG1} is 
\beq
\label{pg}
\sum_j \ln |u(j)| +{1\over 2} \lt|{b^2(j)\over |u(j)|}-|u(j)|\rt|^2
\eeq
which is even rougher than \eqref{Poisson} where $u(i)$ vanishes but $b(i)$ does not. 
To rid of the divergent derivatives at $u(j)=0$ we make the substitution 
\[
{{b(j)}\over |u(j)|}\to 1,\quad \ln|u(j)|\to \ln{b(j)}=\mbox{const.},
\]
in \eqref{pg} and obtain
\beq
 L(u)&=&\half \| |u|- b\|^2 \label{Gaussian}
\eeq
after dropping irrelevant constant terms. Expanding the loss function \eqref{Gaussian} 
\beq\label{3'}
L(u)&=& \frac{1}{2}\|u\|^2-\sum_{j} b(j)|u(j)|+\frac{1}{2}\|b\|^2
\eeq   
we see that \eqref{3'} has a bounded sub-differential  where $u(j)$ vanishes but $b(j)$ does not.
There are various tricks to smooth out \eqref{Gaussian} e.g. by introducing an additional regularization parameter as
\beq
 L(u)&=&\half \| \sqrt{|u|^2+\ep}- \sqrt{b^2+\ep}\|^2,\quad\ep>0\nn
\eeq
(see  \cite{March19}).

\subsubsection{Complex Gaussian noise}

Another type of noise due to interference from multiple scattering can be modeled as  complex circularly-symmetric Gaussian noise
(aka Rayleigh fading channel), resulting in
\beq
\label{ray}
b&=&|A\x+\eta|
\eeq
where $\eta$ is a complex circularly-symmetric Gaussian noise. 
Squaring the expression, we obtain
\beq
b^2&=& |A\x|^2+|\eta|^2+2\Re(\overline{\eta}\odot A\x)\nn
\eeq

Suppose $|\eta|\ll |A\x|$ so that $|\eta|^2 \ll 2\Re(\overline{\eta}\odot A\x)$.
Then 
\beq
b^2&\approx &   |A\x|^2+2\Re(\overline{\eta}\odot A\x).\label{23-1}
\eeq
Eq. \eqref{23-1} says that at the photon counting level, the noise appears additive and Gaussian but with variance
proportional to $|A\x|^2$, resembling the distribution \eqref{PG1}. Therefore 
the loss function \eqref{Gaussian} is suitable for Rayleigh fading interference noise at low level. 

\subsubsection{Thermal noise}
 
 On the other hand, if the measurement noise is thermal (i.e. incoherent background noise) as in 
 \beq
 |b|^2&=& |A\x|^2+\eta,\nn
 \eeq
 where $\eta$ is real-valued Gaussian vector of covariance $\sigma^2 I_N$,   then the suitable loss function
 is 
 \beq
 L(u)&=&\half \| |u|^2- b^2\|^2 \label{Gaussian2}
\eeq
which is smooth everywhere. 
 See \cite{noise}, \cite{noise2},\cite{noise3} for more choices of loss functions. 

 In general the amplitude-based Gaussian loss function  \eqref{Gaussian} outperforms the
intensity-based loss function \eqref{Gaussian2} \cite{Waller15}.

Finally, we note that the ambiguities discussed in Section \ref{sec:amb1} are global minimizers of 
the loss functions \eqref{Poisson},  \eqref{Gaussian} and \eqref{Gaussian2} along with $e^{\im\theta} A\x$ 
in the noiseless case. Therefore, to remove the undesirable global minimizers, we need sufficient number of measurement data
as well as proper measurement schemes.

\subsection{Spectral gap and local convexity}
 For sake of convenience, we shall assume that $A$ is an isometry which can always be realized by rescaling the columns of the measurement matrix. 

In local convexity of the loss functions as well as geometric convergence of iterative algorithms, the following matrix plays a central role:
  \beq B&=& \diag\lt[\sgn(\overline{A x})\rt] A\label{25.1}
 \eeq
which  is an isometry and  varies with $x$. 

With the notation  
 \beq\label{3.6'}
\nabla f(x)&:= &
{\half}\lt({\partial f(x)\over \partial \Re(x)}+i {\partial f(x) \over \partial \Im(x)}\rt),\quad x\in \IC^{n^2} 
\eeq
we can write the sub-gradient of  the loss function \eqref{Gaussian} as
\beqn
{2}\Re[\zeta^* \nabla L(Ax)] &=&\Re(x^{*}\zeta)-b^{\top} \Re(B \zeta),\quad 
\forall \zeta\in \IC^{n^2}. 
\eeqn
In other words,  $x$ is a stationary point if and only if 
\beqn
x&=&B^*b=A^* (\sgn(Ax)\odot b)
\eeqn
or equivalently
\beq
B^*\lt[|Ax|-b\rt]=0.\label{stationary}
\eeq
Clearly, with noiseless data, $|A\x|=b$ and hence $\x$ is a stationary point. 
In addition, there likely are other stationary points 
since $B^*$ has many more columns than rows. 

{On the other hand,  with noisy data  there is no regular solution to $|Ax|=b$ with high probability (since $A$ has many more rows than columns)  and the true solution $x_*$ is unlikely to be a stationary point (since \eqref{stationary} imposes extra constraints on noise).}

Let $\mbox{Hess}(x)$ be the Hessian of $L(Ax)$. If $Ax$ has no vanishing components,  $\mbox{Hess}(x)$ can be given explicitly as 
 \beqn
\Re[\zeta^*\mbox{Hess}(x) \zeta]&=&\|\zeta\|^2-\Im(  B\zeta)^T\diag\lt[{b\over |Ax|}\rt]\Im(  B\zeta),\quad 
\forall \zeta\in \IC^{n^2}.
\eeqn

\begin{thm}
\label{gap}\cite{AP-phasing}, \cite{FDR}, \cite{DRAP-ptycho}
Suppose $\x$ is not a line object. For $A$ given by \eqref{one}, \eqref{two} or the ptychography scheme under the connectivity condition \eqref{s-conn}  with independently and continuously distributed mask phases,  the second largest singular value $\lambda_2$ of the real-valued matrix
\beq 
  \label{B'} \mathcal{B}=\left[
\begin{matrix}
- \Re(B)    &
 \Im(B)
\end{matrix}
\right]
\eeq
 is strictly less than 1 
 with probability one. 
 
 Therefore, the Hessian of \eqref{Gaussian} at $A\x$ (which is nonvanishing almost surely) is positive semi-definite and has
 one-dimensional eigenspace spanned by $\im \x$ associated with eigenvalue zero.

\end{thm}

\section{Nonconvex optimization}\label{s:nonconvex}

\subsection{Alternating Projections (AP)}\label{ss:ap}

The earliest phase retrieval algorithm for a non-periodic object (such as a single molecule)  is the Gerchberg-Saxton algorithm~\cite{GS72} and its variant, Error Reduction \cite{Fie82}. The basic idea is Alternating Projections (AP), going all the way back to the works of 
von Neumann, Kaczmarz and Cimmino  in the 1930s \cite{Cimmino}, \cite{Kac}, \cite{Neuman}. 
And these further trace the history  back to  Schwarz \cite{Schwarz}  who in 1870 used AP to solve the Dirichlet
problem on a region given as a union of regions each having a simple to solve Dirichlet problem.

AP is defined by 
\beq\label{3.9}\label{papf}
x_{k+1}=A^*[ b\odot \sgn(Ax_k)].\eeq
  In the case with real-valued objects, \eqref{papf} is exactly Fienup's Error Reduction algorithm \cite{Fie82}.  

The AP fixed points satisfy 
\beqn
x=A^*[b\odot { \sgn}(Ax)]&\mbox{or}& B^*[|Ax|-b]=0
\eeqn
which is exactly the stationarity  equation \eqref{stationary} for $L$ in \eqref{Gaussian}. 
The existence of non-solutional fixed points (i.e. $|Ax|\neq b$), and hence local minima of $L$ in \eqref{Gaussian}, can not
be proved presently but manifests in numerical stagnation of AP iteration. 
 
Indeed,  AP can be formulated as a  gradient descent  for the loss function \eqref{Gaussian}.  The function \eqref{Gaussian} has the sub-gradient 
\beqn
{2}\nabla L(Ax) &= &x-A^*[b\odot  \sgn(A x)]
\eeqn
and hence we can write the AP map as 
\beqn
{T}(x)&=&x-{2} \nabla L(Ax)
 \eeqn
implying  a constant
step size $1$. In \cite{AP-phasing}, local geometric convergence to $\x$ is proved for AP. 
In other words, AP is both noise-agnostic in the sense that it projects onto the data set 
as well as noise-aware in the sense that it is the sub-gradient descent of the loss function \eqref{Gaussian}.

The following result identifies  any limit point of the AP iterates  
with a fixed point of AP with 
a norm criterion for distinguishing the phase retrieval solutions from the non-solutions among many coexisting fixed points.

\begin{prop}\label{Cauchy}  \label{prop2.21} \cite{AP-phasing}
Under the conditions of Theorem \ref{unique1} or \eqref{unique2}, the AP sequence $x_k={T}^{k-1}(x_1)$,  with any starting point $x_1$,  is bounded
and 
every limit point  is a fixed point.

Furthermore, if  a fixed point $x$ satisfies $\|A x\|= \|b\|$, then
 $|Ax|=b$ almost surely. 
On the other hand, if $|Ax|\neq b$, then $\|A x\|<\|b\|$. 
\end{prop}

\subsection{Averaged Alternating Reflections (AAR)}
\label{sec:AAR} 

AAR is based on the following characterization of {\em convex} feasibility problems. 

Let \beqn
R_X=2P_X-I, &&R_Y=2P_Y-I.
\eeqn
 Then we can characterize the feasibility condition as 
\beqn
u\in X\cap Y\quad \mbox{if and only if }\quad u=R_YR_X u 
\eeqn
in the case of convex constraint sets $X $ and $Y$ \cite{Boyd17}. This  motivates the Peaceman-Rachford (PR) method: For $k=0,1, 2, \cdots$
\[
u_{k+1}= R_Y R_X y_{k}. 
\]
AAR is the {\em averaged} version of PR: For $k=0,1, 2, \cdots$
\beq
\label{dr}
u_{k+1}= \half u_k+\half R_Y R_X u_{k},
\eeq
hence the name {\em Averaged Alternating Reflections} (AAR). 
With a different variable $v_k:=R_X u_k$, we see that AAR \eqref{dr} is equivalent to
\beq
v_{k+1}=\half v_k+ \half R_X R_Y v_k.\label{drx}
\eeq
In other words, the order of applying $R_x$ and $R_Y$ does not matter. 

A standard result for AAR in the convex case is this.

\begin{prop} \cite{BCL04} Suppose $X$ and $Y$ are closed and convex sets of a finite-dimensional vector space $E$. Let $\{u_k\}$ be an AAR-iterated sequence for any $u_1\in E$. Then one of the following alternatives holds:

\noindent (i) $X\cap Y\neq \emptyset$ and $(u_k)$ converges to a point $u$ such that $P_X u\in X\cap Y$;\\
(ii) $X\cap Y= \emptyset$ and $\|u_k\|\to\infty$. 
\label{prop0}
\end{prop}
In alternative (i), the limit point $u$ is a fixed point of the AAR map \eqref{dr}, which is necessarily in $X\cap Y$; in alternative (ii) the feasibility problem is inconsistent, resulting in divergent AAR iterated sequences, a major drawback of AAR {since the inconsistent case is prevalent with noisy data because of the higher dimension of data compared to the object}. 

Accordingly,  the alternative (i) in Proposition \ref{prop0} means
that if a convex feasibility problem is consistent then   every AAR iterated sequence converges to a generalized solution and hence every fixed point is a generalized solution.

We begin with showing that AAR can be viewed as an ADMM method
with  the indicator function $\II_Y $ of the set $Y=\{z\in \IC^N: |z|=b\}$ as
the loss function.

AAR for phase retrieval can be viewed as relaxation of the linear constraint of $X$
by alternately minimizing 
 the augmented Lagrangian  function 
 \beq\label{AL'}
 \cL(z,x,\lamb)&=& \II_Y (z)+\lamb^*(z- A x)+{1\over 2} \|z-Ax\|^2
 \eeq
 in the order 
 \beq
 \label{350} z_{k+1}&=& \arg\min_z  \cL(z,x_k,\lamb_k)= P_Y \lt[A x_k-\lamb_k\rt]\\
\label{351} x_{k+1}&=&\arg\min_{\nu} \cL(z_{k+1},x,\lamb_k) = A^+(z_{k+1}+\lamb_k)\\
  \label{352} \lamb_{k+1}&=& \lamb_k+z_{k+1}-Ax_{k+1}.
 \eeq
Let  $u_{k}:= z_k+\lamb_{k-1}$  and 
 we have from \eqref{352} 
 \beqn
 \lamb_{k}&=& u_{k}-Ax_{k}\\
 &=& u_{k}-P_X u_{k}\nn
 \eeqn
 and hence
 \beqn
 u_{k+1}&=& P_Y (A x_k-\lamb_k)+\lamb_k\\
 &=& P_Y (P_X u_k-\lamb_k)+\lamb_k\\
 &=& P_Y R_X u_k+u_k-P_X u_k\\
 &=& \half u_k+\half R_Y R_X u_k
 \eeqn
 which is AAR \eqref{dr}.

As proved in  \cite{FDR}, when  uniqueness holds, the fixed point set of the AAR map \eqref{dr} is exactly the continuum set 
\beq
\label{fixed}
\{u=e^{\im \alpha}A\x-z: P_Xz=0,\sgn(u)=\alpha+\sgn(A\x), z\in \IC^N, \alpha\in \IR\}. 
\eeq
In \eqref{fixed}, the phase relation $\sgn(u)=\alpha+\sgn(A\x)$ implies that
$z=\eta\odot \sgn(u),\eta\in \IR^N, b+\eta\ge 0.$ So the set \eqref{fixed} can be written as
\beq
\label{fixed2}
\{e^{\im \alpha} (b-\eta)\odot \sgn(A\x):P_X(\eta\odot\sgn(A\x))=0, b+\eta\ge 0,\eta\in \IR^N,\alpha\in \IR\}, 
\eeq
which is an $(N-2n^2)$ real-dimensional set,  a much larger set than the circle $\{e^{\im \alpha} A\x: \alpha \in \IR\}$ for a given $f$. 
 On the other hand, the fixed point set \eqref{fixed2} is
 $N$-dimension lower than  
 the set  \eqref{gen} of generalized solutions {and projected (by $P_X$) onto the circle of true solution $\{e^{\im \alpha} A\x:\alpha\in \IR\}$}. 
 
 A more intuitive characterization of the fixed points can be obtained by applying $R_X$ to the set \eqref{fixed2}. Since
 \[
 R_X [e^{\im \alpha} (b-\eta)\odot \sgn(A\x)] = e^{\im \alpha} (b+\eta)\odot \sgn(A\x)
 \]
 amounting to the sign change in front of $\eta$, 
 the set \eqref{fixed2} under the map $R_X$ is mapped to
 \beq
\label{fixed2'}
\{e^{\im \alpha} (b+\eta)\odot \sgn(A\x):P_X(\eta\odot\sgn(A\x))=0,\,\, b+\eta\ge 0,\,\,\eta\in \IR^N,\alpha\in \IR\}. 
\eeq
The set \eqref{fixed2'} is the fixed point set  of the alternative form of AAR:
\beq
v_{k+1}&=& \half x_k+\half R_X R_Y v_k \label{aar2}
\eeq
in terms of  $v_k:=R_X u_k$. The expression \eqref{fixed2'} says that the fixed points of \eqref{aar2} are generalized solutions with the ``correct" Fourier phase.

However, the boundary points of the fixed point set \eqref{fixed2'} are degenerate in the sense that they have vanishing components, i.e. $|v|(j)=(b+\eta)(j)=0$ for some $j$ and   can slow down convergence \cite{Fie86}.  Such points are points of discontinuity of the AAR map \eqref{aar2}
because they are points of discontinuity of $P_Y=b\odot\sgn(\cdot)$. 
 Indeed, even though AAR converges linearly  
in the vicinity of the true solution, numerical evidence suggests that
globally (starting with a random initial guess) AAR  converges sub-linearly. 
Due to the non-uniformity of convergence, the additional step of applying $P_X$ (Proposition \ref{prop0}(i))  at the ``right timing" of the iterated process can jumpstart the geometric convergence regime   \cite{FDR}.

As noted in Section \ref{sec:amb1}, with a uniform mask, noiseless data and the real-valued prior, all the ambiguities in \eqref{12'} are global
minima of $L$ in \eqref{Gaussian} and fixed points of both AP and AAR. Figure \ref{fig:ui} demonstrates how detrimental these ambiguities are to numerical reconstruction.

\subsection{Douglas-Rachford Splitting (DRS)}

AAR \eqref{dr} is often written in the following form
\beq\label{aar}
u_{k+1}= u_k+P_YR_Xu_k-P_Xu_k  
\eeq
which is 
equivalent to the 3-step iteration
\beq
\label{dr2}
v_{k} &=& P_Xu_k; \\
w_{k} &=& P_Y (2v_{k}-u_k)=P_YR_X u_k\\
u_{k+1}&= & u_k+ w_{k}-v_{k}\label{dr2'}
\eeq

AAR can be modified in various ways by the powerful method of 
Douglas-Rachford splitting  (DRS) which is simply an application of the 3-step procedure \eqref{dr2}-\eqref{dr2'} to
 proximal maps. 

Proximal maps are generalization of projections. The proximal map relative to a function $f$
is defined by
\[
\prox_{f}(u):= \mathop{\text{argmin}}\limits_{x} {f}(x)+\frac{1}{2}\| x-u\|^2. 
\]
Projections $P_X$ and $P_Y$ are proximal maps relative to $\II_X$ and $\II_Y$, the indicator functions
of $X$ and $Y$, respectively.

By choosing other proxy functions than $\II_X$ and $\II_Y$, we may obtain different DRS methods that have
more desirable properties than AAR. 

\subsection{Convergence rate}\label{ss:conv}
Next we recall the local geometric convergence property of AP and AAR with convergence rate expressed in terms
of $\lamb_2$, the second largest singular value of $\cB$.

The {Jacobians} of the AP and AAR maps are given, respectively, by
\beqn
\partial {T} (\xi)&=&\im B^*\Im(B \xi),\quad\xi\in\IC^{n^2}
\eeqn
and
\begin{equation} \partial \Gamma (\zeta)=  (I-BB^*) \zeta + \im (2 B B^*-I)\; \diag\lt[{b\over |\zeta|}\rt]\Im(\zeta),\quad\zeta\in\IC^N. 
\nn
 \end{equation}
 Note that $\partial\Gamma$ is a {\em real}, but {\em not} {\em complex},  linear map since
$\partial \Gamma (c\zeta)\neq c\partial \Gamma (\zeta),  c\in \IC$ in general.

  \begin{thm}\label{cor5.2} \cite{FDR}, \cite{DRAP-ptycho},\cite{AP-phasing}
 The local geometric convergence rate of AAR and AP is 
 $\lamb_2$ and $\lamb_2^2$, respectively, where $\lamb_2$ is the second largest singular value of $\cB$ in \eqref{B'}.  

 \end{thm}
 
 As pointed out above, AAR has the true solution as the unique fixed point in the object domain  while AP has a better convergence rate than DR {(since $\lamb_2^2<\lamb_2$)}. A reasonable  way to combine their
strengths is to use AAR as the initialization method for AP.

With a carefully chosen parameter $f$ ($=6/(5\pi)$), the performance of 
 a Fresnel mask (Figure \ref{fig:different-q}(b)) is only slightly inferior to that of a random mask (Figure \ref{fig:different-q}(a)). 
 Figure \ref{fig:different-q} also demonstrates different convergence rates of AP with various $q$. 

\begin{figure}
\centering
 \subfigure[random mask]{\includegraphics[width=5.5cm]{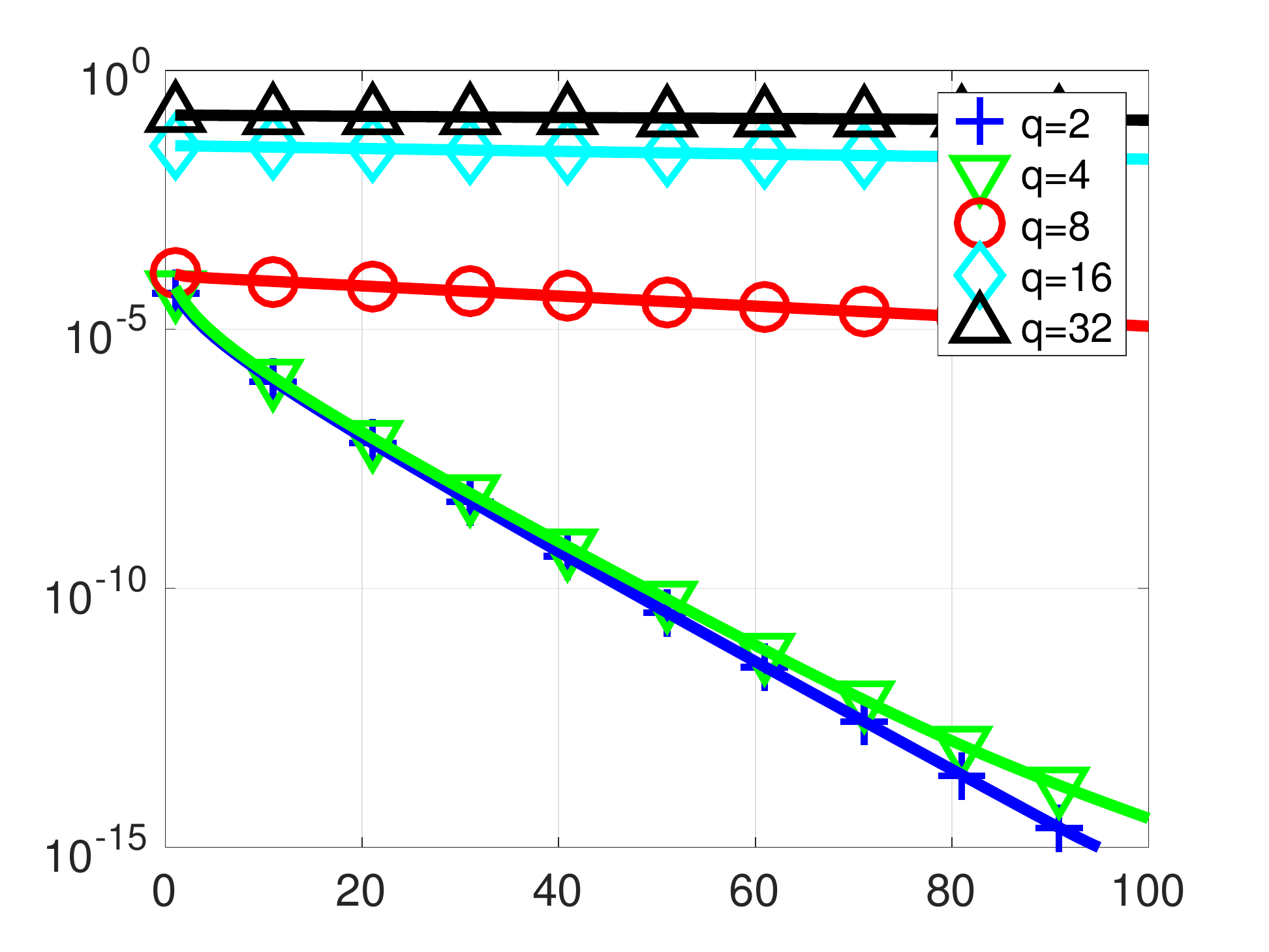}}
  \subfigure[Fresnel mask with ${f}={6\over 5\pi}$]{\includegraphics[width=5.5cm]{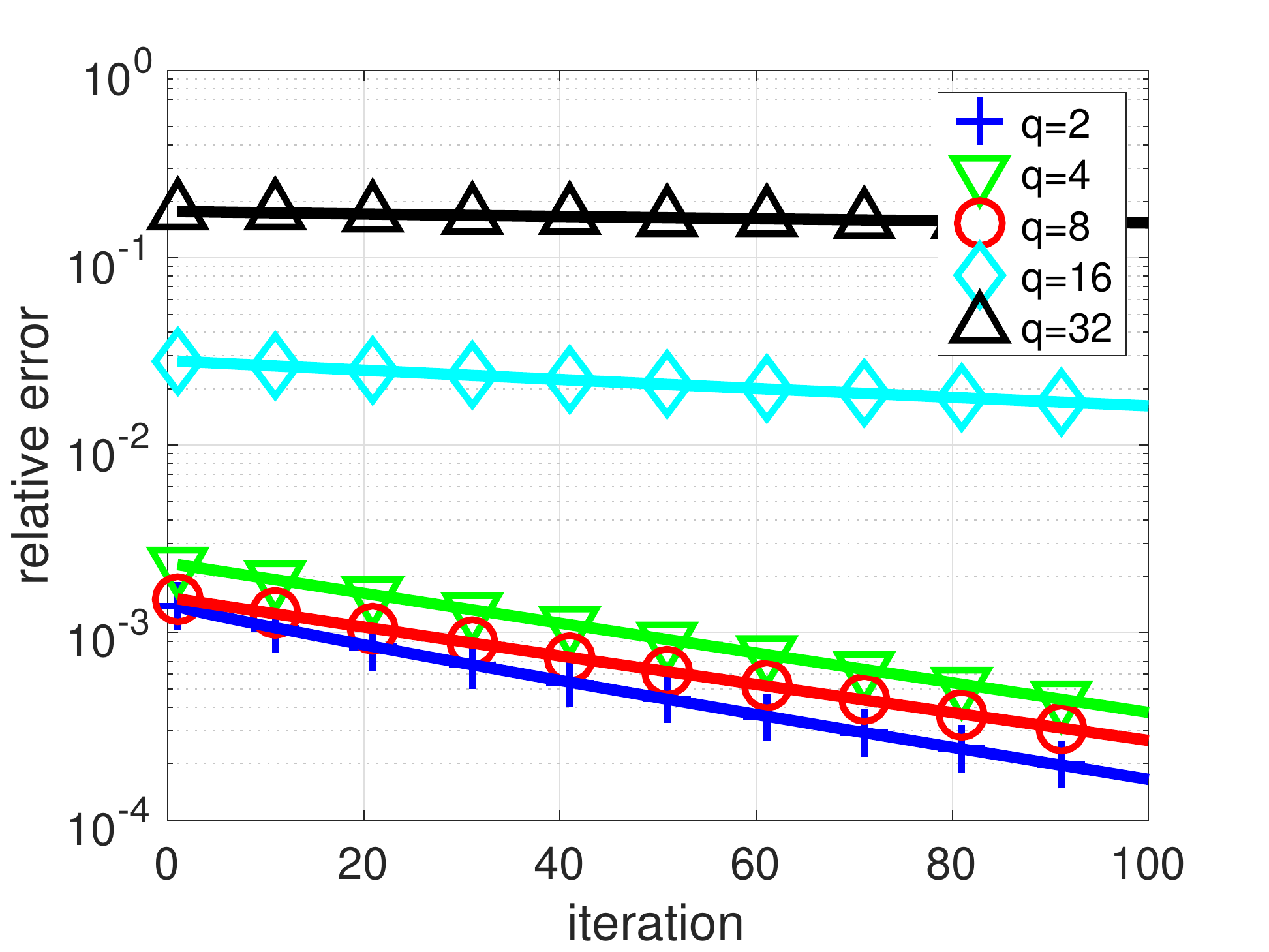}}\hspace{0.5cm}
         \subfigure[${f}={6\over 5\pi}\approx 0.38$]{\includegraphics[width=4cm]{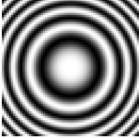}}
  \caption{ RE on the semi-log scale  for the $128\times 128 $ RPP of phase range $[0,2\pi]$ vs 100 AP iterations
  after initialization given by 300 AAR iterations with various $q$.
  }  
  \label{fig:different-q}
\end{figure}

\subsection{Fourier versus object domain formulation}\label{sec:Fourier-object}

It is important to note that due to the rectangular nature (more rows than columns) of the measurement matrix $A$, 
the following {\em object domain} version is a different algorithm from AAR discussed above:
\beq
\label{drxx}
x_{k+1}
&=& x_k+A^+ R_Y ( A x_k)-A^+P_Y(Ax_k)
 \eeq
 which resembles { \eqref{aar}} but operates on the object  domain  instead of the transform domain. 
 Indeed, as demonstrated in \cite{FDR}, the object domain version \eqref{drxx} significantly underperforms the Fourier domain AAR. 

As remarked earlier, this problem can be rectified by zero-padding and embedding the original object vector into $\IC^N$ and explicitly accounting
for this additional support constraint. Let $P_S$ denote the projection from $\IC^N$ onto the zero-padded subspace and let $\tilde A$ be an
invertible  extension of $A$ to $\IC^N$.  Then it is not hard to
see  that
the ODR map
\beqn
G(x)&=& x+P_S\tilde A^{-1} R_Y \tilde A x-\tilde A^{-1} P_Y\tilde A x
\eeqn
satisfies 
\[
\tilde A G \tilde A^{-1} (y)= y+\tilde A P_S\tilde A^{-1}R_Y y-P_Y y
\]
which is equivalent to {\eqref{aar}}  once we recognize that $P_X= \tilde A P_S\tilde A^{-1}$. 

In terms of the enlarged object space $\IC^N$,  Fienup's well-known Hybrid-Input-Output (HIO) algorithm can be expressed as  
\beqn
 x_{k+1}&=& \half \tilde A^{-1}\lt[R_X \lt(R_Y+(\beta-1)P_Y\rt)+I+(1-\beta) P_Y\rt]\tilde A x_k
\eeqn
 \cite{Fie82}. With  $v_k=\tilde A x_k,$ we can also express HIO in the Fourier domain
 \beq
 v_{k+1}&=& \half \lt[R_X \lt(R_Y+(\beta-1)P_Y\rt)+I+(1-\beta) P_Y\rt] v_k.\label{hio2}
\eeq
For $\beta=1$, HIO \eqref{hio2} is exactly AAR \eqref{drx}. 

It is worth pointing out again that the lifting from $\IC^{n^2}$ to $\IC^N$ is a key to the success of HIO over AP {\eqref{papf}, which is an object-domain scheme}. In the optics literature, however, the measurement matrix is usually constructed as a square matrix by zero-padding the
object vector with sufficiently large dimensions (see e.g. \cite{Miao} \cite{Miao2}). Zero-padding, of course, results in an additional
support constraint that must be accounted for explicitly.

\subsection{Wirtinger Flow}\label{ss:WF}

We already mentioned that  the  AP map \eqref{papf} is  a gradient descent  for the loss function \eqref{Gaussian}.  
In a nutshell, Wirtinger Flow is a gradient descent  algorithm with the loss function \eqref{Gaussian2}
proposed by~\cite{candes2015phase} which establishes a basin of attraction at $\x$ of radius $O(n^{-1/2})$ for a sufficiently small step size.  

Unlike many other non-convex methods, Wirtinger Flow (and many of its modifications) comes with a rigorous theoretical framework that provides explicit performance guarantees in terms of required number of measurements, rate of convergence to the true solution, and robustness bounds.
The Wirtinger Flow approach consists of two components: 
\begin{itemize}
\setlength{\itemsep}{-0.5ex}
\item[(i)] a carefully constructed initialization based on a spectral method related to the PhaseLift framework;
\item[(ii)]  starting from this initial guess, applying iteratively a gradient descent type update. 
\end{itemize}
The resulting algorithm is computationally efficient and, remarkably, provides rigorous guarantees under which it will recover the true solution. 
We describe the Wirtinger Flow approach in more detail.  We consider the non-convex problem
\[
\min_z \,\, L(z):= \frac{1}{2N}  \sum_{k=1}^N \left( |\langle a_k, z \rangle |^2 - y_k \right)^2, \qquad z \in \CC^n.
\]
The gradient of $L(z)$  is calculated via the Wirtinger {gradient \eqref{3.6'}}
\[
 \nabla L(z_j) =   \frac{1}{N} \sum_{k=1}^N ( |\langle a_k, z \rangle |^2 - y_k  )  \langle a_k,z\rangle a_k .
 \]
Starting from some initial guess $z_0$, we compute
\begin{equation}\label{wfgrad} 
z_{j+1} =    z_j - \frac{\tau_j}{\|z_0\|_2^2}  \nabla L(z_j),
\end{equation}
where $\tau_j>0$ is a stepsize (learning rate). {Note that the Wirtinger flow, like AP \eqref{papf}, is an object-domain scheme.}

The initialization of $z_0$ is computed via spectral initialization discussed in more detail in Section~\ref{ss:spectral}. We set 
$$\lambda:= n \frac{\sum_j n_j}{\sum_k \|a_k\|_2^2},$$
and let $z_0$ be the principal eigenvector of the matrix 
\begin{equation}\nn
Y = \frac{1}{N} \sum_{k=1}^N y_k a_k a_k^\ast, 
\end{equation}
where $z_0$ is normalized such that $\|z_0\|_2^2 = \lambda.$

\begin{defi}\label{def:dist}
Let $x \in \CC^n$ be any solution to~\eqref{eq:data}. For each $z \in \CC^n$, define
$$
\dist(z,x) = \min_{\phi \in [0,2\pi)} \|z - e^{i\phi} x\|_2.
$$
\end{defi}

\begin{thm}\cite{candes2015phase}
\label{th:WT}
Assume that the measurement vectors $a_k \in \CC^n$ satisfy $a_k \overset{\text{i.i.d.}}{\sim} {\mathcal N}(0,I/2)+i{\mathcal N}(0,I/2)$.
Let $\xo \in \CC^n$ and  $y = \{\langle a_k, \xo\rangle |^2\}_{k=1}^N$ with $N \ge c_0 n \log n$, where
$c_0$ is a sufficiently large constant. Then the Wirtinger Fow initial estimate $z_0$ normalized such that $ \| z_0\|_2 =  m^{-1} \sum_k y_k$, obeys
\begin{equation}\label{zinit}
\dist(z_0,\xo) \le \frac{1}{8} \|\xo\|_2,
\end{equation}
with probability at least  $1-10e^{-\gamma n} -8/n^2$, where $\gamma$ is a  fixed constant. Further, choose a constant stepsize $\tau_j = \tau$
for all $j=1,2,\dots$, and assume $\tau \le c_1/n$ for some  fixed constant $c_1$. Then with high probability
starting from any initial solution $z_0$ obeying~\eqref{zinit}, we have
$$
\dist(z_j,\xo) \le \frac{1}{8} \left(1- \frac{\tau}{4}\right)^{j/2} \|\xo\|_2.
$$
\end{thm}

\medskip
A modification of this approach, called Truncated Wirtinger Flow~\cite{chen2017solving}, proposes a more adaptive gradient flow, both at the initialization step and during iterations. This modification seeks to reduce the variability of the iterations by introducing three additional  control parameters \cite{chen2017solving}. 

Various other modifications of Wirtinger Flow have been derived, see e.g.~\cite{TAF18,tu2015low,cai2016optimal}.
While it is possible to obtain global convergence for such gradient descent schemes with random initialization~\cite{chen2019gradient}, the price is a larger number of measurements. See Section~\ref{s:init} for a detailed discussion and comparison of various initializers combined with Wirtinger Flow.

The general idea behind the Wirtinger Flow of solving a non-convex method provably by a careful initialization followed by a properly chosen gradient descent algorithm has inspired research in other areas, where rigorous global convergence results for gradient descent type algorithms have been established (often for the first time). This includes blind deconvolution~\cite{li2019rapid,ma2018implicit}, blind demixing~\cite{ling2019regularized,jung2017blind}, and matrix completion~\cite{sun2016guaranteed}.

\subsection{Alternating Direction Method of Multipliers (ADMM)}

Alternating Direction Method of Multipliers (ADMM)
is a powerful method for solving 
the joint  optimization problem:
\begin{equation}
\min\limits_{u}K(u)+ L(u)
\label{DRS}
\end{equation}
where the loss functions $ L$ and $K$ represent the data constraint $Y$ and
the object constraint $X$, respectively. 

Douglas-Rachford splitting  (DRS) is another effective method for the joint optimization problem \eqref{DRS}
with a linear constraint. 
For convex optimization, DR splitting applied to the primal problem  is equivalent to ADMM applied to the Fenchel dual problem \cite{Gabay}.
For nonconvex optimization such as \eqref{DRS} there is no clear relation between the two in general. 

However, for phase retrieval,
DRS and ADMM are essentially equivalent to each other \cite{DRS-ptycho}. 
So our subsequent presentation will mostly focus on ADMM.

ADMM seeks to minimize the augmented Lagrangian function 
\beq
\label{AL2}
\cL(y,z)=K(y)+L(z)+\lamb^*(z-y)+{\rho\over 2}\|z-y\|^2
\eeq
 alternatively  as
\beq
\label{700'}y_{k+1}&=&\arg\min_x\cL(y,z_k,\lamb_k)\\
\label{701'}z_{k+1}&=& \arg\min_z\cL(y_{k+1}, z, \lamb_k)
\eeq
or
\beq
\label{700'-2}z_{k+1}&=&\arg\min_x\cL(y_k,z,\lamb_k)\\
\label{701'-2}y_{k+1}&=& \arg\min_z\cL(y, z_{k+1}, \lamb_k)
\eeq
and then update the multiplier by the gradient ascent 
\beq
\nn\lamb_{k+1}=\lamb_k+\rho (z_{k+1}-y_{k+1}). 
\eeq

\subsection{Noise-aware ADMM}\label{sec:noisy-admm}

We apply ADMM  to the augmented Lagrangian $\cL$ \eqref{AL2}
with $K=\II_X$ (the indicator function of the set $X$) and $L$  given by the Poisson \eqref{Poisson} or Gaussian \eqref{Gaussian} loss function. 

Consider \eqref{700'-2}-\eqref{701'-2} and let
\[
u_k:=z_k+\lamb_{k-1}/\rho.
\]
Then  we have
\beq
\label{700"}z_{k+1}&=&\prox_{L/\rho}(y_k-\lamb_k/\rho)\\
\label{701"}y_{k+1}&=&\prox_{K/\rho}(z_{k+1}+\lamb_k/\rho)=AA^*(z_{k+1}+\lamb_k/\rho)\\
\label{702''}\lamb_{k+1}&=&\lamb_k+\rho (z_{k+1}-y_{k+1}). 
\eeq

We have from \eqref{702''} that
\beqn
u_{k+1}=y_{k+1}+\lamb_{k+1}/\rho.
\eeqn
By \eqref{701"}, we also have
\[
y_{k+1}=P_X (z_{k+1}+\lamb_k/\rho)=P_X u_{k+1}
\]
and 
\beqn
y_k - \lamb_k/\rho= 2y_{k}-u_k=R_X u_{k}. 
\eeqn
So  \eqref{700"} becomes
\beq\nn
z_{k+1}=\prox_{L/\rho}(R_X u_k).
\eeq
Note also that by \eqref{702''}
\beqn
u_k-P_X u_k=\lamb_k/\rho 
\eeqn
and hence
\beqn
u_{k+1}=z_{k+1}+\lamb_k/\rho=u_k-P_X u_k+\prox_{L/\rho}(R_X u_k). 
\eeqn

For the Gaussian loss function \eqref{Gaussian}, the proximal map $\prox_{L/\rho}$  can be calculated exactly
\beqn
\prox_{L/\rho}(u) &=& \frac{1}{\rho+1}b\odot\sgn{(u)}+\frac{\rho}{\rho+1}u\\
&=& \frac{1}{\rho+1}(b+\rho|u|)\odot\sgn{(u)}. 
\eeqn
The resulting iterative scheme is given by 
 \beq\label{G1}
u_{k+1} 
&=& {1\over \rho+1} u_k+{\rho-1\over \rho+1} P_X u_k+\frac{1}{\rho+1}b\odot \sgn \big(R_X u_k\big)\nn\\
&:=&\Gamma(u_k). 
\eeq 

Like AAR, \eqref{G1} can also be derived by the DRS method 
\beqn
v_{l} &= &\prox_{K/\rho}(u_l) =AA^*(u_l); \\
w_{l} &= &\prox_{ L/\rho}(2v_{l}-u_l)\\
u_{l+1}&=& u_l+ w_{l}-v_{l}
\eeqn
instead of  \eqref{dr2}-\eqref{dr2'}.  
For the Gaussian loss function \eqref{Gaussian}, the proximal map $\prox_{ L/\rho}$ is
\beqn
\prox_{L/\rho}(u) &=& \frac{1}{\rho+1}b\odot\sgn{(u)}+\frac{\rho}{\rho+1}u\\
&=& \frac{1}{\rho+1}(b+\rho|u|)\odot\sgn{(u)}, 
\eeqn
an averaged projection with the relaxation parameter $\rho$. 
With this, $\{u_k\}$ satisfy eq. \eqref{G1}. Following \cite{DRS-ptycho}, we refer to \eqref{G1} as the  {\em Gaussian-DRS} map.

For the Poisson case the DRS map has a more complicated form 
\beq	\label{P1}
		\lefteqn{u_{k+1}}\\
		& =&
		\half u_k-{1\over \rho+2} R_X u_k+
		      \frac{\rho}{2(\rho+2)}\lt[|R_Xu_{k}|^2+\frac{8(2+\rho)}{\rho^2}b^2\rt]^{1/2}\odot \sgn{\Big(R_X u_{k}\Big)}\nn\\
		      &:=&\Pi(u_k)\nn
	\eeq
	where $b^2$ is the vector with component $b^2(j)=(b(j))^2$ for all $j$. 
	
Note that $\Gamma(u)$  and $\Pi(u)$ are continuous except where $R_X u$ vanishes but $b$ does not due to
arbitrariness of  the value of the $\sgn$ function at zero.

\subsection{Fixed points} 

{With the proximal relaxation in \eqref{G1}, we can ascertain desirable properties that are either false or unproven for AAR.}

By definition, all  fixed points $u$  satisfy the equation 
\beq\nn
u&=& \Gamma(u)
\eeq
and hence after some algebra
\beq
\nn P_X u+\rho  P^\perp_X u=b\odot \sgn(R_X u)
\eeq 
which in terms of $v=R_X u$ becomes 
\beq
 \label{sa'} P_X v-\rho  P^\perp_X v=b\odot \sgn(v). 
 \eeq

{The following demonstrates the advantage of Gaussian-DRS in
avoiding   the divergence behavior of AAR (as stated in Proposition \ref{prop0} (ii) for the convex case)
when the feasibility problem is inconsistent and has no (generalized or regular) solution.

\begin{thm}\label{thm:bounded}\cite{DRS-ptycho}
Let $u_{k+1}:=\Gamma (u_k), \,\,k\in \IN$. Then,  for $\rho>0$, $\{u_k\}$ is a bounded sequence satisfying 
\beq\nn
\limsup_{k\to\infty} \|u_k\|\le {\|b\|\over \min\{\rho, 1\}}&\mbox{for}& \rho>0. 
\eeq

Moreover, if $u$ is a fixed point, then
 \beq\nn
\|u\| <\|b\|&\mbox{for}& \rho>1
\eeq
and
\beq\nn
\|b\|< \|u\|\le  \|b\|/\rho &\mbox{for}& \rho\in (0,1)
\eeq
unless  $P_X u=u$, in which case $u$ is a regular solution. 
On the other hand, for the particular value $\rho=1$, $\|u\|=\|b\|$ for any fixed point $u$.  
\end{thm}

{
The next result says that all attracting points are regular solutions and hence one need not worry about numerical stagnation. }

 \begin{thm}\cite{DRS-ptycho} Let $\rho\ge 1$. 
Let  $u$ be a fixed point such that  $R_X u$ has no vanishing components. Suppose that the {Jacobian} $J$  of
Gaussian-DRS satisfies
 \beq
 \nn
\|J(\eta)\| \le \|\eta\|,\quad\forall \eta\in \IC^N. 
   \eeq
Then
\beq
\nn
 u=P_X u=b\odot\sgn(R_X u),
\eeq
implying $u$ is a regular solution. 
\label{thm:stable}
\end{thm}

{
The indirect implication of Theorem \ref{thm:stable} is noteworthy: In the inconsistent case (such as with noisy measurements prohibiting the existence of a regular solution), convergence is  impossible since all fixed points are locally repelling in some directions. 
The outlook, however,  need not be pessimistic:  A good iterative scheme need not converge in the traditional sense as long as it produces a good outcome when properly terminated, i.e. its iterates stay in the true solution's  vicinity of size comparable to the noise level. 
In this connection, let us  recall the previous observation that in the inconsistent case the true solution is probably not a stationary point of the loss function. Hence a convergent iterative scheme to a stationary point may not a good idea.  
The fact that Gaussian-DRS performs well in noisy blind ptychography (Figure \ref{fig:noise}(b)) with an error amplification factor of about 1/2 dispels much of the pessimism. 
}

The next result says that for any $\rho\ge 0$,  all regular solutions are
{indeed attracting} fixed points.
 
 \begin{thm}\cite{DRS-ptycho}\label{thm:stable2} Let $\rho\ge 0$. 
Let $u$ be  a nonvanishing regular solution. 
Then the {Jacobian} $J$  of 
Gaussian-DRS is nonexpansive: 
\beq
\nn
\|J(\eta)\| \le \|\eta\|,\quad \forall \eta \in \IC^N. 
\eeq

\end{thm}

Finally {we are able to pinpoint  the parameter corresponding to
the optimal rate of convergence}.

   \begin{thm}\cite{DRS-ptycho}\label{cor1} The  leading singular value of the {Jacobian} of Gaussian-DRS is  1 and the second largest singular
   value is strictly less than 1.  Moreover the second largest singular value as a function of the parameter $\rho$ is increasing 
   over $[\rho_*,\infty)$ and decreasing over $[0,\rho_*]$
   achieving
   the global minimum 
 \beq
 \label{450}
 {\lamb_2\over \sqrt{1+\rho_* }}\quad  \mbox{at}\quad  \rho_*=2\lamb_2\sqrt{1-\lamb_2^2}\in [0,1]
 \eeq
 where $\lamb_2$ is the second largest singular value of $\cB$ in {\eqref{B'}}. 
 
 Moreover, for $\rho=1$, the local convergence rate is $\lamb_2^2$ the same as AP.
   
  \end{thm}

 By arithmetic-geometric-mean inequality,
 \[
 \rho_*\le 2 \times \half \sqrt{\lamb_2^2+1-\lamb_2^2}=1
 \]
 where the equality holds only when $\lamb^2_2=1/{2}$. 
 
 As $\lamb^2_2$ tends to 1, $\rho_*$ tends to 0 and as $\lamb^2_2$ tends to $\half$, $\rho_*$ tends to 1.
 Recall that $\lamb_2^2+\lamb^2_{2n^2-1}=1$ and hence $ [1/2,1]$ is the proper range of $\lamb^2_2$.

\subsection{Perturbation analysis for Poisson-DRS}
The full analysis of the  Poisson-DRS \eqref{P1} is more challenging. Instead, we give a
perturbative derivation of  analogous result to Theorem \ref{thm:bounded} for the Poisson-DRS
with small positive $\rho$.

For small $\rho$, by keeping only the terms up to $\cO(\rho)$ we obtain
the perturbed DRS:
\beqn
u_{k+1}=\half u_k-\half(1-{\rho\over 2}) R_X u_k +P_Y R_X u_k.
\eeqn

Writing 
\[
I=P_X+P^\perp_X\quad\mbox{and}\quad R_X=P_X-P^\perp_X, 
\]
we then have the estimates
\beqn
 \|u_{k+1}\|
&\le& \|{\rho\over 4} P_X u_k+(1-{\rho\over 4})P^\perp_X u_k\|+\|P_Y R_X u_k\|\nn\\
&\le &(1-{\rho\over 4})\|u_k\|+\|b\| 
\eeqn
since $\rho$ is small. 
Iterating this bound, we obtain
\beqn
\|u_{k+1}\|\le (1-{\rho\over 4})^k\|u_1\|+\|b\|\sum_{j=0}^{k-1}(1-{\rho\over 4})^j
\eeqn
and hence
\beq
\label{881}
\limsup_{k\to\infty} \|u_k\|\le {4\over \rho} \|b\|. 
\eeq
Note that the small $\rho$ limit and the Poisson-to-Gaussian limit
do not commune, resulting in a different constant in \eqref{881} from Theorem \ref{thm:bounded}. 

\subsection{Noise-agnostic  method}\label{sec:agn}

In addition to  AAR, the {Relaxed Averaged Alternating Reflections (RAAR)} is another noise-agnostic method
which is formulated as the non-convex optimization problem
\beq
\min \|P_X^\perp z\|^2,\quad \mbox{subject to}\quad |z|=b\label{711}
\eeq
or equivalently \eqref{DRS} with the loss functions
\beq
\label{712}
K(y) =\half \|P_X^\perp y\|^2,\quad L(z)=\II_{b}(z)
\eeq
where   the hard constraint represented by the indicator function $\II_{b}$ of the set $\{z\in \IC^N: |z|=b\}$ is oblivious to the measurement noise while the choice of $K$ represents a relaxation of
the object domain constraint. 

If  the noisy phase retrieval problem is consistent, then the minimum value of \eqref{711} is zero and the minimizer is a regular solution 
(corresponding to the noisy data $b$). If the noisy problem is inconsistent, then the minimum value
of \eqref{711} is unknown and the minimizer  $z_*$ is the  generalized solution with the least  inconsistent component. In this case we can use $P_X z_*$ as the reconstruction. 

Let us apply  ADMM to the augmented Lagrangian function
\beq
\cL_\gamma(y,z,\lamb):= K(y) +L(z) +\lamb^*(z-y)+{\gamma\over 2}\|z-y\|^2\nn
\eeq
with $K$ and $L$ given in \eqref{712} in the order
\beq
\label{700}y_{k+1}&=&\arg\min_y\cL_\gamma(y,z_k,\lamb_k)\\
\label{701}z_{k+1}&=& \arg\min_{|z|=b}\cL_\gamma(y_{k+1}, z, \lamb_k)\\
\label{702}\lamb_{k+1}&=&\lamb_k+\gamma (z_{k+1}-y_{k+1}). 
\eeq

Solving \eqref{700} we have
\beq
y_{k+1}=\lt(I+P_X^\perp/\gamma\rt)^{-1}(z_k+\lamb_k/\gamma)=\lt(I-\beta P^\perp_X\rt) (z_k+\lamb_k/\gamma)\label{703}
\eeq
where
\beq
\label{beta}
 \beta:={1\over 1+\gamma}<1.
 \eeq

Likewise, solving \eqref{701}  we obtain
\beqn
z_{k+1}=P_Y u_{k+1},\quad u_{k+1}:= y_{k+1}-\lamb_k/\gamma
\eeqn
and hence by \eqref{702}, \eqref{703}
\beqn
u_{k+1}=(I-\beta \pxp) (\py u_k+\lamb_k/\gamma)-\lamb_k/\gamma.
\eeqn
On the other hand, we can rewrite \eqref{702} as
\[
\lamb_{k}/\gamma=z_{k}-u_{k}=P_Y u_{k}-u_{k}
\]
and hence
\beqn
u_{k+1}&=&(I-\beta \pxp) \py u_k-\beta\pxp \lamb_k/\gamma\\
&=& (I-\beta \pxp) \py u_k+\beta\pxp (I-\py) u_k\nn
\eeqn
which after reorganization becomes 
\beq
\label{raar}
u_{k+1}={T}_\beta (u_k):=\beta \lt(\half I+\half R_XR_Y\rt)u_k+(1-\beta)P_Y u_k.
\eeq

The scheme \eqref{raar} {resembles} the  RAAR method
first proposed in  \cite{Luke}, \cite{Luke2} {and  formulated in the object domain} from  a different perspective.  RAAR becomes AAR for $\beta=1$ (obviously) and AP for $\beta=\half$ (after some algebra). 

Let us  demonstrate again  that properly formulated DRS method can also lead to RAAR. 
Let us apply \eqref{dr2}-\eqref{dr2'} to
\eqref{DRS}  in the order
\beq
\label{715}
z_{k+1}&=&\prox_{L/\gamma} (u_k)= P_Y u_k\\
\label{716}y_{k+1}&=& \prox_{K/\gamma} (2z_{k+1}-u_k)= (I-\beta \pxp)(2 \py u_k-u_k)\\
u_{k+1}&=& u_k+y_{k+1}-z_{k+1}.\label{717'}
\eeq
Substituting \eqref{715} and \eqref{716} into \eqref{717'} we obtain after straightforward algebra 
the RAAR map \eqref{raar}.

With the splitting $I$ and $R_X$ as 
\[
I=P_X +P^\perp_X \quad\mbox{and}\quad R_X =P_X-P^\perp_X, 
\]
the fixed point equation $u={T}_\beta(u)$ becomes
\beqn
P_X u+P_X^\perp u=\beta P_X^\perp u+\lt[P_X+(1-2\beta)P_X^\perp\rt] P_Y u
\eeqn
from which it follows that 
\beq
\px u= \px\py u, && \pxp u=\lt({1-2\beta\over 1-\beta}\rt) \pxp\py u. \nn
\eeq
and hence
\beq
\px u-\lt({1-\beta\over 2\beta-1}\rt) \pxp u&=& \px\py u+\pxp\py u=\py u.\label{411}
\eeq
If the fixed point satisfies  $\pxp u=0$, then \eqref{411} implies  
\[
u=\px u=\py u=b\odot\sgn(u)
\]
i.e. $u$ is a regular solution. 

Notably \eqref{411} is exactly the RAAR fixed point equation \eqref{sa'} with the corresponding parameter
\beq
\rho&=& {1-\beta\over 2\beta-1} \in [0,\infty) \label{beta2}
\eeq
which tends to 0 and $\infty$ as $\beta$ tends to 1 and $\half$, respectively.

Local geometric convergence of RAAR has been proved in \cite{raar17}. Moreover, 
like Theorem \ref{thm:bounded} RAAR possesses the desirable property that every RAAR sequence is explicitly 
bounded  in terms of $\beta$ as follows. 
\begin{thm}Let $\{u_k\}$ be an RAAR-iterated sequence. Then
\beq\label{bound2}
\limsup_{k\to\infty} \|u_k\|\le {\|b\|\over 1-\beta}.
\eeq

Let $u$ be an RAAR fixed point. Then
\beq
\label{bound2'} \|u\|
&\le&\|b\| \times \lt\{\begin{matrix}
{2\beta-1\over 1-\beta}&&\mbox{for $\beta\in [2/3, 1)$}\\
1&&\mbox{for $\beta \in [1/2, 2/3]$}
\end{matrix}\rt.
\eeq

\end{thm}
\begin{proof}
For $\beta\in [\half, 1)$, $2\beta-1\in [0,1)$ and hence we have
\beqn
\|u_{k+1}\|&\le& \beta \|u_k\|+\|P_Yu_k\|\\
&=& \beta\|u_k\|+\|b\|. 
\eeqn
Iterating the above equation, we obtain
\beqn
\|u_{k+1}\|&\le& \beta^k\|u_1\|+\|b\|\sum_{j=0}^{k-1} \beta^j
\eeqn
and conclude \eqref{bound2}.

From \eqref{411} it follows that 
\beqn
\|u\|&\le& \max \lt({2\beta-1\over 1-\beta}, 1\rt)\|\py u\|\\
&\le& \max \lt({2\beta-1\over 1-\beta}, 1\rt)\|b\|
\eeqn
and hence \eqref{bound2'}.
\end{proof}

\subsection{Optimal parameter}
We briefly explore the optimal parameter for Gaussian-DRS \eqref{G1} in view of the optimal convergence rate \eqref{450}. 

Our test image  is 256-by-256 Cameraman+ $\im$ Barbara (CiB). 

We use three baseline algorithms as benchmark. The first two are AAR and RAAR. 
The third is Gaussian-DRS with $\rho=1$:  
\beq
\Gamma_1(u)&=&\half u+\half P_YR_X u\label{Gdrs}
\eeq
given the basic guarantee that for $\rho\ge 0$ the regular solutions are   attracting (Theorem \ref{thm:stable2}), that  for the range $\rho\ge 1$ no fixed points other than the regular solution(s) are locally attracting (Theorem \ref{thm:stable}) and that 
Gaussian-DRS with $\rho=1$ produces the best convergence rate for any $\rho \ge 1$ (Corollary \eqref{cor1}).
 The contrast between \eqref{Gdrs}  and AAR \eqref{dr} is noteworthy. 
The simplicity of the form \eqref{Gdrs} suggests the name {\em Averaged Projection Reflection} (APR) algorithm. 

\label{sec:num2}
\begin{figure}[t]
\centering
\subfigure[]{\includegraphics[width=5cm]{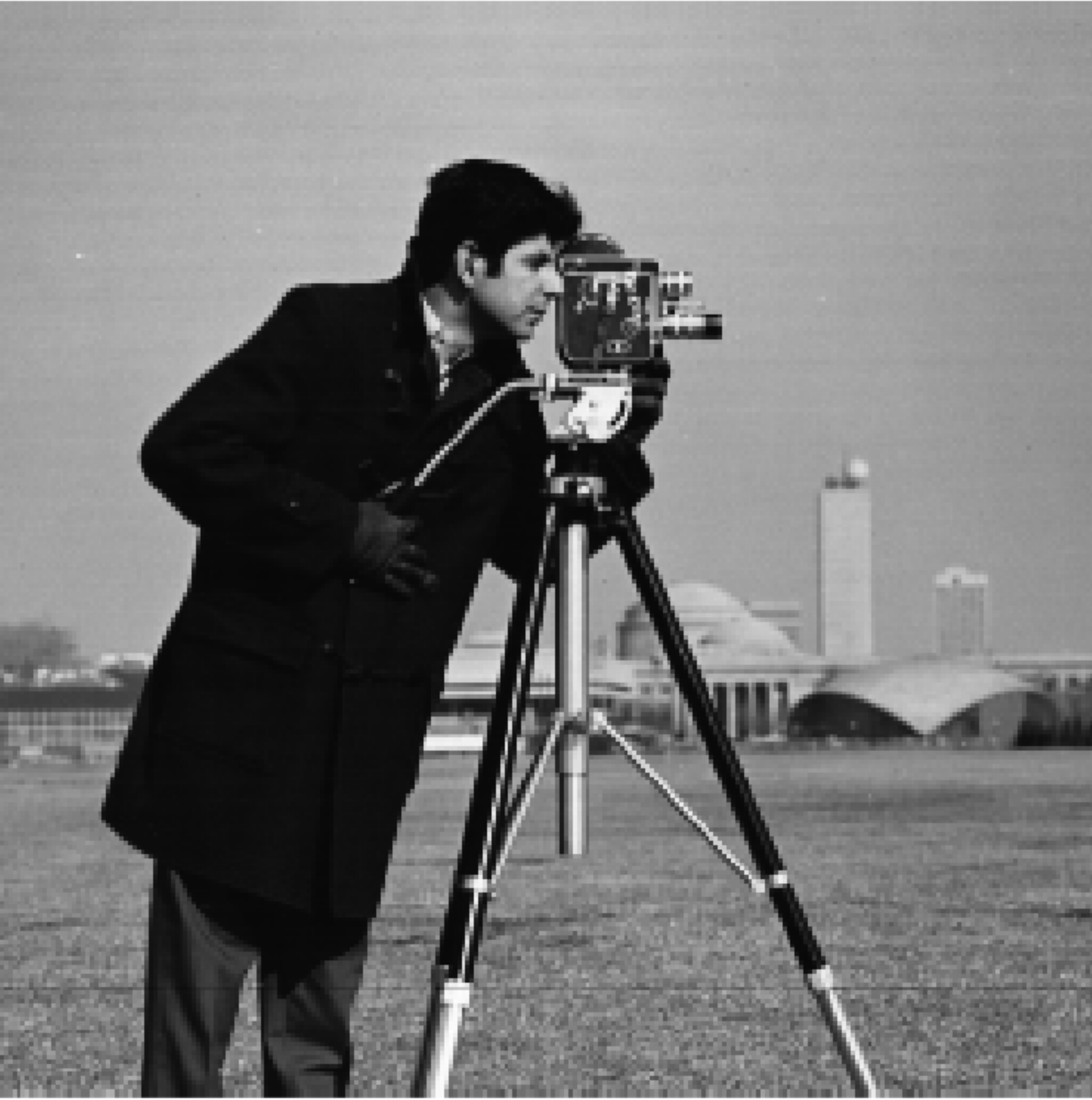}}\hspace{1cm}
\subfigure[]{\includegraphics[width=5cm]{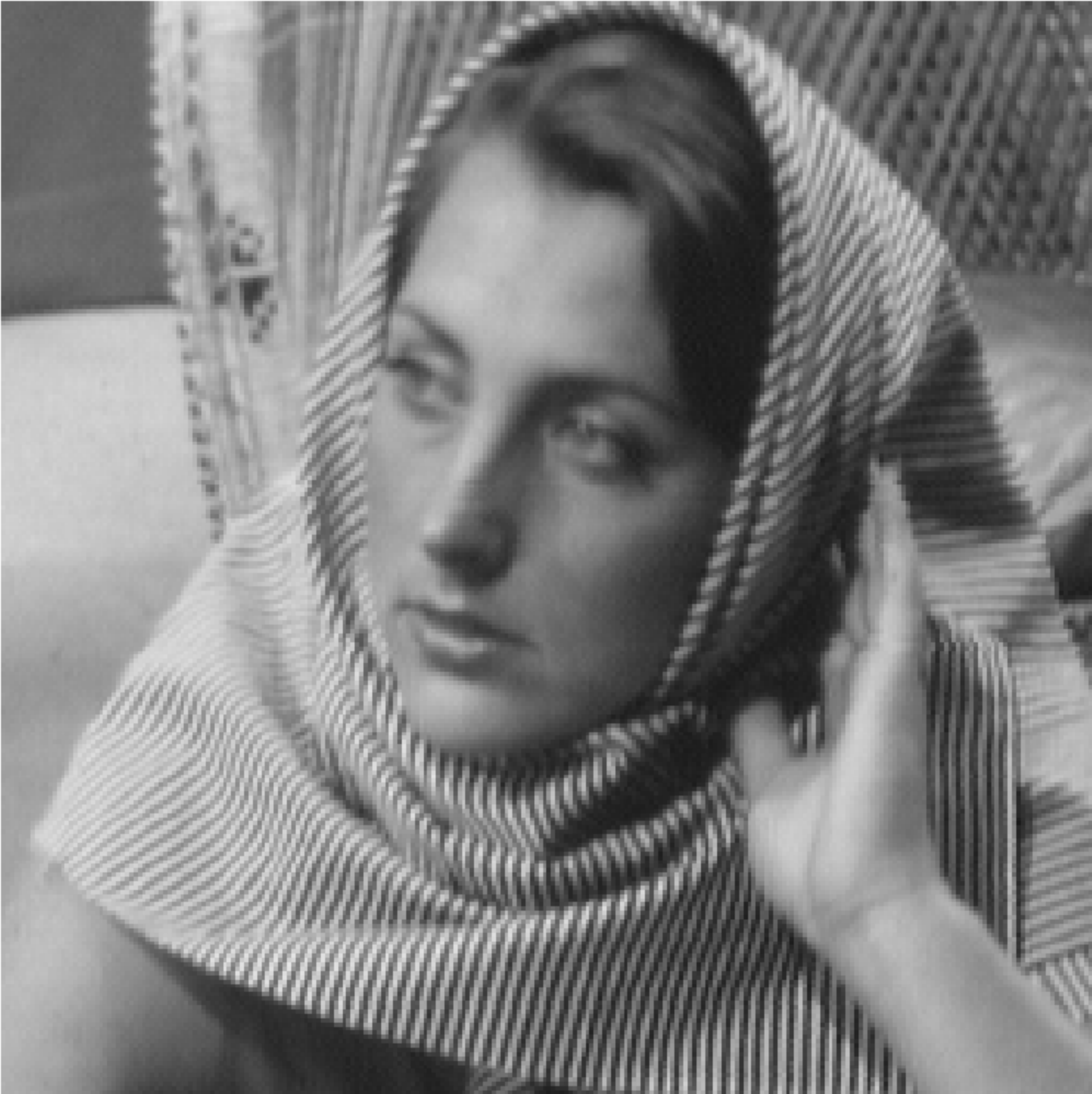}}
  \caption{(a) The real part and (b) the imaginary part of the test image $256\times 256$ CiB.}
  \label{fig:image}
\end{figure}%

\begin{figure}
\centering
\subfigure[$\rho=1.1,\beta=0.9$]{\includegraphics[width=6cm]{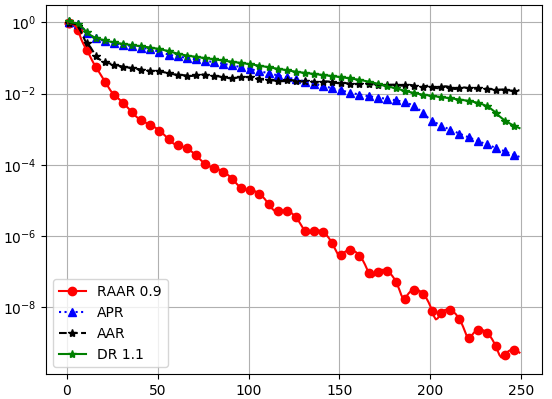}}
\subfigure[$\rho=0.5,\beta=0.9$]{\includegraphics[width=6cm]{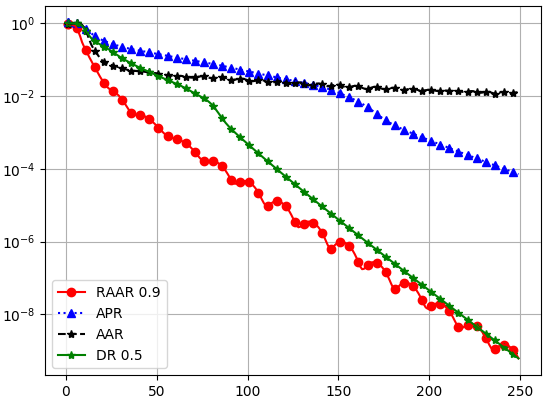}}
\subfigure[$\rho=0.3,\beta=0.9$]{\includegraphics[width=6cm]{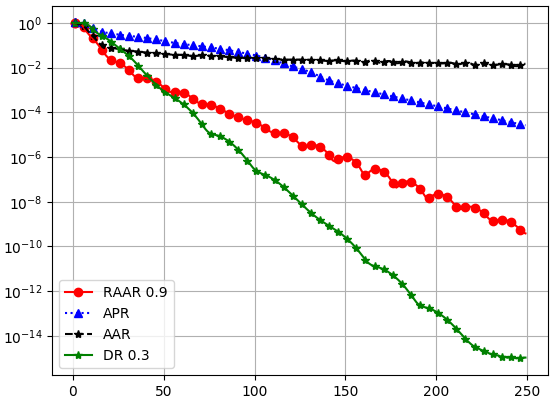}}
\subfigure[$\rho=0.1,\beta=0.9$]{\includegraphics[width=6cm]{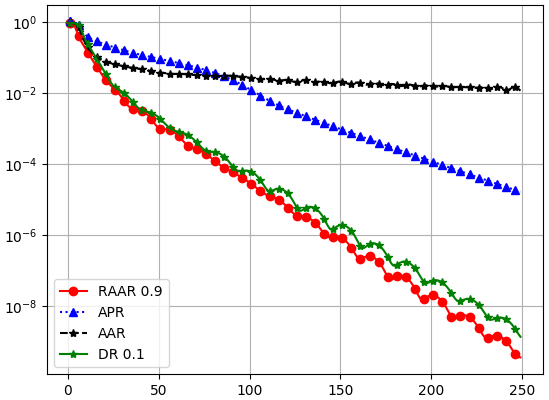}}
 \caption{Reconstruction (relative) error vs. iteration by various methods indicated in the legend with random initialization. 
 The straight-line feature (in all but AAR) in the semi-log plot indicates geometric convergence. 
 }
 \label{fig:raar}
 \end{figure}

According to \cite{raar17} the optimal $\beta$ is usually between 0.8 and 0.9,  corresponding to $\rho=0.125$ and $0.333$ according to \eqref{beta2}. We set  $\beta=0.9$ in Figure \ref{fig:raar}. 

 In the experiments, we consider the setting of non-ptychographic phase retrieval with two coded diffraction patterns, one is the plane wave ($\mu=1$) and the other is $\mu=\exp(\im \theta)$ where $ \theta$  is independent  and uniformly distributed over $[0,2\pi)$. Theory of uniqueness of solution, up to a constant phase factor, is given in \cite{unique}.

 Figure \ref{fig:raar} shows the relative error (modulo a constant phase factor) versus iteration of RAAR ($\beta=0.9$ round-bullet solid line), APR (blue-triangle dotted line), AAR
 (black-star dashed line) and Gaussian-DRS with (a) $\rho=1.1, $ (b) $\rho=0.5,$ (c) $\rho=0.3$ and (d) $\rho=0.1$. 
 Note that the AAR, APR and RAAR lines vary slightly across different plots because of random initialization. 
 
 The straight-line feature (in all but AAR) in the semi-log plot indicates global geometric convergence. 
 The case with AAR is less clear in Figure \ref{fig:raar}. But it has been shown that the AAR sequence converges geometrically near the true object  (after applying $A^+$) but converges in power-law ($ \sim k^{-\alpha}$ with $\alpha\in [1,2]$) from random initialization \cite{FDR}.

 Figure \ref{fig:raar} shows  that  APR outperforms  AAR but underperforms RAAR. By decreasing $\rho$ to either $0.5$ or $0.1$, the performance of Gaussian-DRS closely matches that of RAAR. The optimal parameter
 appears to lie between  $0.1$ and $0.5$. For example, with $\rho=0.3,$ Gaussian-DRS significantly outperforms RAAR.  
The oscillatory behavior of Gaussian-DRS in (d) is due to the dominant complex eigenvalue of $J$.

\section{Initialization strategies}\label{s:init}

Initialization is an important part of non-convex optimization to avoid local minima. 
Good initialization can also help to reduce the number of iterations of iterative solvers for convex optimization problems.
A simple idea for effective initialization is to first capture basic features of the original object. 
There are three tasks we want a good initializer to fulfill: (i)  it should ensure that the algorithm converges to the correct solution; (ii) it should reduce the number of iterations; and (iii) it should be inexpensive to compute. Naturally, there will be a trade-off between achieving the first two tasks  and task~(iii).

\subsection{Spectral initialization}\label{ss:spectral}

Spectral initialization~\cite{candes2015phase} has become a popular means in phase retrieval, bilinear compressive sensing, matrix completion, and related areas. In a nutshell, one chooses the leading eigenvector of the positive semidefinite Hermitian matrix 
\begin{equation}\label{bmatrix}
Y:= \sum_k y_k a_k a_k^{\ast} = A^{\ast} \diag(y) A
\end{equation}
as initializer.  The leading eigenvector of $Y$
can be computed efficiently via the power method by repeatedly applying $A$, entrywise multiplication by $y$ and $A^{\ast}$.

To give an intuitive explanation for this choice, consider the case in which the measurement vectors $a_k$ are i.i.d.\ ${\mathcal N}(0,I_n)$.
Let $x$ be a solution to~eqref{eq:data} so that $y_k =  |\langle x, a_k \rangle |^2$ for $k=1,\dots,N$.
In the Gaussian model, a simple moment calculation gives
\begin{equation}\nn
\EE\left[ \frac{1}{N} \sum_{k=1}^m y_k a_k a_k^{\ast} \right] = I_n + 2xx^{\ast}.
\end{equation}
By the strong law of large numbers, the matrix $Y = \sum_k y_k a_k a_k^{\ast}$  converges to the right-hand side as the number of samples goes to infinity. Since any leading eigenvector of $I_n + 2xx^{\ast}$ is of the form  $\lambda x$ for some 
$\lambda \in \IR$,  it follows that if we had infinitely many samples, this spectral initialization would recover $x$ exactly (up to a usual global phase factor). Moreover, the ratio between the top two eigenvalues of   $I_n + 2xx^{\ast}$ is $1+2 \|x\|_2^2$, which means these eigenvalues are
well separated unless $\|x\|_2$ is very small. This in turn implies that the power method would converge fast.
For a finite amount of measurements, the leading eigenvector of $Y$ will of course not recover $x$ exactly, but with the power of {\em concentration of measure} on our side, we can hope that the resulting (properly normalized) eigenvector will serve as good initial guess to the true solution.
This is made precise in connection with Wirtinger Flow in Theorem~\ref{th:WT}.

There is a nice connection between the spectral initialization and the PhaseLift approach, which will become evident in
Section~\ref{s:convex}.

\subsection{Null initialization}\label{ss:null}

Another approach to construct an effective initializer proceeds by  choosing  a threshold  for separating the ``weak"
 signals from the ``strong" signals. The classification of signals into the class of weak signals and the class of
 strong signals is a basic feature of the data.
 
Let $I\subset \{1,\cdots,N\}$ be the support set of the weak signals and $I_c$ its complement such that $b(i)\leq b(j)$ for all $i\in I, j\in I_c$.  In other words, $\{b(i): i\in I_c\}$ are the strong signals. Denote  the sub-row  matrices  consisting of $\{a_i\}_{i\in I} $ and $\{a_j\}_{j\in I_c}$  by $A_I$ and $A_{I_c}$, respectively. Let $b_I=|A_I\x|$ and $b_{I_c}=|A_{I_c}\x|$. We always assume $|I|\ge n$ so that $A_I$ has a trivial null space and hence preserves the information of $\x$.

The significance of the weak signal support $I$ lies in the fact that $I$ contains the
best loci  to ``linearize" the problem since 
$A^*_I \x$ is small. We then  initialize the object estimate  by the
 {\em ground state} of the sub-row matrix $A_I$, i.e. 
the following variational principle
  \beq
\label{nul'}
x_{\rm null}\in \hbox{\rm arg}\min\lt\{\|A_I x\|^2: x\in \IC^{n}, {\|x\|=\|b\|}\rt\}
 \eeq
 which by the isometric property of $A$ is equivalent to
   \beq
\label{strong}
x_{\rm null}\in \hbox{\rm arg}\max\lt\{\|A_{I_c} x\|^2: x\in \IC^{n}, {\|x\|=\|b\|}\rt\}. 
 \eeq
Note that \eqref{strong} can be solved by the power method for finding the leading singular value.
The resulting initial estimate $x_{\rm null}$ is called the null vector \cite{null}, \cite{AP-phasing} (see \cite{TAF18} for
the similar idea for real-valued Gaussian matrices). 

In the case of non-blind ptychography,  for each diffraction pattern $k$, the ``weak signals" are those less than some chosen threshold $\tau_k$ and we collect  the corresponding indices in the set $I_k$. Let $I=\cup_k I_k$. 
We then  initialize the object estimate  by  the variational principle \eqref{nul'} or \eqref{strong}.

A key question then is how to choose the threshold for separating weak from strong signals?
The following performance guarantee provides a guideline for choosing the threshold.

\begin{thm}\label{thm:null} \cite{null}
 Let $A$ be an $N\times n$ i.i.d.\ complex Gaussian matrix and 
let 
  \beq
\label{nul3}
 \xi_{\rm null}\in \hbox{\rm arg}\min\lt\{\|A_I x\|^2: x\in \IC^{n}, {\|x\|=\|\x\|}\rt\}. 
 \eeq

Let  $ \ep:={|I|/N}<1,\quad {|I|>n}.
$
 Then for any $\x\in \IC^n$  the error bound 
\beq
\label{error}
\|\x\x^* - \xi_{\rm null} \xi_{\rm null}^*\|_{\rm F}/\|\x\|^2
&\le& c_0\sqrt{\ep} 
 \eeq
holds with probability at least $1-5\exp\left(-c_1{|I|^2/N}\rt) -  4\exp(-c_2 n)$.
Here $\|\cdot\|_{\rm F}$ denotes the Frobenius norm.
\end{thm}

 By Theorem~\ref{thm:null}, we have that,   for $N=Cn\ln n$ and $|I|=Cn, C>1$, 
\[
\|\x\|^{-2}\|\x\x^* - \xi_{\rm null} \xi_{\rm null}^*\|_{\rm F}\le {c\over \sqrt{\ln n}}
\]
with probability exponentially (in $n$) close to one, implying crude reconstruction from
one-bit intensity measurement is easy. 
Theorem~\ref{thm:null} also gives a simple guideline
\[
n< |I|\ll N\ll |I|^2
\]
 for the choice of $|I|$ (and hence the intensity threshold)
to achieve a small $\ep$ with high probability. 
In particular, the choice 
\beq
\label{16'}
|I|=\lceil n^{1-\alpha}N^\alpha\rceil=\lceil n \delta^{\alpha}\rceil,\quad \alpha\in [0.5, 1)
\eeq
yields   the (relative) error  bound $\cO(\delta^{(\alpha-1)/2})$,
with probability exponentially (in $n$) close to 1, achieving 
the asymptotic minimum at $\alpha=1/2$ (the geometric mean rule). 
The geometric mean rule will be used in the numerical experiments below. 

Given the wide range of effective thresholds, the null vector is robust because the noise tends to mess up primarily the indices near
the threshold and can be compensated by choosing a smaller 
$I$, unspoiled by noise and thus satisfying the error bound \eqref{error}.

For null vector initialization with a non-isometric matrix such as the Gaussian random matrix in Theorem \ref{thm:null}, 
it is better to first perform QR factorization of $A$, instead of computing \eqref{nul3}, as follows.  

For a full rank $A\in \IC^{N\times n}$, let $A=QR$ be  the QR-decomposition of $A$ where
$Q$ is isometric and $R$ is an invertible upper-triangular square matrix.  
Let $Q_I$ and $Q_{I_c}$ be the sub-row  matrices  of $Q$  corresponding to the index sets $I$ and $I_c$, respectively.  Clearly, $A_I=Q_IR$ and $A_{I_c}=Q_{I_c}R$. 

Let $z_0=R\x$. Since $b_I=|Q_Iz_0|$ is small, the rows of $Q_I$ are nearly orthogonal  to $z_0$. 
A first approximation can be obtained from  $x_{\rm null}=R^{-1} z_{\rm null}$
where 
\beqn
\label{nul}
 z_{\rm null}\in \hbox{\rm arg}\min\lt\{\|Q_I z\|^2: z\in \IC^n, {\|z\|=\|b\|}\rt\}.  
 \eeqn
In view of the isometry property \beqn
 \label{isom'}
\|z\|^2= \|Q_I z \|^2+\|Q_{I_c} z \|^2=\|b\|^2
\eeqn
minimizing $ \|Q_I z \|^2$ is equivalent to maximizing $\|Q_{I_c}z\|^2$
over $\{z:\|z\|=\|b\|\}$. This leads to the alternative variational principle
  \beq
\label{strong2}
x_{\rm null}\in \hbox{\rm arg}\max\lt\{\|A_{I_c} x\|^2: x\in \IC^n, {\|Rx\|=\|b\|}\rt\}
 \eeq
solvable by 
 the power method.

The initial estimate $ \xi_{\rm null}$ in \eqref{nul3}  is close  to $x_{\rm null}$ in \eqref{strong2} when the oversampling ratio $\delta=N/n$ of the i.i.d.\ Gaussian matrix is large
or  when the measurement matrix is isometric ($R=I$) as for the coded Fourier matrix. 
 Numerical experiments show that
 $ \xi_{\rm null}$  is close to $x_{\rm null}$
 for $\delta\ge 8$. But for $\delta=4$, $x_{\rm null}$ is a significantly better approximation than $ \xi_{\rm null}$. 
 Note that  $\delta=4$ is near the threshold of having an injective 
intensity map: $x\longrightarrow |A x|^2$ for  a {\em generic} (i.e. random) $A$ \cite{Balan1}.

\subsection{Optimal pre-processing}\label{ss:optimal}

In both null and spectral initializations,  the  estimate $\xh$ is given by the principal eigenvector of a suitable {\em positive-definite} matrix
constructed from $A$ and $b$.
In the case
of spectral initialization,  an asymptotically exact recovery is guaranteed; in the case of null initialization, a non-asymptotic error bound exists and guarantees asymptotically exact recovery.

Contrary to these, the weak recovery problem of finding an estimate $\xh$ that has a positive correlation with $\x$:
\beq
\liminf_{N\to\infty} \IE\lt\{{|\xh^*\x|\over \|\x\|\|\xh\|}\rt\}> \ep\quad\mbox{for some}\quad \ep>0,\label{weak-rec}
\eeq
is analyzed in \cite{Mont},\cite{Lu17,luo2019optimal}. The fundamental interest with the weak recovery problem lies in the phase transition phenomenon
stated below. 
\begin{thm}\label{thm:Mont}
Let $\x$ be uniformly distributed on the $n$-dimensional complex sphere with radius $\sqrt{n}$ and let 
the rows of $A\in \IC^{N\times n}$ be i.i.d. complex circularly symmetric Gaussian vectors of covariance $I_n/n$. 
Let \beq
\label{MM'}
\yeta=|A\x|^2+\eta
\eeq
 where $\eta$ is real-valued Gaussian vector of covariance $\sigma^2 I_N$ and let
$N,n\to\infty$ with $N/n\to \delta\in (0,\infty)$.  

\begin{itemize}
\item For $\delta<1$, no algorithm can provide non-trivial estimates on $\x$;
\item For $\delta>1$, there exists $\sigma_0(\delta)>0$ and a spectral algorithm that returns an estimate $\xh$ satisfying \eqref{weak-rec} for
any $\sigma\in [0,\sigma_0(\delta)]$. 
\end{itemize}
\end{thm}

Like spectral initialization, weak recovery theory considers spectral algorithm of computing the principal eigenvalue of 
$A^* T A$ where $T$ is a pre-processing diagonal matrix. 
An important discovery of \cite{Mont} is that by removing the positivity assumption
$T>0$ and allowing negative values, an explicit recipe for $T$ is given and shown to be optimal in the sense that it provides the smallest possible threshold $\delta_u$ for the signal model \eqref{MM'}. Specifically, with vanishing noise $\sigma\to 0$, the threshold $\delta_u$ tends
to 1 as 
\beqn
\delta_u(\sigma^2)=1+\sigma^2+ o(\sigma^2)
\eeqn
 and
the optimal function is given by 
\beq
T_{\rm op}(\yeta,\delta)&=& {\yeta_+-1\over \yeta_++\sqrt{\delta}-1},\quad \yeta_+=\max(0,\yeta)\label{optimal}
\eeq
which has a large negative part for small $\yeta$ \cite{Mont}. This counterintuitive feature tends to slow down convergence of the power
 method as the principal eigenvalue of $A^*TA$ may not have the largest modulus, see \cite{Mont} for more details.

\subsection{Random initialization}
\label{ss:random}

While the aforementioned initializations are computationally quite efficient, one may wonder if such carefully designed initialization
is even necessary for achieving convergence for non-convex algorithms or to reduce the number of iterations for iterative solvers of convex approaches.  In particular, random initialization has been proposed as a cheap alternative to the more costly initialization strategies described above. In this case we simply construct a random signal in $\CC^n$, for instance with i.i.d.\ entries chosen from  ${\mathcal N}(0,I_n)$, and use it as initialization.

For non-convex solvers, we clearly cannot expect in general that starting the iterations at an arbitrary point will work, since we may get stuck in a saddle point or some local minimum. But if the optimization landscape is benign enough, it may be that there are no undesirable local extrema or that they can be easily avoided. A very thorough study of the optimization landscape of phase retrieval has been conducted in~\cite{sun2018geometric,chen2019gradient,Mont}.  

For instance, it has been shown  in~\cite{chen2019gradient} that for Gaussian measurements, gradient descent combined with random initialization will converge to the true solution and at a favorable rate of convergence, assuming that the number of measurements satisfies $N \gtrsim n \polylog N$. This result may suggest that random initialization is just fine and there is no need for more advanced initializations. 
The precise theoretical condition for  $N$ is $N \gtrsim n \log^{13} N$.  This large exponent in the log-factor becomes negligible if $n$ is in the order of at least, say, $10^{25}$, which makes this result somewhat less compelling from a theoretical viewpoint. However,  it is likely that this large exponent can be attributed to technical challenges in the proof and in truth it is actually much smaller. This is also suggested by the numerical simulations conducted in Section~\ref{ss:compinit}.

\subsection{Comparison of initializations}
\label{ss:compinit}

We conduct an empirical study by comparing  the effectiveness of different initializations. 

First we present experiments comparing the performance of the null initialization and the optimal pre-processing
 methods for noiseless as well as noisy data, see Figure \ref{fig:init1} and \ref{fig:init2}. While the optimal pre-processing function has no adjustable parameter,
we use the default threshold $|I|=\sqrt{Nn}$ for the null initialization ($\alpha=\half$ in \eqref{16'}). 

In the noisy case, we consider the complex Gaussian noise model \eqref{ray} which sits between the Poisson noise and the thermal noise in some sense. The nature of
 noise is unimportant for the comparison but the level of noise is. We consider three different levels of noise (0\%, 10\% and 20\%)
 as measured by the noise-to-signal ratio (NSR) defined as
\beq
\label{nsr}
{\rm NSR}= \frac{\|b-|A\x|\|}{\|A\x\|}. 
\eeq
 Because the noise dimension $N$ is larger than that of the object dimension,
the feasibility problem is  inconsistent with high probability.

   \begin{figure}[t]
\centering
\includegraphics[width=15cm]{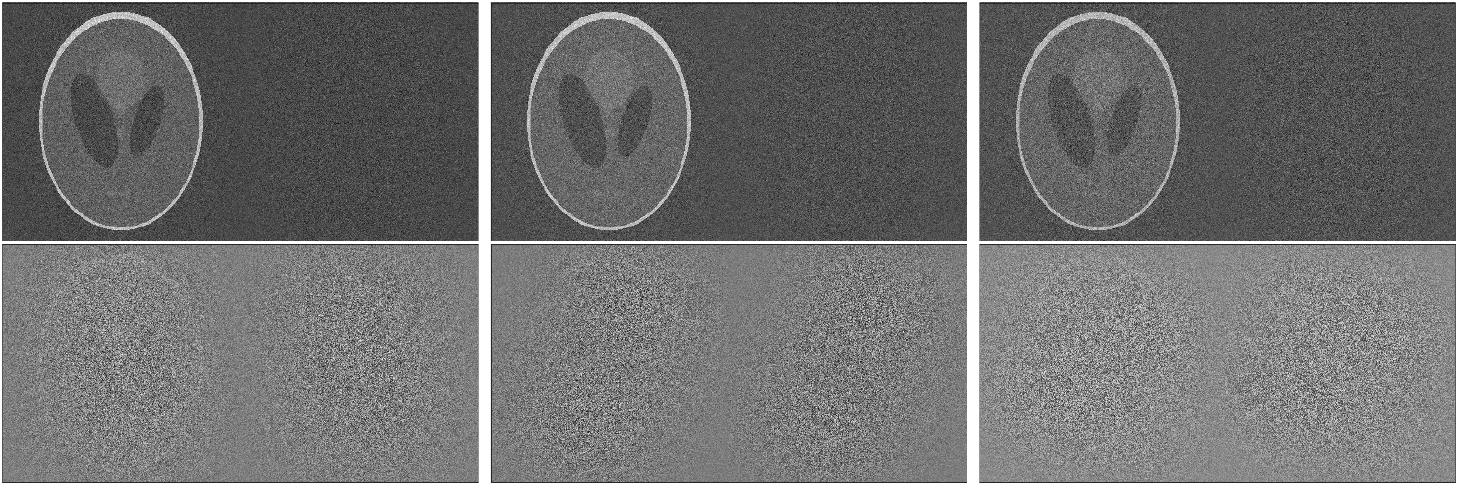}
\caption{Initialization for RPP with 2 OCDPs at NSR 0\% (left), 10\% (middle) and 20\% (right). Each panel shows 
$|\Re[\overline{x}\odot\sgn(\x)]|$ (left half) and $ |\Im[\overline{x}\odot\sgn(\x)]|$ (right half) where $x=x_{\rm null}$ (top row), or $x_{\rm op}$ (bottom row).}
\label{fig:init1}
\end{figure}

 Figure \ref{fig:init1} shows the results with 2 oversampled randomly coded diffraction patterns (OCDPs). 
 Hence $\delta=8$ for the optimal pre-processing function \eqref{optimal} and the outcome is denoted by 
 $x_{\rm op}$. We see that $x_{\rm null}$ significantly outperforms
 $x_{\rm op}$, consistent with the relative errors shown in the following table: 
 
 \begin{center}
 \begin{tabular}[c]{||c||c|c|c||}
 \hline
 {\rm 2 OCDPs @ NSR}& 0\%& 10\%&20\%\\
 \hline
 $x_{\rm null}$ &0.6531&0.6943&0.8146\\
 $x_{\rm op}$ &1.3636&1.3952&1.3889 \\
 \hline
 \end{tabular}
 \end{center}
 Here the optimal pre-processing method returns an essentially random output all noise levels.
This is somewhat surprising since the null vector uses only 1-bit information (the threshold) compared to
the optimal pre-processing function \eqref{optimal} which uses the full information of the
signals.

 \begin{figure}[h]
\centering
\includegraphics[width=15cm]{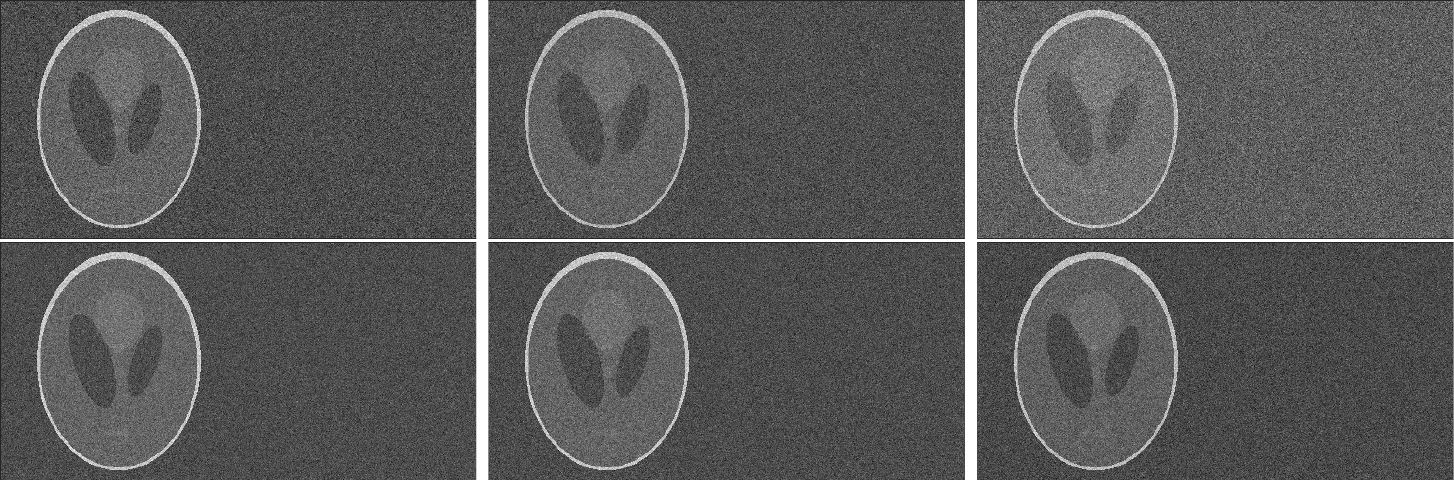}
\caption{Initialization for RPP with 4 CDPs  at NSR 0\% (left), 10\% (middle) and 20\% (right). Each panel shows $|\Re[\overline{x}\odot\sgn(\x)]|$ (left half) and $ |\Im[\overline{x}\odot\sgn(\x)]|$ (right half) where $x=x_{\rm null}$ (top row), or $x_{\rm op}$ (bottom row).}
\label{fig:init2}
\end{figure}

 On the other hand, with 4 randomly coded diffraction patterns (CDPs) that are not oversampled
 ($\delta=4 $ for \eqref{optimal}),  $x_{\rm op}$ outperforms $x_{\rm null}$ especially at large NSR, see Figure \ref{fig:init2}
for the visual effect and the following table for relative errors of initialization:
 
 \begin{center}
 \begin{tabular}[c]{||c||c|c|c||}
 \hline
 {\rm 4 CDPs @ NSR}& 0\%& 10\%&20\%\\
 \hline
 $x_{\rm null}$ & 0.7374&0.7761& 0.8991\\
 $x_{\rm op}$ & 0.6269& 0.6437& 0.6888\\
 \hline
 \end{tabular}
 \end{center}
The important lesson here is that the null vector and the optimal pre-processing function make use
of differently
sampled CDPs in different ways: the oversampled CDPs favor the former while the standard CDPs
favor the latter. { In particular, the optimal spectral method \eqref{optimal} is optimized for {\em independent } measurements and
does not perform well with highly correlated data in oversampled CDPs (Figure \ref{fig:init1}). As pointed out by \cite{Mont}, the performance
of \eqref{optimal} can often be improved by manually setting $\delta$ very close to 1. }

What follows are more simulations with higher number of CDPs that are not oversampled, for various initialization methods. 
We analyze their performance with respect to three different aspects: (i) number of measurements; (ii) number if iterations, (iii) overall runtime.
The initializers under comparison are the standard spectral initializer,  the truncated spectral initializer introduced in~\cite{chen2017solving}, the optimal spectral initializer, the null initializer (sometimes also referred to as ``orthogonality-promoting'' initializer), and random initialization. The computational complexity of constructing each of the first four initializers is roughly similar; they all require the computation of the leading eigenvector of a self-adjoint matrix associated with the measurement vectors $a_k$, which can be done efficiently with the power method (the matrix itself does not have be constructed explicitly). 

We choose a complex-valued Gaussian random signal of length $n=128$ as ground truth and obtain phaseless measurements with $k$ diffraction illuminations, where $k=3,\dots, 12$. Thus the number $N$ of phaseless measurements ranges from $3n$ to $12n$. The signal has no structural properties that we can take advantage of, e.g. we cannot exploit any support constraints.
We use the PhasePack toolbox~\cite{chandra2017phasepack}  with its default settings for this simulation, except for the threshold for the null initialization we use  $|I|=\lceil \sqrt{nN}\rceil$, as suggested by Theorem~\ref{thm:null}.

We run Wirtinger Flow with different initializations until the residual error is smaller than $10^{-4}$.  
For each $k=3,\dots,12$ and each fixed choice of signal and illuminations we repeat the experiment 100 times, and do so for 100 different random choices of signal and illuminations. For each $k$ the results are then averaged over these 10000 runs. For each number of illuminations, 
we compare the number of iterations as well as the overall runtime of the algorithm needed to achieve the desired residual error. 
We also compare the rate of successful recovery, where success is  (generously) defined as the case when the algorithm returns a solution with an relative $\ell_2$ error less than 0.1. A success rate of 1 means that the algorithm succeeded in all simulations for a fixed number of illuminations.

\begin{figure}
\begin{center}
\subfigure[]{
\includegraphics[width=60mm,height=40mm]{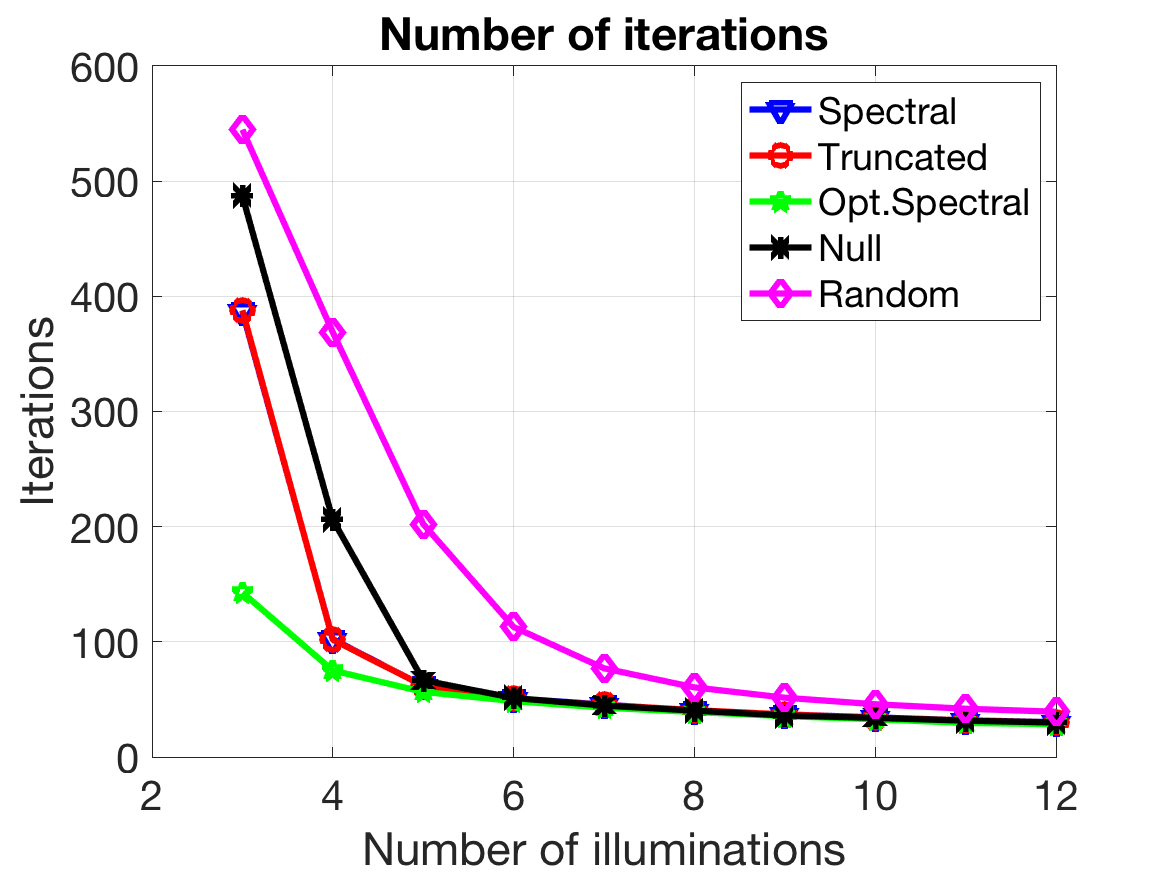}}
\subfigure[]{
\includegraphics[width=60mm,height=40mm]{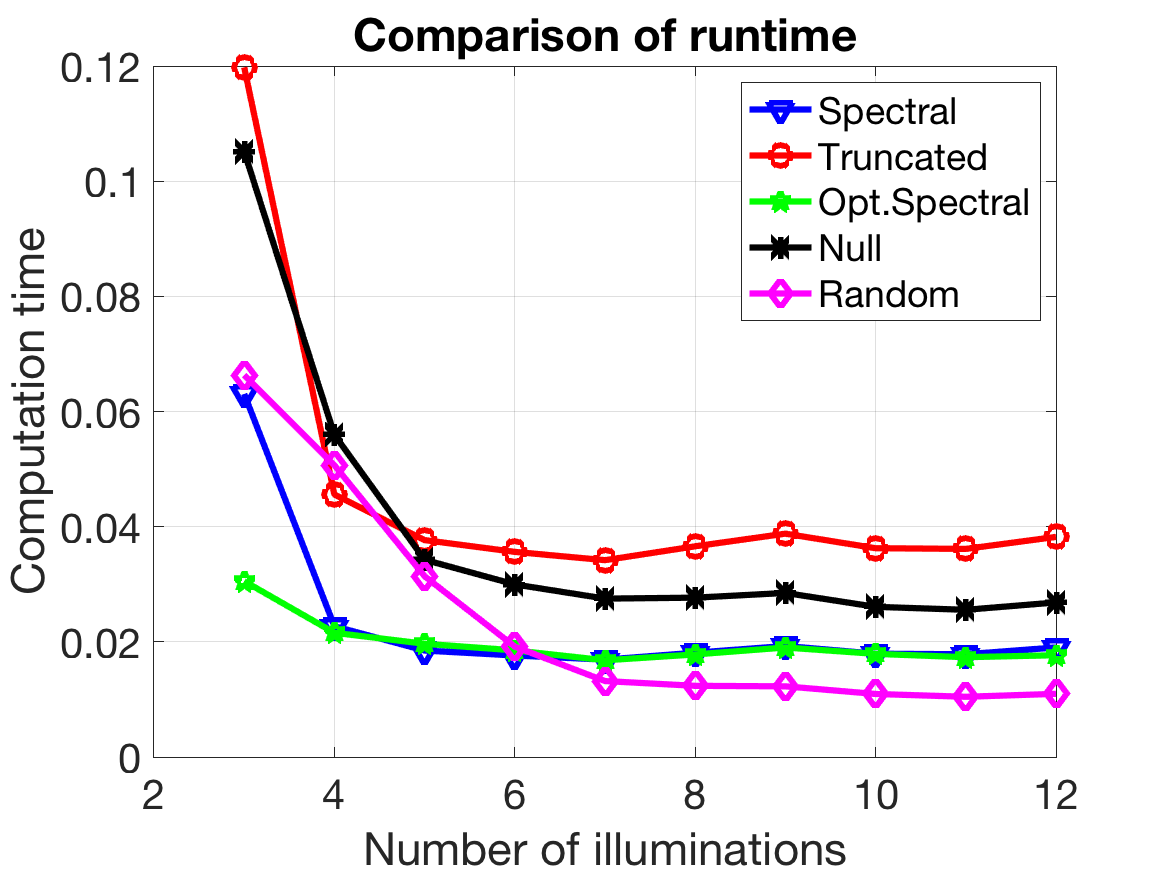}}
\caption{The initializers under comparison are the standard-,  the truncated-, and the optimal spectral initializer, the ``orthonality-promoting'' initializer, and random initalization. We run Wirtinger Flow with different initializations and compare (a) the number of iterations, (b) the total computation time needed for Wirtinger Flow to achieve a residual error less than $10^{-4}$.}
\label{fig:compinit1}
\end{center}
\end{figure}

\begin{figure}
\begin{center}
\includegraphics[width=60mm,height=40mm]{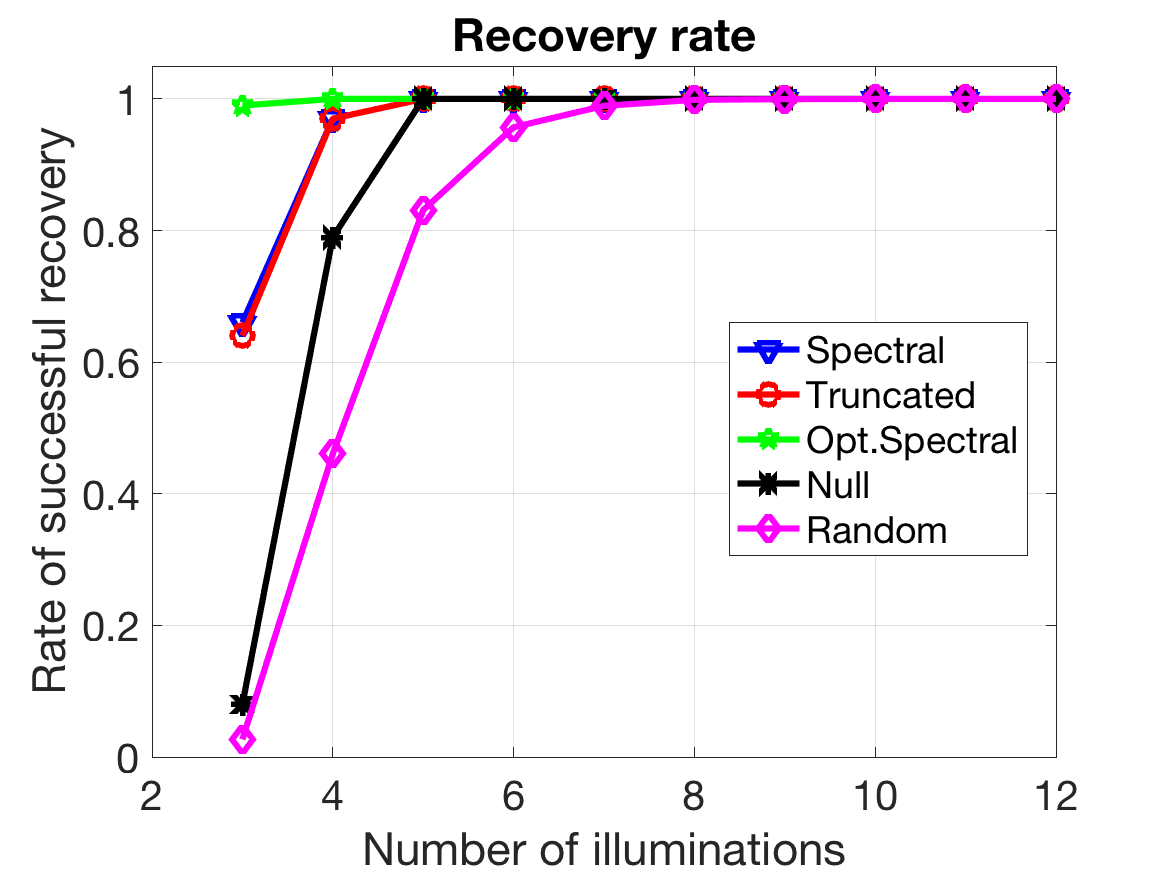}
\caption{Same setup as in Figure~\ref{fig:compinit1}. We compare the success rate for Wirtinger Flow with different initializations. For this experiment, a ``successful recovery'' means that the algorithm returns a solution with a relative $\ell_2$ error less than 0.1. A success rate of 1 means that the algorithm succeeded in all simulations. The optimal spectral initialization clearly outperforms all other initializations when the number of measurements is small.}
\label{fig:compinit2}
\end{center}
\end{figure}

The most relevant and important case from a practical viewpoint is when the required number of illuminations is as small as possible, as this reduced the experimental burden.  The clear winner in this case is the optimal spectral initializer. When we use only three illuminations, it significantly outperforms all the other initializers.
In general, for the recovery of a complex-valued signal of length $n$ from phaseless measurements, we cannot expect that any method can succeed at a perfect rate when we use only $N=3n$ measurements,

The exact number of measurements {\em necessary} to make recovery of a signal $x \in \IR^n$ from phaseless measurements  at least theoretically possible (setting aside the existence of a feasible algorithm and issues of numerical stability) is $n \ge 2n-1$. For complex-valued signals the precise lower bound is still open. The asymptotic estimate $N =(4+o(1))n$ follows from~\cite{heinosaari2013quantum,BCE07}, see also~\cite{bandeira2014saving}. For dimensions $n=2^k =1$ it has been shown in~\cite{conca2015algebraic} that $N=4n-4$ is necessary\footnote{However, this is not true for all $n$. In~\cite{vinzant2015small} Vinzant  gave an example of a frame with $4n-5=11$ elements in $\CC^4$ which enables phase retrieval.}.
In general, for the recovery of a complex-valued signal of length $n$ from phaseless measurements is
$4n-4$, we cannot expect that any method can succeed at a perfect rate when we use only $N=3n$ measurements,

As the  number of illuminations increases, the difference becomes less pronounced which is in line with theoretical predictions.
For a moderate number of illuminations the random initializer performs as well as the others, at a lower computational cost. As expected the theory for random initialization (which involves the term $\log^{13} N$) is overly pessimistic. Nevertheless, in practice there can be a substantial difference in the experimental effort if we need to carry, say, six illuminations instead of just three illuminations. Hence, we conclude that 
{\em ``there is no free lunch with random initialization!''}

\section{Convex optimization}
\label{s:convex}

While phase retrieval is a non-convex optimization problem, it has become very popular in recent years to pursue convex relaxations of this problem. A major breakthrough in this context was the PhaseLift approach~\cite{CSV2013,CESV2013} which demonstrated that under fairly mild conditions the solution of a properly constructed semidefinite program coincides with the true solution of the original non-convex problem. This discovery has ignited a renewed interest in the phase retrieval problem.
We will describe the key idea of PhaseLift below.

\subsection{PhaseLift: Phase retrieval via matrix completion}\label{ss:phaselift}

As is well known, quadratic measurements can be lifted up and
interpreted as linear measurements about the rank-one matrix $X = x
x^*$. Indeed,
\begin{equation}\label{tracemeasurements}
|\langle a_k, x\rangle |^2 = \trace(x^* a_k a_k^* x) = \trace(a_k a_k^* x x^*).
\end{equation}
We write $\Hn$ for the Hilbert space of all $n \times n$ Hermitian matrices equipped
with the Hilbert-Schmidt inner product $\langle X,Y \rangle_{{\text{HS}}} := {\trace}(Y^\ast X)$
Now,
letting $\cA$ be the linear transformation
\begin{equation}
\label{linmap}
\begin{array}{lll}
  \Hn & \rightarrow & \IR^N\\
  X & \mapsto & \{a_k a_k^* X \}_{1 \le i \le N}
\end{array}
\end{equation}
which maps Hermitian matrices into real-valued vectors, one can
express the data collection $b_k = |\langle x, a_k\rangle |^2$ as
\begin{equation}
\nn
y = \mathcal{A}(x x^*).
\end{equation}
For reference, the adjoint operator $\cA^*$ maps real-valued inputs
into Hermitian matrices, and is given by
\[
\begin{array}{lll}
  \IR^{N} & \to & \mathcal{H}^{n \times n}\\
  z & \mapsto & \sum_i z_i \, a_k a_k^*.
\end{array}
\]
Moreover, we define $\cT_{x}$ to be the set of symmetric matrices of the
form
\begin{equation}
\nn
\cT_{x} = \{X = x z^* + z x^*: z \in \CC^n\}
\end{equation}
and denote $\cT_{x}^\perp$ by its orthogonal complement.  Note that
$X \in \cT_{x}^\perp$ if and only if both the column and row spaces
of $X$ are perpendicular to $x$.

Hence, the phase retrieval problem can be cast as the matrix recovery problem~\cite{CSV2013,CESV2013}
\begin{equation}
\nn
  \begin{array}{ll}
    \text{minimize}   & \quad \rank(X)\\
    \text{subject to} & \quad  \cA(X) = y \\
& \quad X \succeq 0.
\end{array}
\end{equation}
Indeed, we know that a rank-one solution exists so the optimal $X$ has
rank at most one. We then factorize the solution as $x x^*$ in order
to obtain solutions to the phase-retrieval problem. This gives $x$ up
to multiplication by a unit-normed scalar.

Rank minimization is in general NP hard, and we propose, instead,
solving a trace-norm relaxation. Although this is a fairly standard
relaxation in control \cite{Beck98,Mesbahi97}, the idea of casting the
phase retrieval problem as a trace-minimization problem over an
affine slice of the positive semidefinite cone is more recent\footnote{This idea was first proposed by one of the authors at a workshop
``Frames for the finite world: Sampling, coding and quantization'' at the American Institute of Mathematics in August 2008.}.
Formally, we suggest solving
\begin{equation}
\label{eq:tracemin}
 \begin{array}{ll}
    \text{minimize}   & \quad \trace(X)\\
    \text{subject to} & \quad  \cA(X) = y\\
& \quad X \succeq 0.
\end{array}
\end{equation}
If the solution has rank one, we factorize it as above to recover our
signal. This method which lifts up the problem of vector recovery from
quadratic constraints into that of recovering a rank-one matrix from
affine constraints via semidefinite programming is known under the
name of {\em PhaseLift} \cite{CSV2013,CESV2013}.

A sufficient (and nearly necessary) condition
for $x x^\ast$ to be the unique solution to~\eqref{eq:tracemin} is given by the following lemma.
\begin{lem}
\label{lemma:dual}
If for a given vector $x \in \CC^n$ the measurement mapping $\cA$ satisfies the following two conditions
\begin{itemize}
\item[(i)] the restriction of $\cA$ to $T$ is injective ($X \in T$ and $\cA(X) = 0 \Rightarrow X = 0$),
\item[(ii)] and there exists a {\em dual certificate} $Z$ in the range of
  $\cA^*$ obeying\footnote{The notation $A \prec B$ means that $B - A$ is positive definite.}
\begin{equation}
  \nn
  Z_T = x x^{\ast} \quad \text{and} \quad Z_{\Tp} \prec I_{\Tp}.
\end{equation}
\end{itemize}
then $X = x x^\ast$ is the only matrix in the feasible set of~\eqref{eq:tracemin}, i.e. $X$ is the unique solution of~\eqref{eq:tracemin}.
\end{lem}
The proof of Lemma~\ref{lemma:dual} follows from standard duality arguments in semidefinite programming.
\begin{proof}
Let $\tilde{X} = X+H$ be a matrix  in the feasible set of~\eqref{eq:tracemin}. We want to show that $H=0$.
By assumption $H \in \Hn$ and $H \in \null(\cA)$, hence we can express $H$ as $H  = H_\cT + H_\Tp$.
Since $\tilde X \prec 0$, it follows for all $z \in \CC^n$ with $\langle z,x\rangle = 0$ that
$$z^\ast \tilde{X} z = z^\ast (x x^{\ast} + H_\cT + H_\Tp)z = z^\ast  H_\Tp y \ge 0.$$
Because the range spaces of $H_\Tp$ and of $H^{\ast}_{\Tp}$ are contained in orthogonal complement
of $\overline{\text{span}}\{x\}$ this shows that $H_\Tp \prec 0$. Since $Z \in \range (\cA) = \null(\cA)^\perp$
it holds that $\langle H,Z \rangle = 0$ and because $Z_\cT = 0$, it follows that 
$\langle H,Z \rangle = \langle H_\Tp,Z_\Tp \rangle = 0$. 
But since $Z_\Tp \prec  0 $, this shows that $H_\Tp = 0$.  By injectivity of $\cA$ on $\cT$ we also have
$H_\cT = 0$, such that $H = 0$ and therefore $\tilde{X}  = X$.
\end{proof}


Asserting that the conditions of Lemma~\ref{lemma:dual} hold under reasonable conditions on the number of measurements is the real challenge here.  A careful strengthening of the injectivity property in~Lemma~\ref{lemma:dual} allows one to relax the properties of the dual certificate,  as in the approach pioneered in~\cite{gross2011recovering} for matrix completion. This observation is at the core of 
the proof of Theorem~\ref{theo:main} below. In a nutshell, the theorem states that under mild conditions  PhaseLift can recover $x$ exactly (up to a global phase factor) with high probability, provided that the
number of measurements is on the order of $n \log n$. 
\begin{thm}\cite{CSV2013}
\label{theo:main}
Consider an arbitrary signal $x$ in $\IR^n$ or $\CC^n$. Let the measurement vectors $a_k$  be sampled independently and uniformly at
random on the unit sphere, and suppose that the number of measurements obeys $N \ge c_0 \, n \log n$, where
$c_0$ is a sufficiently large constant. Then the solution to the trace-minimization program is exact
with high probability in the sense that \eqref{eq:tracemin} has a
unique solution obeying
\begin{equation}
  \nn
  \hat X = x x^\ast.
\end{equation}
This holds with probability at least $1 - 3e^{-\gamma \frac{m}{n}}$,
where $\gamma$ is a positive absolute constant.
\end{thm}

Theorem~\ref{theo:main} can be extended to noisy measurements, see~\cite{CSV2013,hand2017phaselift}, demonstrating that PhaseLift is robust visavis noise.  In~\cite{ImprovedPL}, the condition $m ={\mathcal O}( n \log n)$ was further improved  to $m = {\mathcal O}( n)$. As noted
in~\cite{ImprovedPL,demanet2014stable} under the conditions of Lemma~\ref{lemma:dual} the feasible set of~\eqref{eq:tracemin} reduces to the single point $X = xx^{\ast}$. Thus, from a purely theoretical viewpoint, the trace minimization  in~\eqref{eq:tracemin} is actually not necessary, while from a numerical viewpoint, in particular in the case of noisy data, using the program~\eqref{eq:tracemin} still seems beneficial.   

We also note that the spectral initialization of Section~\ref{ss:spectral} has a natural interpretation in the PhaseLift framework. Comparing equation~\eqref{bmatrix} with the definition of $\cA$ in~\eqref{linmap}, it is evident that the spectral initializer is simply given by
the solution extracted from computing $\cA^{\ast} y$.

Although PhaseLift favors low-rank solutions,  in particular in the case of noisy data it is not guaranteed to find a rank-one solution.  Therefore, if our optimal solution $\hat{X}$ does not have exactly rank one, we extract the rank-one approximation 
$\hat{x} \hat{x}^*$ where $\hat{x}$ is an eigenvector associated with the largest eigenvalue of $\hat{X}$. In that case one can further improve the accuracy of the solution $\hat{x}$ by ``debiasing'' it. We replace $\hat{x}$ by its
rescaled version $s \hat{x}$ where $s = \sqrt{\sum_{k=1}^{n}  \hat{\lambda}_k}/\|\hat{x}\|_2$.  This corrects for the energy
leakage occurring when $\hat{X}$ is not exactly a rank-1 solution, which could cause the norm of $\hat{x}$ to be smaller than that of
the actual solution. Other corrections are of course possible. 

\begin{rmk}
For the numerical solution of~\eqref{eq:tracemin} it is not necessary to actually set up the matrix $X$ explicitly. Indeed, this fact is already described in detail in~\cite{CESV2013}. Yet, the misconception that the full matrix $X$ needs to be computed and stored  can sometimes be found in the non-mathematical literature~\cite{elser2018benchmark}.
\end{rmk}

Theorem~\ref{theo:main} serves as a benchmark result, but using Gaussian vectors as measurement vectors $a_k$ is not very realistic. For practical purposes, we prefer sets of measurement vectors that obey e.g.\ the coded diffraction structure illustrated in~Figure~\ref{fig:mask}. The extension of PhaseLift to such more realistic conditions was first shown in~\cite{candes2015phase}, where  a  result similar to~Theorem~\ref{theo:main} was proven to also holds for Fourier type measurements when ${\mathcal O}(\log^4 n)$  different specifically designed random masks are employed. Thus, compared to Theorem~\ref{theo:main}, the total number of measurements increases to $N=  {\mathcal O}(n \log^4 n)$. This result was improved in~\cite{gross2017improved}, where the number of measurements was reduced to $ {\mathcal O}(n \log^2 n)$.
Since the coded diffraction approach is both mathematically appealing and relevant in practice, we describe a typical setup that is also the basis of~\cite{candes2015phase,gross2017improved} in more detail below.

We assume that we collect the magnitudes of the discrete Fourier transform of a random modulation of the unknown signal $x$. 
Each such modulation pattern represents one mask and is modeled by a random diagonal matrix. 
Let $\{e_1,\dots, e_n\}$ denothe the standard basis of $\CC^n$. We define the $\ell$-th  (coded diffraction) mask via
$$D_\ell = \sum_{i=1}^n \eps_{\ell,i} e_i e_i^\ast,$$
where the $\eps_{\ell,i}$ are independent copies of a real-valued random variable $\eps$ which obeys
\begin{align}
\EE[\eps] & = \EE[\eps^3] = 0 \notag  \\
|\eps | & \le b \quad \text{almost surely for some $b > 0$},  \label{coded} \\
\EE[\eps^4] & = 2 \EE[\eps^2]^2. \notag
\end{align}
Denote 
$$f_k = \sum_{j=1}^n e^{2\pi \im jk/n} e_j.$$
Then the measurements captured via this coded diffraction approach can be written as
\begin{equation}\label{codedy}
y_{k,\ell} = | \langle f_k, D_\ell x \rangle |^2, \quad k=1,\dots,n, \,\, \ell=1,\dots, L.
\end{equation}

As shown in~\cite{gross2017improved}, condition~\eqref{coded} ensures that the measurement ensemble
forms a spherical 2-design, a concept that has been proposed in connection with phase retrieval in~\cite{BBC09,gross2015partial}.
As a particular choice in~\eqref{coded} we may select each modulation to correspond to a Rademacher vector with random erasures, i.e.,
\begin{equation*}
\eps \sim \begin{cases} \sqrt{2}  & \text{with prob.~$1/4$,} \\
				     0  & \text{with prob.~$1/2$,} \\
				     -\sqrt{2}  & \text{with prob.~$1/4$,}
				     \end{cases}
\end{equation*}
as suggested in~\cite{candes2015phase}.

In the case of such coded diffraction measurements the following theorem, proved in~\cite{gross2017improved}, guarantees the success of PhaseLift with high probability (see also~\cite{candes2015phase}).
\begin{thm}
Let $x \in \CC^n$ with $\|x\|_2 = 1$ and let $n\ge 3$ be an odd number.
Suppose that $N=nL$  Fourier measurements using $L$ independent
random diffraction patterns (as defined in~\eqref{coded} and~\eqref{codedy}) are gathered.
Then, with probability at least  $1 - e^{-\omega}$,  PhaseLift  endowed with the additional constraint $\trace(X) = 1$
recovers $x$ up to a global phase, provided that
$$L \ge C \omega \log^2 n.$$
Here, $\omega \ge 1$  is an arbitrary parameter and $C$ a dimension-independent constant that can
be explicitly bounded.
\end{thm}

While the original PhaseLift approach works for multidimensional signals, there exist specific constructions of masks for the special case of 
one-dimensional signals that provide further improvements. For instance, in~\cite{Boche15}, the authors derive a deterministic, carefully designed set of $4n - 4$ measurement vectors and prove that a semidefinite program will successfully recover {\em generic} signals from the associated measurements. The authors accomplish this by showing that the conditions of Lemma~\ref{lemma:dual} hold on a dense subspace of $\CC^n$.
Another approach that combines the PhaseLift idea with the construction of a few specially designed one-dimensional masks can be found in~\cite{jaganathan2015phase}.


The PhaseCut method, proposed in~\cite{waldspurger2015phase}, casts the phase retrieval problem as an equality constrained quadratic program and then uses the famous MaxCut relaxation for this type of problem. Interestingly, while the PhaseCut and PhaseLift relaxations
are in general different, there is a striking equivalence between these two approaches, see~\cite{waldspurger2015phase}.

\medskip
Concerning the numerical solution of~\eqref{eq:tracemin}, there exists a wide array of fairly efficient numerical solvers, see e.g.~\cite{NesterovBook,toh1999sdpt3,monteiro1997primal}. The numerical algorithm to solve~\eqref{eq:tracemin} in the example illustrated in Figure~\ref{fig:goldballs} was implemented in Matlab using TFOCS~\cite{becker2011templates}. That implementation avoids setting up the matrix $X$ explicitly and only keeps an $n\times r$ matrix with $r \ll n$ in memory.
More custom-designed solvers have also been developed, see e.g.~\cite{huang2017solving}.

\subsection{Convex phase retrieval without lifting}

Despite its mathematical elegance, a significant drawback of PhaseLift is that its computational complexity is too high (even when $X$ is not set up explicitly) for large-scale problems. A different route to solve the phase retrieval problem via convex relaxation was pursued independently in~\cite{bahmani2016phase,goldstein2018phasemax}. Starting from our usual setup, assume we are given phaseless  measurements
\begin{equation}\label{phaseless2}
 |\langle a_k, x\rangle |^2 = y_k, \quad k=1,\dots,N.
 \end{equation}
We relax each measurement to an inequality
\begin{equation}
\label{slab}
   |\langle a_k, x\rangle | \le \sqrt{y_k}=b_k, \quad k=1,\dots,N.
\end{equation}
This creates a symmetric slab ${\mathcal S}_i$ of feasible solutions.
Collectively, these slabs describe a ``complex polytope'' ${\mathcal K}$ of feasible solutions. The target
signal $x$ is one of the extreme points of ${\mathcal K}$, as illustrated in Figure~\ref{fig:slab}.

\begin{figure}
\begin{center}
\includegraphics[width=50mm,height=50mm]{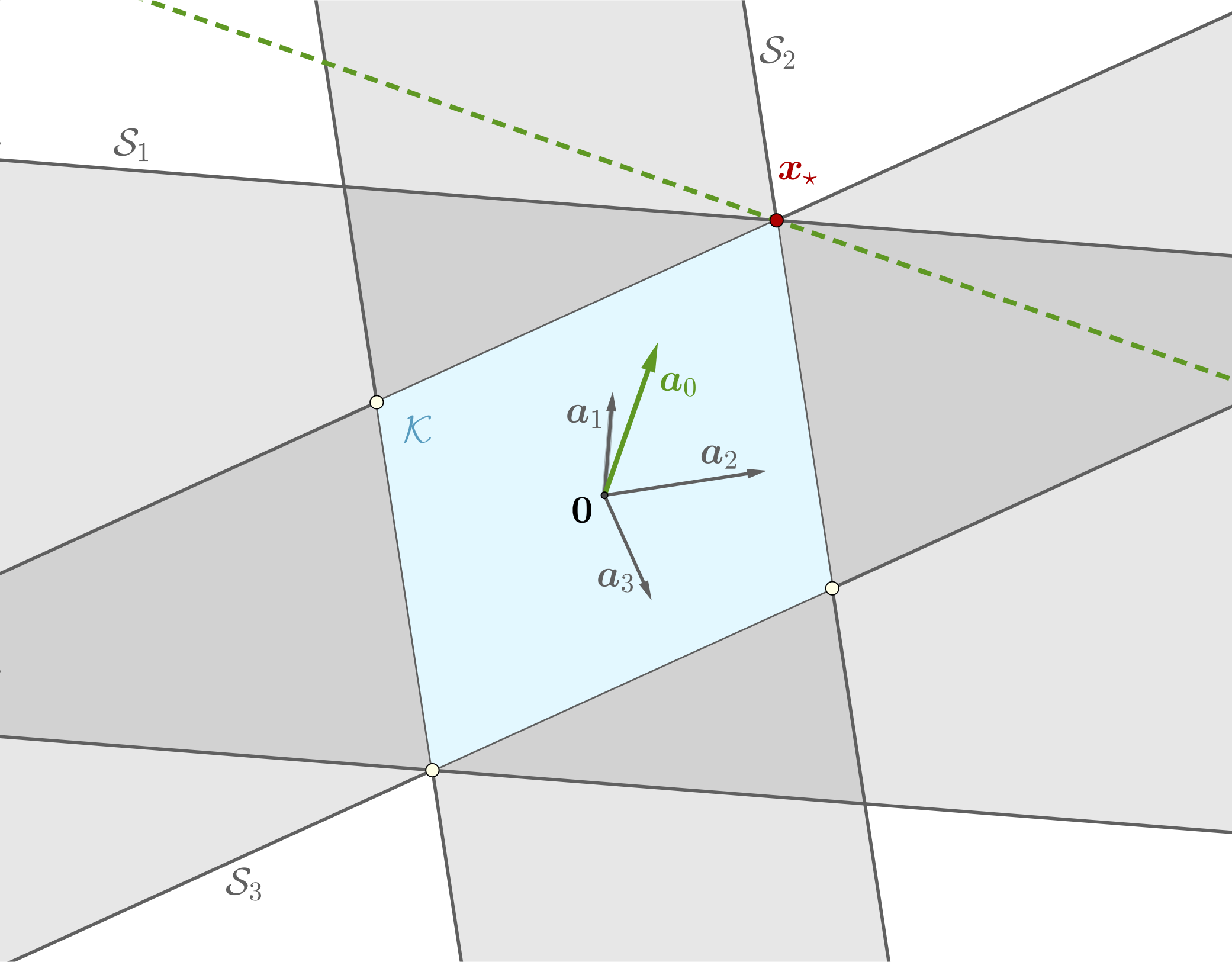}
\caption{The ``complex polytope'' of feasible solutions intersecting at $\xo=x_*$. Here, the role of the anchor vector $u$ is  played by $a_0$. Image courtesy of~\protect\cite{bahmani2016phase}.}
\label{fig:slab}
\end{center}
\end{figure}

How do we distinguish the desired solution $x$ from all the other extreme points of ${\mathcal K}$? The idea proposed in~\cite{bahmani2016phase,goldstein2018phasemax} is to use a (non-zero) ``anchor'' vector $u$ that is sufficiently close to $x$. 
Following~\cite{bahmani2016phase},
from a geometrical viewpoint, the idea is to find a hyperplane tangent to ${\mathcal K}$ at $x$ and the anchor vector $u$ acts as the normal for the desired tangent hyperplane see~Figure~\ref{fig:slab}; $u$ is required to have a non-vanishing correlation with $x$ in the sense that
\begin{equation}\label{slabcorr}
\frac{| \langle x,u \rangle |}{ \|u\|_2 \|x\|_2}  > \epsilon,
\end{equation}
for some  $\epsilon >0$. See also~\eqref{weak-rec} related to the optimal initializiation in Section~\ref{ss:optimal}.
The idea of~\cite{bahmani2016phase,goldstein2018phasemax} is now to recover $x$ by finding the vector that is most aligned with $u$ {\em and} satisfies the relaxed measurement constraints in~\eqref{slab}. 

This approach can be expressed as the following convex problem, dubbed {\em PhaseMax} in~\cite{goldstein2018phasemax}:
\begin{equation}
\label{eq:phasemax}
 \begin{array}{ll}
    \max\limits_{x}   & \quad \langle x,u \rangle \\
    \text{subject to} & \quad  b_k \le   |\langle a_k, x\rangle |^2 +\xi_k, \quad k=1,\dots,N.
\end{array}
\end{equation}
It is remarkable that this convex relaxation of the phase retrieval problem does not involve lifting and operates in the original parameter space.

Choosing an appropriate anchor vector $u$ is crucial, since $u$ must be sufficiently close to $x$. tIt has been shown in~\cite{bahmani2016phase} that under the assumptions of Theorem~\ref{theo:main}, the condition~\eqref{slabcorr} holds with probability at least $1-{\mathcal O}(n^{-2})$. The authors of~\cite{bahmani2016phase} then showed  that
the convex program in~\eqref{eq:phasemax}  can successfully recover the original signal from measurements of the form~\eqref{phaseless2}
under conditions similar to those in  Theorem~\ref{theo:main} (and under additional technical assumptions), and moreover that this recovery is robust in the presence of measurement noise. 
A slightly stronger result was proven in~\cite{hand2016elementary}. There, the authors established the following result:
\begin{thm}\cite{hand2016elementary}
Fix $x\in \IR^n$. Let $a_k$ be i.i.d~${\mathcal N}(0,I_{n})$  for $k = 1,\dots, N$. Let $|\langle a_k, x\rangle |^2 = y_k$. Assume that
$u \in \IR^n$ satisfies $\| u - x\|_2 \le 0.6 \|x\|_2$.  If $N \ge cn$, then with probability at least  $1-6e^{-\gamma N}$, $x$ is the
unique solution of the linear program PhaseMax. Here, $\gamma$ and $c$ are universal constants.
\end{thm}
Using  for instance the truncated spectral initialization proposed in~\cite{chen2017solving}, one can show that $\| u = x\|_2 \le 0.6 \|x\|_2$ holds 
with probability at least $1-e^{-\gamma N}$, provided that $N \ge c_0 n$.

In~\cite{dhifallah2017phase},  it was  shown that even  better signal recovery guarantees
can be achieved by iteratively applying PhaseMax. The resulting method is called PhaseLamp; the name derives 
from the fact that the algorithm is based on the idea of successive linearization and maximization over a polytope.

\medskip
Denote the $n \times N$ matrix $A =[a_1,\dots,a_N]$ and the  $N \times N$ diagonal matrix $B = \diag (b_1,\dots,b_N)$.
Then, as noticed in~\cite{goldstein2018phasemax} the following basis pursuit problem
\begin{equation}
\label{eq:bp}
 \begin{array}{ll}
    \min\limits_{z \in \CC^N}   & \quad \|z\|_1 \\
    \text{subject to} & \quad  u = A B^{-1}z.
\end{array}
\end{equation}
is dual to the convex program~\eqref{eq:phasemax}. Moreover, as pointed out in~\cite{goldstein2018phasemax}, as a
consequence, if PhaseMax succeeds, then the phases of the solution vector $z$ to~\eqref{eq:bp} are exactly the phases that
were lost in the measurement process in~\eqref{phaseless2}, that is
$$
\frac{z_k}{|z_k|} b_k  = \langle a_k, x\rangle, \quad k=1,\dots,N.
$$
These observation open up the possibility to utilize algorithms associated with basis pursuit for phase retrieval.

\bigskip
Yet another convex approach to phase retrieval has been proposed in~\cite{doelman2018solving}. There, the authors propose a sequence of convex relaxations, where the obtained convex problems are affine in the unknown signal $\xo$. No lifting is required in this approach. However, no theoretical conditions are provided (in terms of number of measurements or otherwise) that would ensure that the computed solution actually coincides with the true solution $\xo$.

\medskip

To illustrate the efficacy of the  approaches described in this section, we consider a stylized version of a setup one encounters in X-ray
crystallography or diffraction imaging.  The test image, shown in
Figure~\ref{fig:goldballs}(a) (magnitude), is a complex-valued image\footnote{Since the original image and the reconstruction are complex-valued, we only display the absolute value of each image.} of size
$256 \times 256$, whose pixel values correspond to the complex
transmission coefficients of a collection of gold balls at nanoscale embedded in a
medium (data courtesy of Stefano Marchesini from Lawrence Berkeley National Laboratory).

We demonstrate the recovery of the image shown in Figure~\ref{fig:goldballs}(a) from noiseless measurements via PhaseLift, PhaseMax, and PhaseLamp.
We use three coded diffraction illuminations, where the entries of the diffraction matrices  are either
$+1$ or $-1$ with equal probability. We use the TFOCS based implementation of PhaseLift from~\cite{CESV2013} with reweighting. 
For PhaseMax and PhaseLamp we use the implementations provided by PhasePack (cf.~\cite{chandra2017phasepack}) with the optimal spectral initializer and the default settings.
The reconstructions by PhaseLift and PhaseLamp, shown in Figure~\ref{fig:goldballs}(b) and Figure~\ref{fig:goldballs}(d) are visually indistinguishable from the original.  The reconstruction computed by PhaseMax, depicted in Figure~\ref{fig:goldballs}(c) is less accurate in this example.

\begin{figure}
\begin{center}
\subfigure[Original image]{
\includegraphics[width=5cm]{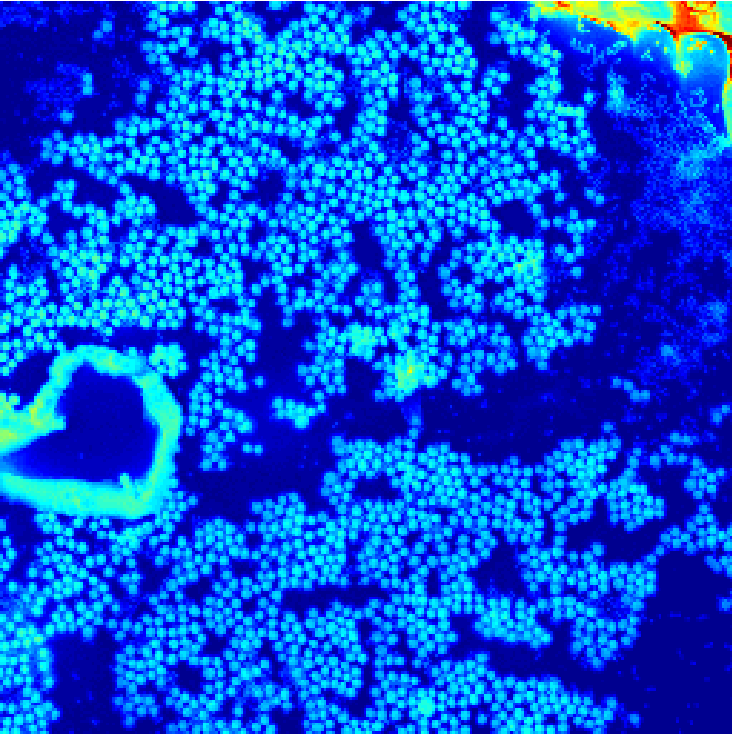}}
\qquad
\subfigure[Reconstruction via PhaseLift]{
\includegraphics[width=50mm,height=50mm]{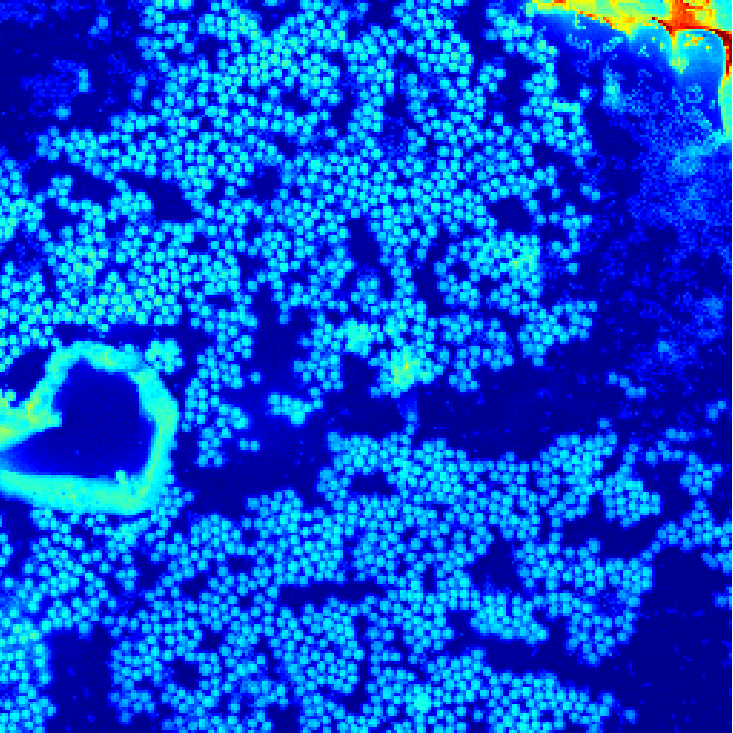}}\\
\smallskip
\subfigure[Reconstruction via PhaseMax]{
  \includegraphics[width=50mm,height=50mm]{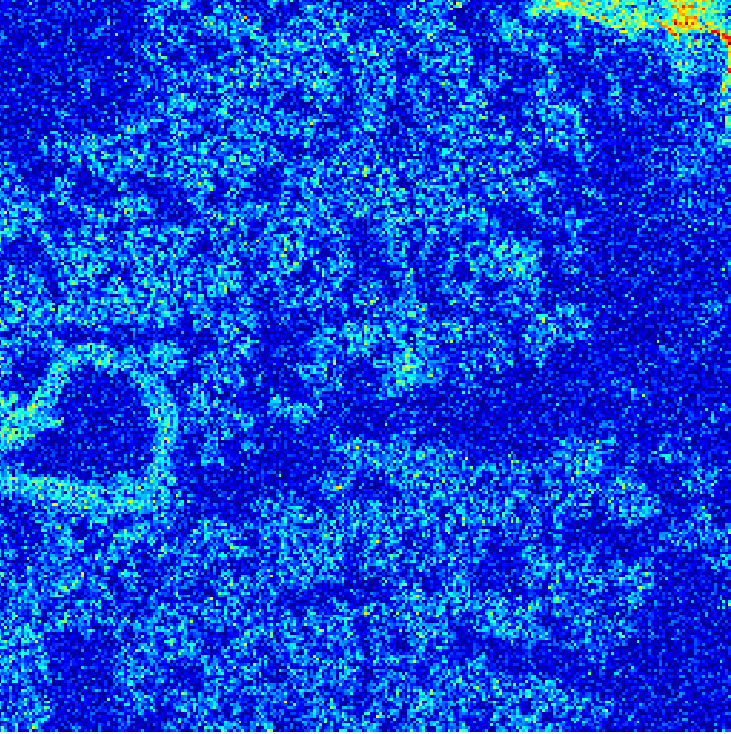}}
\qquad \subfigure[Reconstruction via PhaseLamp]{
\includegraphics[width=50mm,height=50mm]{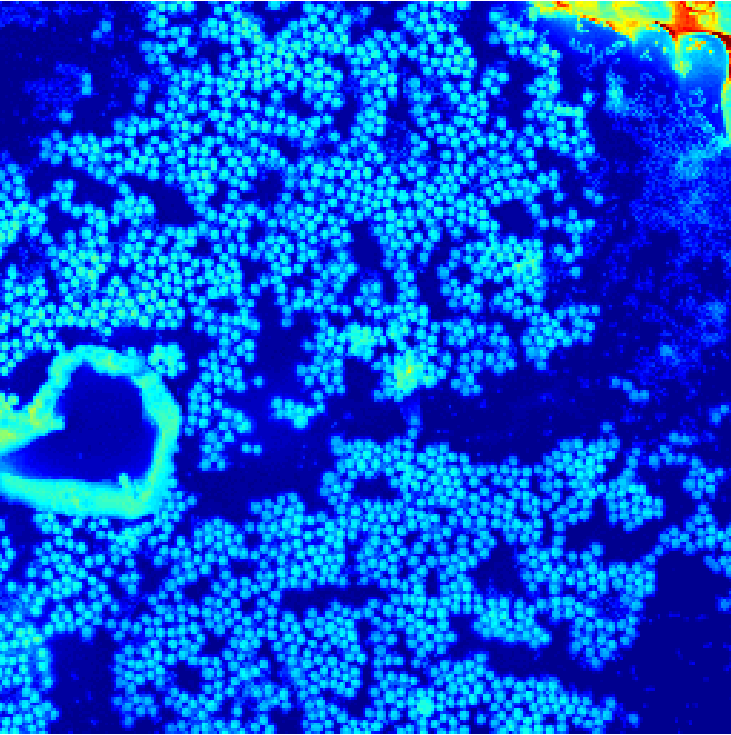}}
\caption{Original goldballs image and reconstructions via PhaseLift, PhaseMax, and PhaseLamp, using three coded diffraction illuminiations.}
\label{fig:goldballs}
\end{center}
\end{figure}

\medskip
Despite the ability of convex methods to recover signals from a small number of phaseless observations,  these methods have not found practical use yet. While there exist fast implementations of PhaseLift, in terms of computational efficiency it cannot compete with the nonconvex methods discussed in Section \ref{s:nonconvex}. The biggest impact PhaseLift has had on phase retrieval is on the one hand it triggered a broad and systematic study of numerical algorithms for phase retrieval, and on the other hand it ignited a sophisticated design of initializations for non-convex solvers.
Beyond phase retrieval, it ignited research in related areas, such as in bilinear compressive sensing~\cite{ling2015self}, including blind deconvolution~\cite{ahmed2013blind,li2016identifiability,krahmer2019complex} and blind demixing~\cite{ling2017blind}.
Moreover, the techniques behind PhaseLift and sparse recovery have influenced other areas directly related to phase retrieval, namely 
low-rank phase retrieval problems as they appear for instance in quantum tomography, as well as utilizing sparsity in phase retrieval.
We will discuss these topics in Sections~\ref{ss:lowrank} and~\ref{ss:sparsity} below.

\subsection{Low-rank phase retrieval problems}\label{ss:lowrank}

The phase retrieval problem has a natural generalization to recovering low-rank positive semidefinite matrices. Consider the problem of
recovering an unknown $n \times n$ rank-$r$ matrix $ \succeq 0$
from linear functionals of the form $y_k = \trace(A_k^{\ast} M)$ for $k=1,\dots,N$,
where $A$ is hermitian. By representing $M$  in factorized form, $M = X X^{\ast}, X \in \CC^{n\times r}$, we can express this problem
as the attempt to recover $ X \in \CC^{n\times r}$ from the measurements $y_k = \trace(A_k^{\ast} X X^{\ast})$, 
which, in light of~\eqref{tracemeasurements}, is a natural generalization of the phase retrieval problem. 

A particular instance of interest of this problem arises in {\em quantum  state tomography}, where one tries to
characterize the complete quantum state of a particle or particles through a series of measurements in different bases~\cite{paris2004quantum,haah2017sample}.
More precisely, we are concerned with the task of reconstructing a finite-dimensional quantum mechanical system which is fully characterized by its density operator $\rho$ -- an $n \times n$  positive semidefinite matrix with trace one. Estimating the density operator of an actual (finite dimensional) quantum system is an important task in quantum physics known as quantum state tomography. One is often interested in performing tomography for quantum systems that have certain structural properties.  One important structural property is {\em purity}. 
A pure quantum state of $n$ ions can be described by its $2^n \times 2^n$ rank-one density matrix.
A quantum state is almost pure if it is well approximated by a matrix of low rank $r$ with $r \ll n$. 

Assuming this structural property, quantum state tomography becomes a low-rank matrix recovery problem~\cite{gross2011recovering,recht2010guaranteed,kueng2017low,davenport2016overview}.  It is obvious that we can recover a general quantum state 
$\rho \in \CC^{n\times n}$ from $n(n-1)$ properly chosen measurements. But  if $\rho$ is low-rank, how many measurements are needed such that we can still recover $\rho$ in a numerical efficient manner? And what properties does measurement system have to satisfy?
An additional requirement is the fact that the measurement process has to be ``experimentally realizable'' and preferably in an efficient manner~\cite{kueng2017low}. Moreover, in a  real experiment, the measurements are noisy, and the true state is only approximately low-rank. Thus, any algorithm that aims to recover quantum states must be robust to these sources of error.

Many of the algorithms discussed in the previous sections can be extended with straightforward modifications to the generalized phase retrieval problem.  For example in~\cite{kueng2017low} it has been shown that  the PhaseLift results can be extended beyond the rank-one case: For Gaussian measurements the required number of measurements is $N \ge C nr$, which is analogous to the rank-one case. 

Perhaps more interestingly, and  similar in spirit to coded diffraction illuminations, there are certain structured measurement systems that are also realizable from an experimental viewpoint. 
For example, using the mathematically intriguing concept of Clifford orbits,  one can reconstruct  a rank-$r$ quantum state exactly in the noisefree case and robustly in the presence of noise if the measurement matrices are chosen independently and uniformly at random from the Clifford orbit, assuming the number of measurements satisfies $N \ge C r n  \log n$, see~\cite{kueng2016low}. Here, the noise can include  additive noise as well as ``model noise'' due to the state being not exactly of rank $r$.
It was shown in~\cite{kueng2016low} that  a similar result holds if we replace the measurement system by
approximate projective 4-designs (see~\cite{kueng2017low} for a precise definition). This line of research opens up beautiful connections to group theory, representation theory, and time-frequency analysis.

\medskip
{
We will demonstrate that the famous Zauner conjecture can  be expressed as a low-rank  phase retrieval problem. At the core of this conjecture is the problem of  finding a family of $n^2$ unit-length vectors $\{v_i\}_{i=1}^{n^2}$ in $\CC^n$ such that 
\begin{equation}\label{etf}
|\langle v_i , v_i' \rangle |^2 = \frac{1}{n+1}, \qquad \forall  i\neq i',
\end{equation}
see~\cite{zaunerquantendesigns}. Such a family constitutes an equiangular tight frame of maximal cardinality (since no more than $n^2$ lines in $\CC^n$ can be equiangular), also known as Grassmannian frame~\cite{SH03}. Equiangular tight frames play an important role in many applications, ranging from signal processing and communications to compressive sensing. In quantum physics~\cite{appleby2005symmetric} such a family of vectors is known as
symmetric informationally complete positive-operator-valued measure (SIC-POVM),~\cite{scott2010symmetric}.}

{
Zauner conjectured that for each $n=2,3,\dots,$ there exists a fiducial vector $v \in\CC^n$
such that the Weyl-Heisenberg (or Gabor) frame $\{T_j M_k v\}_{j,k=1}^n$ satisfies~\eqref{etf}. Moreover, Zauner conjectured that this fiducial vector $v\in\CC^n $ is  an eigenvector of a certain order-3 Clifford unitary ${\mathcal U}_{n}$. We refrain here from going into details about the Clifford group  and refer instead to~\cite{zaunerquantendesigns,appleby2005symmetric,fuchs2017sic}.  Putative fiducial vectors have been found (to machine precision) via computational techniques for every dimension $n$ up to 151, and for a handful of higher dimensions~\cite{fuchs2017sic}.
We also know analytic solutions for a few values of $n$, see e.g.~\cite{appleby2019tight,fuchs2017sic}.}

{
Note that
$ \langle T_j M_k x, T_{j'} M_{k'}x \rangle  =  e^{-2\pi \i (j-j')k'} \langle T_{j-j'} M_{k-k'} x, x \rangle $.
Hence, Zauner's conjecture can be expressed as solving the problem
\begin{equation}\label{zauner}
\text{Find $x \in {\mathcal U}_{n}$  \, s.t. \,\,} 
| \langle T_j M_k x, x \rangle |^2 = 
\begin{cases}  1 & \text{if $k=j=0$}, \\
                 \frac{1}{n+1} & \text{else.}
 \end{cases}
 \end{equation}
 This is a phase retrieval problem. Unfortunately, the unknown vector $x$ appears on both sides of the inner product.
Hence, while the measurement setup may seem similar to ptychography at first glance, the problem~\eqref{zauner} is actually more challenging. }

 {
To arrive at the promised low-rank  formulation, first note that the property $x \in {\mathcal U}_{n}$ can be expressed as $x=U_n z$, where $U_n$ is an $n \times d$ matrix and $z\in \CC^d$ with  $d = \lceil \frac{n+1}{3} \rceil$, see~\cite{scott2010symmetric}. 
Hence, for $x \in {\mathcal U}_{n}$ we obtain 
 $$ \langle T_j M_k x, x \rangle   =  \langle T_j M_k U_n z, U_n z \rangle  = \langle V_{jk} , Z \rangle_{\text{HS}},$$
 where $Z=z z^{\ast}$ and $V_{jk} = U_n^{\ast} T_j M_k U_n$  for $j,k=0,\dots,n-1$. 
 Thus, we arrive at our first low-rank phase retrieval version by rewriting~\eqref{zauner}  as 
  \begin{equation}\label{zaunerrank}
    \begin{array}{ll}
\text{Find}    & \quad Z \\
    \text{subject to} &  |  \langle V_{jk} , Z \rangle_{\text{HS}}  |^2 = 
       \begin{cases}  1 & \text{if $k=j=0$}, \\
                 \frac{1}{n+1} & \text{else,} 
 \end{cases} \\
 & \quad Z \succeq 0\\
& \quad \rank(Z) = 1.
\end{array}
 \end{equation}
In~\eqref{zaunerrank}  we have $n^2$ quadratic equations with about $(n/3)^2$ unknowns.
It is not difficult to devise a simple alternating projection algorithm  with random initialization to solve~\eqref{zauner} that works quite efficiently for $n <100$. However,  for larger $n$ the algorithm seems to get stuck in local minima. Maybe methods from {\em blind}  ptychography can  guide us to solve~\eqref{zauner} numerically for larger $n$.}

 {We can lift the equations in~\eqref{zaunerrank} up using tensors to arrive at our second low-rank scenario. More precisely,
 defining the tensors ${\mathcal V}_{jk} = V_{jk} \otimes V_{jk}$
 and the rank-one tensor ${\mathcal Z} = Z \otimes Z$, we can express~\eqref{zauner} as the  problem
 \begin{equation}\label{zaunertensor}
  \begin{array}{ll}
    \text{Find}   & \quad {\mathcal Z}\\
    \text{subject to} & \trace({\mathcal Z}{\mathcal V}_{jk}) =   \begin{cases}  1 & \text{if $k=j=0$}, \\
                 \frac{1}{n+1} & \text{else,} 
 \end{cases} \\
& \quad {\mathcal Z} \succeq 0\\
& \quad \rank({\mathcal Z}) = 1,
\end{array}
 \end{equation}
 with an appropriate interpretation of trace, positive-definiteness, and rank for tensors. While the equations in~\eqref{zaunertensor} are now linear, this simplification comes at the cost of substantially increasing the number of unknowns to $(n/3)^4$. Perhaps modifications of recent algorithms for low-rank tensor recovery (see e.g.~\cite{rauhut2017low}) can be utilized to solve~\eqref{zaunertensor} .}

\subsection{{Phase retrieval, sparsity and beyond}}\label{ss:sparsity}

Support constraints  have been popular in phase retrieval for a very long time as a means to make the problem well-posed or to make algorithms converge (faster) to the desired solution.  When imposing a support constraint, one usually one assumes that one knows (an upper bound of) the interval or region in which the object is non-zero. Such a constraint is easy to enforce numerically and it  has been discussed in detail in previous sections.

A more general form of support constraint is {\em sparsity}. In recent years the concept sparsity has been recognized as an enormously useful assumption in all kinds of inverse problems. When a signal is sparse, this means  that the signal has only relatively few non-zero coefficients in some (known) basis, but we do not know a priori the indices of these coefficients. For example, in case of the standard basis, this would mean that we know the signal is sparsely supported, but we do not know the locations of the non-zero entries. An illustrative example is depicted in Figure~\ref{fig:corn}. The simplest setting is when the basis in which the signal is represented sparsely is known in advance. When such a basis or dictionary is not given a priori, it may have to be learned from the measurements themselves~\cite{tillmann2016dolphin}.

When we assume sparsity we are no longer dealing with a linear subspace condition as is the case with ordinary support constraints, but with a non-linear subspace. 
Due to this fact, such a ``non-linear'' sparsity constraint is much harder to enforce than the case when the support of the signal is known a priori.

Owing to the theory of compressive sensing~\cite{CanTao06,Don06,FouRa13}  we now have a  thorough and quite broad theoretical and algorithmic understanding of how to exploit sparsity to either reduce the number of measurements and/or to improve the quality of the reconstructed signal. We call a signal $x \in \CC^n$ $s$-sparse if $x$ has at most $s$ non-zero entries and write $\|x\|_0=s$ in this case. The theory of compressive sensing tells us in a nutshell that under appropriate conditions of the sensing matrix $A \in \CC^{N\times n}$, an $s$-sparse signal
$x\in\CC^n$ can be recovered from the linear measurements $b=Ax$  via linear programming (with high probability) if $N \gtrsim s \log n$, see~\cite{FouRa13} for precise versions and many variations.

{\em Classical} compressive sensing assumes a linear data acquisition mode, where measurements are of the form
$\langle a_k, x\rangle$. Obviously, this data acquisition mode does fit the phase retrieval problem. Nevertheless, the tools and insights we have gained from compressive sensing can be adapted to some extent to the setting of quadratic measurements, i.e., for phase retrieval. 

The problem we want to address is: assume $\xo$ is a sparse signal, how can we utilize this prior knowledge effectively in the phase retrieval problem? For example, what are efficient ways to enforce sparsity in the numerical reconstruction,  or by how much can we reduce the number of phaseless measurements and still successfully recover $\xo$ with theoretical guarantees, and do so in a numerically robust manner?

There exists a plethora of methods to incorporate sparsity in phase retrieval. This includes convex approaches~\cite{ohlsson2011compressive,li2013sparse},  thresholding strategies~\cite{wang2017sparse,yuan2019phase}, greedy algorithms~\cite{shechtman2014gespar}, algebraic methods~\cite{beinert2017sparse} and tools from deep learning~\cite{hand2018phase,kim2019fourier}. In the following we briefly discuss a few selected techniques in more detail.

Following the paradigm of compressive sensing, it is natural to consider the following semidefinite program to recover a sparse signal $\xo$
from phaseless measurements. We denote $\|X\|_1: = \sum_{k,l} |X_{k,l}|$, and similar to using the trace-norm of a matrix $X$ as a convex surrogate of the rank of $X$, we use  $\|X\|_1$ as a convex surrogate of $\|X\|_0$. Hence, we are led to the following semidefinite program (SDP), cf.~\cite{ohlsson2011compressive,li2013sparse}:
\begin{equation}
\label{eq:tracesparse}
 \begin{array}{ll}
    \text{minimize}   & \quad \|X\|_1 + \lambda \trace(X)\\
    \text{subject to} & \quad  \cA(X) = y\\
& \quad X \succeq 0.
\end{array}
\end{equation}

In~\cite{li2013sparse} it is shown that for Gaussian measurement vectors, $N = {\mathcal O}(s^2 \log n)$ measurements are sufficient to recover an
$s$-sparse input from phaseless measurements using~\eqref{eq:tracesparse}.   Based on optimal sparse recovery results from compressive sensing using Gaussian matrices, one would hope that $N = {\mathcal O}(s \log n)$ should suffice. 
However,~\cite{li2013sparse} showed that the SDP 
in~\eqref{eq:tracesparse} cannot outperform this suboptimal sample complexity by direct $\ell_1$-penalization. 

It is conceptually easy to enforce some sparsity of the signal to be reconstructed in the algorithms based on alternating projections or gradient descent, described in Section~\ref{s:nonconvex}. One only needs to incorporate an additional greedy step or a thresholding step during each iteration. For example, for gradient descent we modify the update rule~\eqref{wfgrad} to 
\begin{equation}\nn
z_{j+1} =    {\mathcal T}_{\tau} \big( z_j - \frac{\mu_j}{\|z_0\|_2^2}  \nabla L(z_j) \big),
\end{equation}
where ${\mathcal T}_{\tau}(z)$ is a threshold operator that e.g.\ keeps the $\tau$ largest entries of $z$ and sets the other entries of $z$ to zero; or
alternatively, ${\mathcal T}_{\tau}$ leaves all values of $z$ above a certain threshold (indicated by $\tau$) unchanged, and sets all values of $z$
below this threshold to zero. We can also replace the latter hard thresholding procedure by some soft thresholding rule. Here, it is assumed that the signal is sparse in the standard basis, otherwise the thresholding procedure has to be applied in the suitable basis that yields a sparse representation, such as perhaps a wavelet basis (at the cost of applying additional forward and inverse transforms).

While such modifications are easy to carry numerically,  providing theoretical guarantees is significantly harder. For example, it has been shown  that sparse Wirtinger Flow~\cite{yuan2019phase} as well as truncated amplitude flow~\cite{wang2017sparse} succeed if the sampling complexity is at least ${\mathcal O}(s^2 \log n)$.
Applying a thresholded Wirtinger flow to a non-convex empirical risk minimization problem that is derived from the phase retrieval 
problem,~\cite{cai2016optimal} have established optimal  convergence rates for noisy sparse phase retrieval under sub-exponential noise.

Two-stage approaches have been proposed as well, where in the first stage the support of the signal is identified and in the second state the signal is recovered using the information from the first stage~\cite{iwen2017robust,jaganathan2017sparse}. For example, Jaganathan et al.\ propose such a two-sate scheme for the one-dimensional Fourier phase retrieval problem, consisting of (i)~identifying  the  locations  of  the  non-zero  components,  of  the  signal  using  a combinatorial algorithm, (ii)~identifying  the  signal  values in  the support using a convex algorithm.  This algorithm is shown experimentally to recover $s$-sparse signals from ${\mathcal O}(s^2)$ measurements, but the theoretical guarantees require a higher sample complexity.

An alternative approach to model signals with a small number of parameters  is proposed in~\cite{hand2018phase}, based on generative models.
In this work, the authors suppose that the signal of interest is in the range of a deep generative neural network $G: \IR^s \to \IR^n$, where
the generative model is a $d$-layer, fully-connected, feed forward neural network with random weights.
The authors introduce an empirical risk formulation and prove, assuming a range of technical conditions holds,  that this optimization problem has favorable global geometry for gradient methods, as soon as the number of measurements satisfies $N = {\mathcal O}(s d^2 \log n)$.

Given the current intense interest in deep learning, it is  not surprising that numerous other deep learning based methods 
for phase retrieval have been proposed, see e.g.~\cite{metzler2018prdeep,rivenson2018phase,gladrow2019digital,zhang2018fast}. 
Many of the deep learning based methods  come with little theoretical foundation and are sometimes difficult to reproduce. Moreover, if one changes the input parameters just by a small amount, say, by switching to a slightly different image resolution, a complete retraining of the network is required.
As most  deep learning applications, there is currently almost no theory about any kind of reconstruction guarantee, convergence rate, stability analysis, and other basic questions one might pose to a numerical algorithm. On the other hand, there is  anecdotal evidence that deep learning has the potential to achieve convincing results in phase retrieval. 

Instead of designing an end-to-end deep learning based phase retrieval algorithm (and thereby ignoring the underlying physical model), a more promising direction seems to be to utilize all the information available to model the inverse problem and bring to bear the power of deep learning as a data-driven regularizer. Such an approach has been advocated for general inverse problems in~\cite{li2018nett,arridge_maass_2019}.
It will interesting to adapt these techniques to the setting of phase retrieval.

\medskip
In~\cite{schniter2015message} the authors have proposed an Approximate Message Passing (AMP)  approach for phase retrieval of sparse signals. AMP based methods were originally developed for compressed sensing problems of estimating sparse vectors from underdetermined linear measurements~\cite{donoho2010message}. They have now been extended to a wide range of estimation and learning problems including  matrix completion, dictionary learning, and  phase retrieval. 
The first AMP algorithm designed for phase retrieval for sparse signals using techniques from compressive sensing can be found in~\cite{schniter2014compressive}.
Various extensions and improvements have been developed~\cite{dremeau2015phase,metzler2016bm3d,metzler2017coherent}. 

As pointed out in~\cite{metzler2017coherent}, one downside is that AMP algorithms are heuristic algorithms and  at best offer only asymptotic guarantees. In the case of the phase retrieval problem, most AMP algorithms offer no guarantees at all. Despite this shortcoming, they often perform well in practice and a  key appealing feature of AMP is its computational scalability. See~\cite{metzler2017coherent} for a more detailed discussion of AMP algorithms for phase retrieval.

\medskip

{In another line of research, the randomized Kaczmarz method has been adapted to phase retrieval, see~\cite{wei2015solving}. Competitive theoretical convergence results can be found in~\cite{tan2019phase,jeong2017convergence}, where it has been shown that the convergence is exponential and comparable to the linear setting~\cite{strohmer2009randomized}.}

\section{Blind ptychography}\label{sec:blind-ptycho}\label{ss:blindptycho}

An important development in ptychography since the work of~\cite{TDB09} is the potential of simultaneous recovery of the object and the illumination. This is referred to as blind ptychography. There are two ambiguities inherent to any blind ptychography.
 
The first is the affine phase ambiguity. Consider the mask and object estimates
\beq
\label{lp1}
\nu^{0}(\bn)&=&\mu^{0}(\bn) \exp(-\im a -\im \bw\cdot\bn),\quad\bn\in\cM^{0}\\
\label{lp2} \xh(\bn)&=& \x(\bn) \exp(\im b+\im \bw\cdot\bn),\quad\bn\in \IZ^2_n
\eeq
for any $a,b\in \IR$ and $\bw\in \IR^2$.  For any $\bt$, we have the following
calculation
\beqn
\nu^\bt(\bn)&=&\nu^{0}(\bn-\bt)\\
&=&\mu^{0}(\bn-\bt) \exp(-\im \bw\cdot(\bn-\bt))\exp(-\im a)\\
&=&\mu^\bt(\bn) \exp(-\im \bw\cdot(\bn-\bt))\exp(-\im a)
\eeqn
and hence for all $\bn\in \cM^\bt, \bt\in\cT$
\beq
\label{drift2}
\nu^\bt(\bn)  x^\bt(\bn)&=&\mu^\bt(\bn)\x^\bt(\bn) \exp(\im(b-a))\exp(\im \bw\cdot\bt). 
\eeq
Clearly, \eqref{drift2} implies that $g$ and $\nu^{0}$ produce the same ptychographic data as $f$ and $\mu^{0}$ since
for each $\bt$, $\nu^\bt\odot  x^\bt$ is a constant phase factor times $\mu^\bt\odot \x^\bt$ {  where $\odot$ is the entry-wise (Hadamard) product}.  It is also clear that the above statement holds
true regardless of the set $\cT$ of shifts and the type of mask. 

In addition to the affine phase ambiguity \eqref{lp1}-\eqref{lp2},  a scaling factor ($\xh=c \x, \nu^{0}=c^{-1} \mu^{0}, c>0$) is inherent to any blind ptychography. 
Note that when the mask is exactly known (i.e. $\nu^{0}=\mu^{0}$), neither ambiguity can occur.

\subsubsection{Local rigidity} 

Motivated by \eqref{drift2} we seek sufficient conditions for results such as
\beq
\label{3.100}
\nu^{k}\odot  x^k=e^{\im\theta_{k}}\mu^{k}\odot \x^k,\quad k=0,\dots,Q-1, 
\eeq
for some constants $\theta_k\in \IR.$ We call \eqref{3.100} the property of {\em local rigidity}.

 \begin{figure}
 \centering
 \includegraphics[width=7cm]{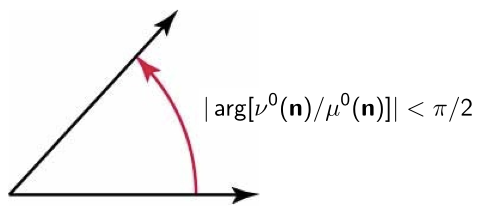}
\caption{$\nu^0$ satisfies MPC if $\nu_0(\bn)$ and $\mu^0(\bn)$ form an acute angle for all $\bn$.}
\label{fig:MPC}
\end{figure}

A main assumption needed here is the {\em mask phase constraint} (MPC):  
 \begin{quote}\em 
The mask estimate $\nu^0$ has  the property
$\Re(\overline{\muh^0}\odot \mu^0)>0$ at every pixel (where $\odot$ denotes the component-wise product and the bar denotes the complex conjugate). 
 \end{quote}
 Another ingredient in the measurement scheme is that at least for one block (say $\cM^\bt$) the corresponding object part
  $f^\bt$ has a tight support in $\cM^\bt$, i.e. 
\beq
\nn
\mb{\rm Box}[\supp(f^\bt)]=\cM^\bt
\eeq 
where $\mb{\rm Box}[E]$ stands for the box hull,  the smallest
rectangle containing $E$ with sides parallel to $\be_1=(1,0)$ or $\be_2=(0,1)$. We call such an object part
{\em an anchor}. Informally speaking, an object part $f^\bt$ is an anchor if
its support touches four sides of $\cM^\bt$ (Figure \ref{fig:corn}). 

In the case $\supp({x})=\cM$, every object part is an anchor. 
For an extremely sparse object such as shown in Figure \ref{fig:corn}, the anchoring assumption can pose a challenge. 

Both the anchoring assumption and MPC are nearly necessary conditions for local rigidity \eqref{3.100} to hold
as demonstrated by counterexamples constructed in \cite{blind-ptycho}.

\begin{thm}\cite{blind-ptycho}\label{thm:many} Suppose that 
$\{\x^k\} $ has an anchor and is $s$-connected with respect to the ptychographic scheme. 
 
Suppose that an object estimate $\xh=\bigvee_k  x^k$, where $ x^k$ are defined on  $\cM^k$,  and a mask estimate $\nu^0$ produce the same ptychographic data as $\x$ and $\mu^0$.   Suppose  that the mask estimate $\nu^0$ satisfies MPC. 
Then local rigidity \eqref{3.100} holds with probability exponentially (in $s$) close to 1. 

\end{thm}

 \begin{figure}[t]
\begin{center}
\includegraphics[width=12cm]{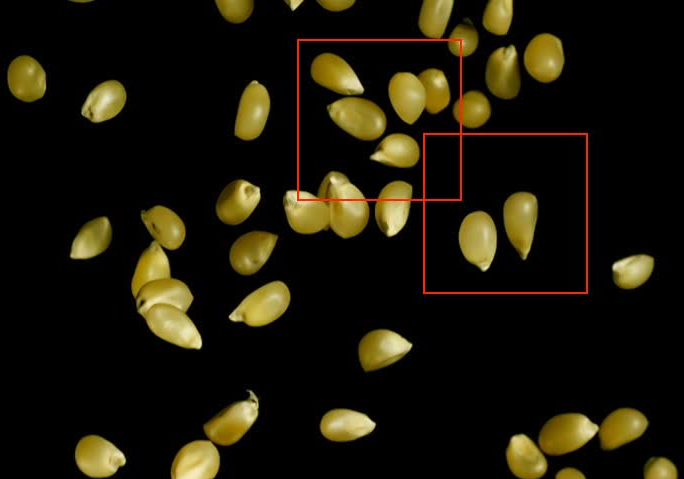}
\caption{Sparse objects such as this image of corn grains, where the dark area represents zero pixel value, can be challenging to ptychographic measurements.  The two red-framed blocks are not connected even though they overlap. The object part in the lower-right block is not an anchor since the object support does not touch the four sides of the block while
the object part in the upper-left block is an anchor. Indeed, the two corn grains at the lower-left and upper-right corners alone of the latter block
suffice to create a tight support. 
}
\label{fig:corn}
\end{center}
\end{figure}

\subsubsection{Raster scan ambiguities}
\label{sec:recovery}

Before describing  the global rigidity result, let us review 
the other ambiguities associated with the raster scan \eqref{raster}  other than the inherent ambiguities of
the scaling factor
and the affine phase ambiguity \eqref{lp1}-\eqref{lp2}.
{These ambiguities  include} the arithmetically progressing phase factor  inherited from the block phases and
the raster grid pathology which
has  a $\tau$-periodic structure of $\tau\times\tau$ degrees of freedom.  

Let $\cT'$ be any cyclic subgroup of $\cT$ generated by $\bv$, i.e. $\cT':=\{\bt_j=j\bv:j=0,\dots,s-1\}$, of order $s$, i.e.  $s\bv=0\mod\, n$. 
For ease of notation, denote by $\mu^k, \x^k, \muh^k, \xh^k$ and $M^{k}$ for the respective $\bt_k$-shifted quantities.
\medskip
\begin{thm}\cite{raster}\label{lem2}  Suppose that
\beqn
\muh^k\odot \xh^k=e^{\im \theta_k}\mu^k\odot \x^k,\quad k=0,\dots,s-1,
\eeqn
{  where $\mu^k$ and $\muh^k$ vanish nowhere in $\cM^k$.}
 If, for all $ k=0,\dots,s-1,$
\beq
\label{33}
 \cM^{k}\cap \cM^{k+1}\cap \supp(\x) \cap(\supp(\x)+ {\bv})\neq \emptyset,
\eeq
then the sequence 
 $\{\theta_0,\theta_1,\dots,\theta_{s-1}\}$ is an arithmetic progression  where $\Delta\theta=\theta_k-\theta_{k-1}$  is an integer multiple of $2\pi/s$.
 
 For the full raster scan $\cT$,  the block phases have the profile 
\beq
\label{a3} \theta_{kl}=\theta_{00}+\br\cdot(k,l),\quad k,l=0,\dots,q-1, 
\eeq
for some $\theta_{00}\in \IR$ and $  \br=(r_1,r_2)$ where $r_1$ and $r_2$ are integer multiples
of $2\pi/q$.

\end{thm}
Note that 
if $\x$ has a full support, i.e. $\supp(\x)=\IZ^2_n$, then \eqref{33} holds for any step size $\tau<m$ (i.e. positive overlap). 

The next example shows an ambiguity resulting from the arithmetically progressing block phases \eqref{a3} 
which make positive
and negative imprints on the object and phase estimates, respectively. 

\medskip
\begin{ex}\label{ex6} For $q=3, \tau=m/2$, let
 \beq
\nn \x&=&
\lt[\begin{matrix}
f_{00}&f_{10}&f_{20}\\
f_{01}& f_{11}  & f_{21}   \\
f_{02} & f_{12} & f_{22}
\end{matrix}\rt]\\
\nn\xh&=&
\left[
\begin{matrix}
f_{00}  & e^{\im 2\pi/3} f_{10} & e^{\im 4\pi/3} f_{20}   \\
e^{\im 2\pi/3} f_{01} &e^{\im 4\pi/3} f_{11} &f_{21}  \\
e^{\im 4\pi/3} f_{02}  & f_{12}& e^{\im 2\pi/3}  f_{22}  
\end{matrix}
\right]
\eeq
be  the object and its reconstruction, respectively, where $f_{ij}\in \IC^{n/3\times n/3}$.  Let 
 \beq
 \nn
\mu^{kl}=
\left[
\begin{matrix}
\mu^{kl}_{00}  & \mu^{kl}_{10}    \\
\mu^{kl}_{01}&\mu^{kl}_{11}
\end{matrix}
\right],\quad \muh^{kl}=
\left[
\begin{matrix}
\mu^{kl}_{00}  &e^{-\im 2\pi/3}  \mu^{kl}_{10}    \\
e^{-\im 2\pi/3} \mu^{kl}_{01}&e^{-\im 4\pi/3} \mu^{kl}_{11}
\end{matrix}
\right], 
\eeq
$ k,l=0,1,2,$
be the $(k,l)$-th shift of the mask and estimate, respectively, where $\mu^{kl}_{ij}\in \IC^{n/3\times n/3}$. 

Let $\x^{ij}$ and $\xh^{ij}$ be the part of the object and estimate illuminated by $\mu^{ij}$ and $\muh^{ij}$, respectively. For example, we have
\[   \x^{00}=
\lt[\begin{matrix}
f_{00}&f_{10}\\
f_{01}& f_{11}
\end{matrix}\rt],\quad  \x^{10}=
\lt[\begin{matrix}
f_{10}&f_{20}\\
 f_{11}  & f_{21}
\end{matrix}\rt],\quad  \x^{20}=
\lt[\begin{matrix}
f_{20}&f_{00}\\
f_{21}& f_{01}
\end{matrix}\rt]
\]
{  and likewise for other $\x^{ij}$ and $\xh^{ij}$}. 
 It is easily seen that $\muh^{ij}\odot \xh^{ij}=e^{\im (i+j)2\pi/3} \mu^{ij}\odot \x^{ij}.$
\end{ex}

Example \ref{ex6} illustrates  the non-periodic ambiguity inherited from the affine block phase profile. 
 The non-periodic arithmetically progressing  ambiguity  is different from  the affine phase ambiguity \eqref{lp1}-\eqref{lp2} as they manifest  on different scales: {  
the former is constant in each $\tau\times\tau$ block (indexed by $k,l$)  while the latter varies from pixel to pixel}.

The next example illustrates the periodic artifact called raster grid pathology. 

\begin{ex}\label{ex5} For $q=3, \tau=m/2$ and any  $\psi\in \IC^{{n\over 3}\times {n\over 3}}$, let 
 \beq
\nn \x&=&
\lt[\begin{matrix}
f_{00}&f_{10}&f_{20}\\
f_{01}& f_{11}  & f_{21}   \\
f_{02} & f_{12} & f_{22}
\end{matrix}\rt]\\
\label{7.2-1} \xh&=&
\left[
\begin{matrix}
e^{-\im \psi}\odot f_{00}  & e^{-\im \psi} \odot f_{10} & e^{-\im \psi} \odot f_{20}   \\
e^{-\im \psi}\odot f_{01} &e^{-\im\psi} \odot f_{11} & e^{-\im\psi}\odot f_{21}  \\
e^{-\im \psi}\odot  f_{02}  & e^{-\im\psi}\odot  f_{12}& e^{-\im \psi} \odot  f_{22}  
\end{matrix}
\right]
\eeq
be  the object and its reconstruction, respectively, where $f_{ij}\in \IC^{n/3\times n/3}$.  Let 
 \beq\label{7.2-2}
\mu^{kl}=
\left[
\begin{matrix}
\mu^{kl}_{00}  & \mu^{kl}_{10}    \\
\mu^{kl}_{01}&\mu^{kl}_{11}
\end{matrix}
\right],\quad \muh^{kl}=
\left[
\begin{matrix}
e^{\im\psi}\odot \mu^{kl}_{00}  &e^{\im \psi} \odot \mu^{kl}_{10}    \\
e^{\im \psi}\odot \mu^{kl}_{01}&e^{\im \psi} \odot \mu^{kl}_{11}
\end{matrix}
\right], 
\eeq
$ k,l=0,1,2,$
be the $(k,l)$-th shift of the mask and estimate, respectively, where $\mu^{kl}_{ij}\in \IC^{n/3\times n/3}$. 

Let $\x^{ij}$ and $\xh^{ij}$ be the part of the object and estimate illuminated by $\mu^{ij}$ and $\muh^{ij}$, respectively (as
in Example \ref{ex6}). It is verified easily that $\muh^{ij}\odot \xh^{ij}= \mu^{ij}\odot \x^{ij}.$
\end{ex}

Since $\psi$ in Example \ref{ex5}  is any complex $\tau\times \tau$ matrix, \eqref{7.2-1} and \eqref{7.2-2} represent the maximum degrees of ambiguity
over the respective initial sub-blocks. This ambiguity is transmitted to other sub-blocks, forming periodic artifacts called the raster grid pathology.

For a complete analysis of ambiguities associated with raster scan, we refer the reader to \cite{raster}.

\subsubsection{Global rigidity}

In view of Theorem \ref{thm:many},  we make simple observations and transform \eqref{3.100} into the ambiguity equation  that will be a key to subsequent development. 

 Let 
\beqn
\alpha(\bn)\exp[\im \phi(\bn)]=\muh^0(\bn)/\mu^0(\bn), \quad \alpha(\bn)>0,\quad \forall \bn \in \cM^0
\eeqn
and 
\beqn
h(\bn)&\equiv& \ln \xh(\bn)-\ln \x(\bn),\quad\forall \bn\in \cM, \eeqn
where $\x$ and $\xh$ are assumed to be non-vanishing.

Suppose that
\beq
\nn
\muh^{k}\odot \xh^k=e^{\im\theta_{k}}\mu^{k}\odot \x^k,\quad \forall k, 
\eeq
where $\theta_k$ are constants.
Then
\beq \label{200.1}
h(\bn+\bt_{k}) &=&\im \theta_{k}-\ln\alpha(\bn)-\im \phi(\bn) \mod \im 2\pi,\quad\forall\bn\in \cM^0,\eeq
and 
for all $\bn\in  \cM^{k} \cap \cM^{l}$
\beq
\nn
\alpha(\bn-\bt_l)&=&\alpha(\bn-\bt_k)\\
 \theta_k- \phi(\bn-\bt_k)&= & \theta_l- \phi(\bn-\bt_l) \mod 2\pi.  \nn
\eeq

 The ambiguity  equation \eqref{200.1} 
is a manifestation of local uniqueness \eqref{3.100} and  has the immediate consequence 
 \beq
 \label{200.4}
 h(\bn+\bt_{k})-h(\bn+\bt_{l})=\im \theta_{k}-\im\theta_{l} \mod \im2\pi,\quad \forall \bn\in \cM^{0},\quad \forall k,l
 \eeq
 or equivalently 
 \beq
 \label{200.5'}
 h(\bn+\bt_{k}-\bt_{l})-h(\bn)=\im \theta_{k}-\im\theta_{l} \mod \im2\pi,\quad \forall \bn\in \cM^{l}
 \eeq
 by shifting the argument in $h$. 

We refer to  \eqref{200.4}or \eqref{200.5'} as the {\em phase drift equation} which determines  the ambiguity 
(represented by $h$)  at different locations connected by ptychographic shifts.

We seek sufficient conditions 
for guaranteeing the following global rigidity
properties
\beqn
h(\bn)&=&h(0)+\im \bn\cdot (r_1,r_2)\mod \im 2\pi, \\
\phi(\bn)&=&\theta_0-\Im[h(0)]-\bn\cdot(r_1,r_2) \mod 2\pi\\
\alpha&=&e^{-\Re[h(0)]}\label{473''}\\
\theta_\bt&=&\theta_0
+\bt\cdot(r_1,r_2)\mod 2\pi,\quad\forall \bt\in \cT,
\eeqn
for some $ r_1, r_2\in \IR$ and all $\bn\in \IZ_n^2$.

In \cite{blind-ptycho} a class of {ptychographically complete schemes
 are introduced.  A ptychographic scheme is complete if global rigidity  holds under the minimum prior constraint MPC {defined in Figure 7.1}. }
A simple example of {ptychographically complete} schemes is the perturbed scan (Figure \ref{fig1'}(b)  \beq
\label{perturb}\label{rank2}
\bt_{kl}=\tau (k,l)+(\delta^1_{kl},\delta^2_{kl}),\quad k,l=0,\dots,q-1
\eeq
where $\tau=n/q$ needs only to be slightly greater than $m/2$ (i.e overlap ratio slightly greater than $50\%$)  and $ \delta^1_{kl}, \delta^2_{kl}$ are small integers with
some generic, non-degeneracy conditions \cite{blind-ptycho}. In particular, if we set 
\beq
\label{perturb2}\label{rank1}
\delta^1_{kl}=\delta^1_k, \quad \delta^2_{kl}=\delta^2_l,\quad\forall k,l=0,\cdots, q-1,  
\eeq
then we obtain the scan pattern shown in Figure \ref{fig1'} (a).

\begin{figure}
\centering
\subfigure[Perturbed  scan \eqref{perturb2}]{\includegraphics[width=4.8cm]{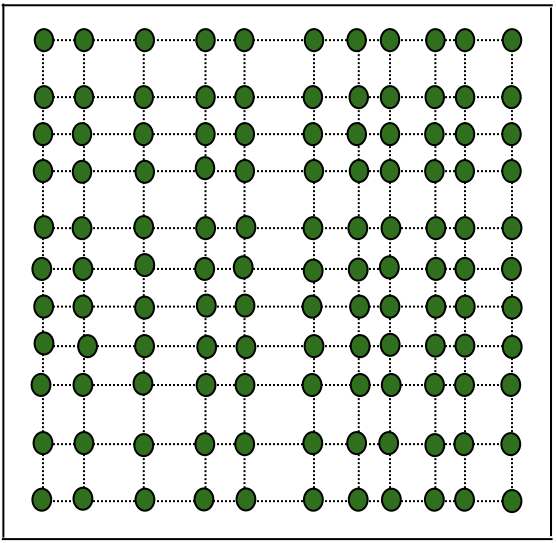}}\hspace{1cm}
\subfigure[Perturbed scan  \eqref{perturb}]{\includegraphics[width=5cm,height=4.7cm]{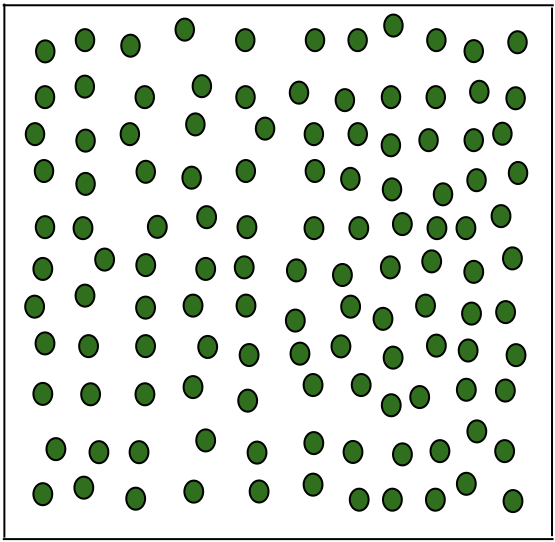}}
\caption{Perturbed raster scan patterns}
\label{fig1'}
\end{figure}

\subsubsection{Minimum overlap ratio}
 \label{ex50} 
\begin{figure}
\centering
\includegraphics[width=8cm]{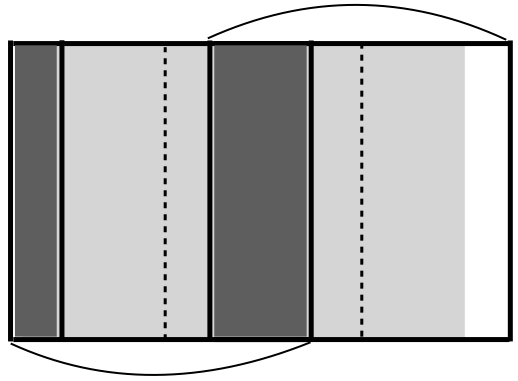}
\caption{A perturbed scan with $q=2$.  The arcs indicate the extend of the two blocks $\cM^{00}$ and $\cM^{10}$. The dotted lines mark the midlines of the two blocks. The grey area represents the object with the light grey areas being $R_{00}$ and $R_{10}$ and the dark grey areas being the overlap of the two blocks. The white area inside $\cM^{10}$ folds into the other end inside $\cM^{00}$ by the periodic boundary condition.}
\label{fig50}
\end{figure}

In this section, we show that $50\%$ overlap is roughly the minimum overlap ratio required by uniqueness
among the perturbed raster scans defined by  \eqref{perturb}-\eqref{perturb2}.

Let us consider the perturbed scheme \eqref{perturb2} with $q=2$ and  
\beqn
\bt_{kl}=(\tau_k,\tau_l),\quad k,l=0,1,2
\eeqn where $\tau_0=0,  \tau_2=n$ and 
\beq
{3m/2}<n<m+\tau_1. \label{84.0}
\eeq
The condition \eqref{84.0} is to ensure that the overlap ratio ($2-n/m$) between two adjacent blocks is less than (but can be made arbitrarily close to) $50\%$. To avoid the raster scan (which has many undesirable ambiguities \cite{raster}), we assume that $\tau_1\neq n/2$ and hence $\tau_2\neq 2\tau_1$. 
Note that the periodic boundary condition implies that $\cM^{00}=\cM^{20}=\cM^{02}=\cM^{22}$. 
Figure \ref{fig50} illustrates the relative positions of $\cM^{00}$ and $\cM^{10}$. 

First let us focus on the horizontal shifts $\{\bt_{k0}: k=0,1,2\}$. As shown in Figure \ref{fig50}, two subsets  of $\cM=\IZ^2_n$
\[
R_{00}= \lb m+\tau_1-n,\tau_1-1 \rb\times \IZ_m,\quad R_{10}= \lb m, n-1\rb\times \IZ_m
\]
 are covered only once  by $\cM^{00}$ and $\cM^{10}$
respectively due to the \eqref{84.0}. 

 Now consider the intersections 
 \beqn
 \tilde R_{10}:=R_{10}\cap (\bt_{10}+R_{00})&=&R_{10}\cap \lb m+2\tau_1-n,2\tau_1-1\rb\times \IZ_m\\
 \tilde R_{00}:= (R_{10}-\bt_{10})\cap R_{00}&=&\lb m-\tau_1, n-\tau_1-1\rb\times \IZ_m\cap R_{00}
 \eeqn
which respectively  correspond to the same region of the mask in $\cM^{10}$ and $\cM^{00}$
and let $h_1$ be any function defined on $\cM$ such that $h_1(\bn)=0$ for any $\bn\neq  \tilde R_{10}\cup \tilde R_{00}$
and $h_1(\bn+\bt_{10})=h_1(\bn)$ for any $\bn\in \tilde R_{00}$. 

Consider the object estimate
$ x(\bn)=
e^{h_1(\bn)}\x(\bn)$ and 
 the mask estimate
$
\nu^{k0}(\bn):=e^{-h_1(\bn)}\mu^{k0}(\bn) $,  which is well defined because $\tilde R_{10}=\bt_{10}+\tilde R_{00}$ and
both correspond to the same region of the mask.

By the same token, we can construct a similar ambiguity function $h_2$ for the vertical shifts. 
With both horizontal and vertical shifts, we define the ambiguity function $h=h_1h_2$ and the
associated pair of mask-object estimate $
\nu^{kl}(\bn):=e^{-h(\bn)}\mu^{kl}(\bn) $ and $ x(\bn)=
e^{h(\bn)}\x(\bn).$

Clearly, the mask-object pair $(\nu,x)$ produces the identical set of
diffraction patterns as $(\mu, \x)$. 
Therefore this ptychographic scheme has
at least $(2\tau_1-m)^2$ or $(2n-2\tau_1-m)^2$ degrees of ambiguity dimension
depending on whether $2\tau_1<n$ or $2\tau_1>n$.

\subsection{Algorithms for blind ptychography}

Let $\cF(\nu,x)$ be the bilinear transformation representing the totality of the Fourier (magnitude and phase) data for any mask $\nu$ and object $x$.
From $\cF(\nu^0,x)$ we can define two measurement matrices. First, for a given $\nu^0\in \IC^{m^2}$, 
let $A_{\nu}$ be defined via the relation $A_{\nu} x:=\cF(\nu^0,x)$ for all $x\in \IC^{n^2}$; second, for a given $x\in \IC^{n^2}$,
let $B_{x}$ be defined via $B_{x} {\nu}=\cF(\nu^0,x)$ for all $\nu^0\in \IC^{m^2}$. 

More specifically, let $\Phi$ denote the over-sampled Fourier matrix. The measurement matrix $A_{\nu}$ is a concatenation of $\{\Phi \,\diag(\nu^\bt):\bt \in \cT\}$ (Figure \eqref{fig10}(a)). Likewise,
$B_{x}$ is $\{\Phi \,\diag(x^\bt):\bt \in \cT\}$ stacked on top of each other (Figure \eqref{fig10}(b)). 
Since
$\Phi$ has orthogonal columns, both $A_{\nu}$ and $B_{x}$ have orthogonal columns.
We simplify the notation by setting $A=A_\mu$ and $B=B_{\x}$.

Let ${\nu}^{0}$ and $\xh=\vee_\bt  x^\bt$ be any pair
 of the mask and the object estimates producing  the same ptychography data as $\mu^{0}$ and $\x$, i.e.
 the diffraction pattern of ${\nu}^\bt\odot  x^\bt$ is identical to that of $\mu^\bt\odot \x^\bt$ where
 ${\nu}^\bt$ is the $\bt$-shift of ${\nu}^{0}$ and $ x^\bt$ is the restriction of $\xh$ to $\cM^\bt$. 
We refer to the pair $({\nu}^0,\xh)$ as   a  blind-ptychographic solution and
$(\mu^0,\x)$ as the true
solution  (in the mask-object domain).

We can write  the total measurement data as $b=|\cF(\mu^0,\x)|$ where $\cF$ is  the concatenated oversampled Fourier transform acting on
$\{\mu^\bt\odot \x^\bt:\bt \in \cT\}$  (see Fig. \ref{fig10}), i.e. a bi-linear transformation in the direct product of the mask space and the object space.  By definition, a blind-ptychographic solution $({\nu}^0,\xh)$ satisfies  $|\cF({\nu}^0, \xh)|=b$.

 According to the global rigidity theorem, we use relative error (RE) and relative residual (RR) as the merit metrics  for the recovered image $x_k$ and mask $\mu_k$ at the $k^{th}$ epoch:
\beq\label{RE}
\mbox{RE}(k)&= &\min_{\alpha\in \mathbb{C}, \mathbf{r} \in \mathbb{R}^2}\frac{\sqrt{\sum_\bn|\x(\bn) - \alpha e^{-\im {2\pi}\mathbf{n}\cdot \mathbf{r}/n} x_k(\bn)|^2}}{\|f\|}\\
\mbox{\rm RR}(k)&= & \frac{\|b- |A_kx_{k}|\|}{\|b\|}.
\eeq
Note that in \eqref{RE} both the affine phase and the scaling factors are waived.

\subsubsection{Initial mask estimate}
\label{sec:ppc}
 For non-convex iterative optimization, a good initial guess  or some regularization is usually crucial for convergence \cite{ML12}, \cite{Poisson2}.  This is even more so for blind ptychography which is doubly non-convex because, in addition to the phase retrieval step, extracting the mask and the object from their product is also non-convex.

We say  that a mask estimate $\nu^0$ satisfies MPC$(\delta)$ if 
\beq\nn
\measuredangle (\nu^0(\bn),\mu^0(\bn))<\delta\pi,\quad \forall\bn 
\eeq
where  $\delta\in (0, 1/2]$ is the uncertainty parameter. 
The weakest condition necessary for uniqueness is $\delta=0.5$, equivalent to
$\Re(\overline{\muh^0}\odot \mu^0)>0.$ Non-blind ptychography gives rise to infinitesimally small $\delta$.

We use MPC$(\delta)$ as measure of initial mask estimate for blind ptychographic reconstruction
and randomly choose $\nu^0$ from the set MPC$(\delta)$. 
Specifically,  we use the following mask initialization
\beq
\nn
\mu_1 (\mathbf{n})=\mu^0(\mathbf{n})\,\exp{\lt[\im 2\pi\frac{ \mathbf{k}\cdot\mathbf{n}}{n}\rt]}\,\exp{[\im \phi(\mathbf{n})]},\ \ \ \mathbf{n}\in \mathcal{M}^0
\eeq
where $\phi(\bn)$ are independently and uniformly distributed on $(-\pi\delta , \pi\delta)$.

Under MPC,  however, the initial mask  may be  significantly far away from the true mask in
norm. Even if  $|\nu^0(\bn)|=|\mu^{0}(\bn)|=\mbox{const.}$, 
 the mask guess with uniformly distributed $\phi$ in $(-\pi/2, \pi/2]$  has the relative error close to
\[
\sqrt{{1\over \pi}\int^{\pi/2}_{-\pi/2} |e^{\im\phi}-1|^2 d\phi}=\sqrt{2(1-{2\over \pi})}\approx   0.8525 
\]
with high probability.

\subsubsection{Ptychographic iterative engine (PIE)}
The ptychographic iterative engines, PIE \cite{PIE104}, \cite{PIE05}, \cite{PIE204}, ePIE \cite{ePIE09} and rPIE \cite{rPIE17},  are related to the mini-batch gradient method.

In PIE and ePIE, the exit wave estimate is given by 
\beq
\label{717}
\tilde\psi^{k}=\Phi^* \lt[b^{k}\odot\sgn(\Phi ({\nu}^{k}\odot  x^k))\rt]
\eeq
analogous to AP where
the $k$-th  object part $ x^k$  is updated  by a gradient descent 
\[
 x^k-{1\over 2\max_{\bn} |{\nu}^k(\bn)|^2} \nabla_{\muh}\|\muh^{k}\odot  x^k-\tilde \psi^{k}\|^2.  
\]
This choice of step size resembles  the Lipschitz constant of the gradient of  the loss function $\half \|\muh^{k}\odot  x^k-\tilde \psi^{k}\|^2$. 
The process continues in random order  until each of the diffraction patterns has been used to update the object and mask estimates, at which point a single PIE iteration has been completed.
The mask update proceeds in a similar manner. 

The update process can be done in parallel as in \cite{DM08}, \cite{TDB09}. First the exit wave estimates are updated in parallel by
the AAR algorithm instead of \eqref{717}, i.e.
\beq
\tilde\psi_{j+1}=\half \tilde\psi_j+R_YR_X \tilde \psi_j\nn
\eeq
where $\tilde\psi_j=[\tilde \psi^k_j]$ is the $j$-th iterate of the exit wave estimate. Second, the object and the mask are updated by solving iteratively
the Euler-Lagrange equations
\beqn
x_j(\bn)&=& {\sum_k [\mu_j^k\odot \tilde\psi_j^k](\bn) \over \sum_k |\mu_j^k(\bn)|^2}
\eeqn
of the bilinear loss function 
\beq
\nn
\half\sum_k \|\mu_j^{k}\odot x_j^k-\tilde \psi_j^{k}\|^2 &=&\half\sum_k \|\Phi\lt[\mu_j^{k}\odot x_j^k\rt]-\Phi\tilde \psi_j^{k}\|^2\\
&=&\half\sum_k \|\cF(\mu^k_j, x^k_j)-\Phi\tilde \psi^k_j\|^2\nn
\eeq
for given $\tilde\psi_j$ (recall the isometric property of $\Phi$).

\subsubsection{Noise-aware method} As a first step of the noise-aware ADMM method  for blind ptychography, 
we may consider the 
augmented Lagrangian
\beq
\nn
\cL(\nu,x,z,\lamb)=\half\|b-|z|\|^2+\lamb^*(z-\cF(\nu,x))+{\beta\over 2} \|z-\cF(\nu,x)\|^2
\eeq
and the scheme 
\beqn
\mu_{k+1}&=& \arg\min \cL(\nu,x_k,z_k,\lamb_k)\nn\\
x_{k+1}&=& \arg\min \cL(\mu_{k+1},x,z_k,\lamb_k)\nn\\
z_{k+1}&=& \arg\min \cL(\mu_{k+1},x_{k+1},z,\lamb_k)\\
\lamb_{k+1}&=&\lamb_k+ \beta(z_{k+1} -\cF(\mu_{k+1},x_{k+1})).
\eeqn
In \cite{March19}, more elaborate version of the above scheme is employed to enhance convergence. 

\subsubsection{Extended Gaussian-DRS}

As extension of the Gaussian-DRS \eqref{G1}, 
consider the augmented Lagrangian
 \beq\label{AL3}
 \cL(y,z, x,\nu,\lamb)&=& \half \||z|-b\|^2+\lamb^*(z-y)+{\rho\over 2} \|z-y\|^2+\II_\cF(y)
 \eeq
 where $\II_\cF$ is the indicator function of the set 
 \[
 \{y\in \IC^N: y=\cF(\nu, x) \quad \mbox{for some}\,\,\nu,x\}.
 \]
 Define the ADMM scheme for \eqref{AL3} as 
  \beqn
(z_{k+1}, \mu_{k+1})&=& \arg\min_z  \cL(y_{k}, z, x_{k}, \nu, \lamb_k)\\
 (y_{k+1},x_{k+1})&=& \arg\min_y \cL(y, z_{k+1},  x, \mu_{k+1}, \lamb_k) \\
 \lamb_{k+1}&=& \lamb_k+\rho (z_{k+1}-y_{k+1}) 
 \eeqn
which is carried out explicitly by
\beq
\label{50} z_{k+1}&=& {1\over \rho+1} P_Y (y_k-\lamb_k/\rho)+{\rho\over\rho+1} (y_k-\lamb_k/\rho)\\
\label{51} \mu_{k+1}&=& B_k^+ y_k\\
\label{52} y_{k+1}&=& A_{k+1}A_{k+1}^+ (z_{k+1}+\lamb_k/\rho)\\
\label{53} x_{k+1}&=& A_{k+1}^+ y_{k+1}\\
\label{54} \lamb_{k+1}/\rho&=&{\lamb_k/\rho} +z_{k+1}-y_{k+1}.
\eeq
We can further simplify the above scheme  in terms of the new variable
\[
u_k=z_{k}+\lamb_{k-1}/\rho.
\]

Rewrite eq.  \eqref{52} as 
\beq
\label{52'}
y_{k+1}=A_{k+1}A_{k+1}^+ u_{k+1}
\eeq
and hence \eqref{54} as 
\beq
\label{54'}\lamb_{k+1}/\rho&=& u_{k+1}-y_{k+1}\\
&=& u_{k+1}-A_{k+1}A_{k+1}^+u_{k+1}.\nn
\eeq
Combining \eqref{52'} and \eqref{54'} we obtain
\beqn
z_{k+1}&=& \Big({1\over \rho+1} P_Y+{\rho\over\rho+1}\Big)(2A_kA_k^+-I)u_k
\eeqn

On the other hand, 
\beq
\label{55}\lefteqn{u_{k+1}}\\
&=& {1\over \rho+1} P_Y (2A_{k}A_{k}^+u_k -u_k)+{\rho\over\rho+1} (2A_{k}A_{k}^+u_k -u_k) +
u_{k}-A_{k}A_{k}^+u_{k}\nn\\
&=& {u_k\over \rho+1}+{\rho-1\over \rho+1} A_kA_k^+ u_k+{1\over \rho+1} P_Y(2A_{k}A_{k}^+u_k -u_k)  \nn
\eeq
with the mask and object updated by
\beq
 \label{56} \mu_{k+1}&=& B_k^+ A_{k}A_{k}^+ u_k \\
\label{57} x_{k+1}&=&A_{k+1}^+ u_{k+1}. 
\eeq
Eq. \eqref{55}-\eqref{57} constitute the extended version of Gaussian-DRS (eGaussian-DRS) for blind ptychography.

\begin{figure}
\centering 
\subfigure[$50\%$ overlap; $\delta = 9/20$ ]{\includegraphics[width=5cm]{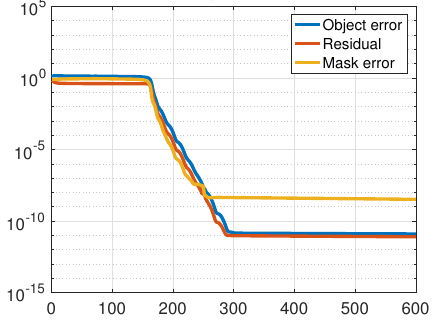}}
\subfigure[$66\%$ overlap; $\delta = 2/5$ ]{\includegraphics[width=5cm,height=3.7cm]{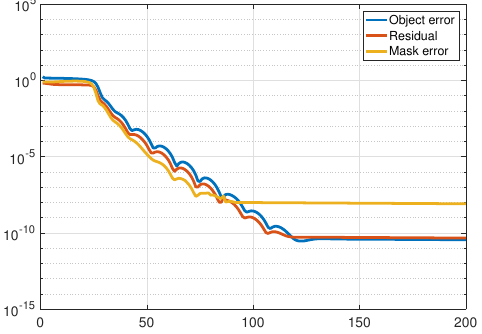}}
\subfigure[$75\%$ overlap; $\delta = 1/2$ ]{\includegraphics[width=5cm]{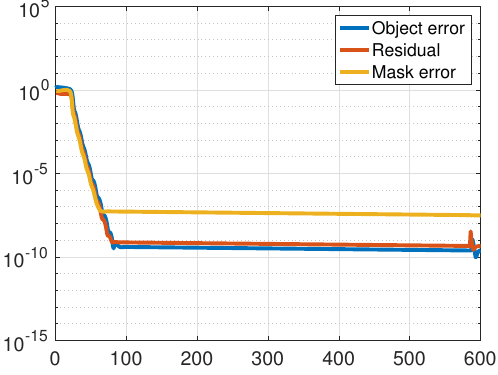}}
\caption{Relative errors versus iteration of blind ptychography by eGaussian-DRS with $\rho=1/3$ for the original object CiB. Scheme \eqref{rank1} with different overlap ratios and initializations are used as indicated in each plot. }
\label{fig:edrs}
\end{figure}

Figure \ref{fig:edrs} shows the relative errors (for object and mask) and residual of eGaussian-DRS with $\rho=1/3$ and
various overlap ratios in the perturbed scan and different initial mask phase uncertainty $\delta$. Clearly increasing the overlap ratio
and/or decreasing the initial mask phase uncertainty speed up  convergence. The straight line  feature of the semi-log plots indicates
geometric convergence and vice versa. 

\subsubsection{Noise-agnostic methods}

As an extension of the augmented Lagrangian \eqref{AL'}, consider 
 \beq\nn
 \cL(z,\nu,x,\lamb)&=& \II_Y (z)+\lamb^*(z- \cF(\nu,x))+{1\over 2} \|z-\cF(\nu,x)\|^2
 \eeq
 and the following ADMM scheme
 \beq
 z_{k+1}&=& \arg\min_z  \cL(z,\mu_k,x_k,\lamb_k)= P_Y \lt[\cF(\mu_k, x_k)-\lamb_k\rt]\\
 \label{1001} (\mu_{k+1}, x_{k+1})&=&\arg\min_{\nu} \cL(z_{k+1},\nu,x,\lamb_k) \\
 \lamb_{k+1}&=& \lamb_k+z_{k+1}-\cF(\mu_{k+1}, x_{k+1}). 
 \eeq
 If instead of the bilinear optimization step \eqref{1001}, we simplify it by one-step alternating minimization 
 \beqn
 \mu_{k+1}&= &\arg\min_{\nu} \cL(z_{k+1},\nu,x_k,\lamb_k) =  B^+_{k}(z_{k+1}+\lamb_k)\\
 x_{k+1}&=&\arg\min_g \cL(z_{k+1},\mu_{k+1},x,\lamb_k)= A^+_{k+1}(z_{k+1}+\lamb_k)
 \eeqn
with $B_k := B_{x_k}$ and $A_{k+1}=A_{\mu_{k+1}},$ then we obtain the DM algorithm for blind ptychography
\cite{TDB09},\cite{DM08}, one of the earliest methods for blind ptychography.

\subsubsection{ Extended RAAR} 
To extend RAAR to blind ptychography, let us consider  the augmented Lagrangian
 \beqn
 \cL(y, z,\nu,x,\lamb)&=& \II_Y (z)+\half \|y- \cF(\nu,x)\|^2+\lamb^*(z-y)+{\gamma\over 2} \|z-y\|^2
 \eeqn
 and the following ADMM scheme
  \beq 
\label{step1} (y_{k+1},x_{k+1})&=& \arg\min_y \cL(y, z_k,  x, \mu_{k}, \lamb_k) \\
 \label{step2} (z_{k+1}, \mu_{k+1})&=& \arg\min_z  \cL(y_{k+1}, z, x_{k+1}, \nu, \lamb_k)\\
\label{step3}  \lamb_{k+1}&=& \lamb_k+\gamma (z_{k+1}-y_{k+1}). 
 \eeq
 In the case of a known mask $\mu_k=\mu$ for all $k$, the procedure \eqref{step1}-\eqref{step3} is equivalent to RAAR. 
 We refer to the above scheme as the {\em extended} RAAR (eRAAR). 
 Note that eRAAR has a non-standard loss function as the term $\|y- \cF(\nu,\xh)\|^2$ is not separable. 
 {A similar scheme is implemented in \cite{Marchesini2016SHARP} in the domain of the masked object (see the discussion in Section \ref{sec:Fourier-object})}. 
 
 With $\beta$ given in \eqref{beta}
the minimizer for \eqref{step1} can be expressed explicitly as
\beq
y_{k+1}&=&\lt(I+P_k^\perp/\gamma\rt)^{-1}(z_k+\lamb_k/\gamma)=\lt(I-\beta P^\perp_k\rt) (z_k+\lamb_k/\gamma)\\
x_{k+1}&=& A_k^+ y_{k+1}=A_k^+ (z_k+\lamb_k/\gamma)\label{obj-update}
\eeq
 where  $A_k=A_{\mu_k}$ and $P_k=A_kA_k^+$. 
On the other hand,  Eq. \eqref{step2} can be solved exactly by
 \beq
z_{k+1}& = &P_Y \lt[y_{k+1}-\lamb_k/\gamma\rt]\\
\mu_{k+1}&=& B_{k+1}^+ y_{k+1} \label{mask-update}
\eeq
where $B_{k+1}=B_{x_{k+1}}$. 

Let
\beq
u_{k+1}:= y_{k+1}-\lamb_k/\gamma\label{704}
\eeq
and hence 
\beqn
u_{k+1}=(I-\beta P_k^\perp) (P_Y u_k+\lamb_k/\gamma)-\lamb_k/\gamma.
\eeqn
On the other hand, we can rewrite \eqref{step3} as
\beq
\label{15}
\lamb_{k}/\gamma=z_{k}-u_{k}=P_Y u_{k}-u_{k}
\eeq
and hence
\beq
\nn u_{k+1}&=&(I-\beta P^\perp_k) P_Y u_k-\beta P_k^\perp \lamb_k/\gamma\\
&=& (I-\beta P_k^\perp) P_Y  u_k+\beta P_k^\perp (I-P_Y) u_k\nn\\
&=& \beta u_k+(1-2\beta) P_Y u_k+\beta P_k R_Y u_k\label{16}
\eeq
where $R_Y=2P_Y-I$. This  is the RAAR map with the mask estimate $\mu_k$ updated by \eqref{mask-update} and \eqref{obj-update}. 

More explicitly, by \eqref{15} and \eqref{704}
\beq
y_{k+1}= u_{k+1}+P_Y u_k-u_k\nn
\eeq
and hence 
\beq
\label{19} x_{k+1}&=&A_k^+ (u_{k+1}+P_Y u_k-u_k )\\
\label{20} \mu_{k+1}&=& B_{k+1}^+ (u_{k+1}+P_Y u_k-u_k). 
\eeq
Eq. \eqref{19} can be further simplified as
\beq
\label{21} x_{k+1}&=& A_k^+R_Y u_k
\eeq
by applying $A^+_k$ to \eqref{16} to get $A^+_k u_{k+1}=A^+_k P_Y u_k$.

Eq. \eqref{16}, \eqref{21} and \eqref{20} constitute a simple,  self-contained iterative system called the extended RAAR (eRAAR).

\begin{figure}
\centering 
\subfigure[$50\%$ overlap; $\delta = 9/20$ ]{\includegraphics[width=5cm]{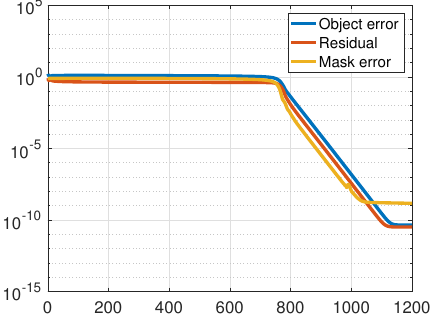}}
\subfigure[$66\%$ overlap; $\delta = 2/5$ ]{\includegraphics[width=5cm,height=3.7cm]{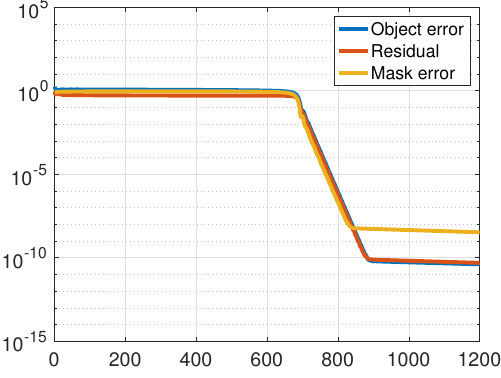}}
\subfigure[$75\%$ overlap; $\delta = 1/2$ ]{\includegraphics[width=5cm]{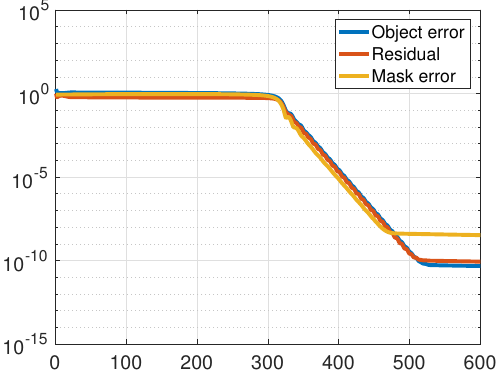}}
\caption{Relative errors versus iteration of blind ptychography for CiB by eRAAR with $\beta=0.8$. } 
\label{fig:eraar}
\end{figure}

Figure \ref{fig:eraar} shows the relative errors (for object and mask) and residual of eRAAR with $\beta=0.8$  corresponding to $\rho=1/3$ according to \eqref{beta2}. The rest of the set-up is the same as for Figure \ref{fig:edrs}. Comparing Figures \ref{fig:edrs} and \ref{fig:eraar} 
we see that eGaussian-DRS converges significantly faster than eRAAR, consistent with the results in Figure \ref{fig:raar}. 

\subsection{Further extensions of blind ptychography algorithms}

\subsubsection{One-loop version}
Let $T_k$ denote the $k$-th RAAR map \eqref{16} or Gaussian-DRS map \eqref{55}.
Starting with the initial guess $u_1$, let 
 \beq
 u_{k+1}&=& T_k^{\ell} (u_k)\quad \mbox{ for sufficiently large $\ell$} \label{25}
 \eeq
 for $k\ge 1$. The termination rule can be based on a predetermined number of iterations, the residual or combination of both. 
 
 Let \beq
\label{26-2} x_{k+1}&=& A_k^+R_Y u_k\\
\label{27-2} \mu_{k+1}&=& B_{k+1}^+ (u_{k+1}+P_Y u_k-u_k). 
\eeq
in the case of RAAR \eqref{16}
\beq
 \label{56-2} \mu_{k+1}&=& B_k^+ A_{k}A_{k}^+ u_k \\
\label{57-2} x_{k+1}&=&A_{k+1}^+ u_{k+1} 
\eeq
in the case of Gaussian-DRS \eqref{55}. 

\begin{algorithm}[h]
\caption{One-loop method}
\begin{algorithmic}[1]
\STATE Input: initial mask guess $\mathbf{\nu}_1$ using MPC  and random object guess $x_1$. 
\STATE Update the object estimate:  
$ x_{k+1} $ is given by  \eqref{25}  with \eqref{26-2} for RAAR or with \eqref{57-2}  for Gaussian/Poisson-DRS; 
\STATE Update the mask estimate:  
$ \mu_{k+1}$ is given by \eqref{27-2} for RAAR or \eqref{56-2} for Gaussian/Poisson-DRS.  
\STATE Terminate if $\||B_{k+1}\mu_{k+1}|- b\| $ stagnates or is less than tolerance; otherwise, go back to step 2 with $k\rightarrow k+1.$
\end{algorithmic}
\end{algorithm}

In a sense, eGaussian-DRS/eRAAR is the one-step version of one-loop Gaussian-DRS/RAAR.

\subsubsection{Two-loop version}

Two-loop methods have  two inner loops: the first is the object loop \eqref{25}-\eqref{26-2}
and the second is the mask loop defined as follows. Two-loop version is an example of
Alternating  Minimization (AM). 

Let $Q_k=B_kB_k^+$ and let $S_k$ be the associated RAAR map:
\beq
S_k(v) &:=& \beta v+(1-2\beta) P_Y v+\beta Q_k R_Yv \nn
 \eeq
 or the associated Gaussian-DRS map
 \beq
S_k(v) &=& {v\over \rho+1}+{\rho-1\over \rho+1} Q_k^+v+{1\over \rho+1} P_Y(2Q_k^+v -v) \nn
 \eeq
 
Starting with the initial guess $v_1$, let 
 \beq
v_{k+1}&=& S_k^{\ell} (v_k)\quad \mbox{ for sufficiently large $\ell$} \label{25'}
 \eeq
 for $k\ge 1$.
 
 Let \beq
\label{27'} \mu_{k+1}&=& B_{k}^+ R_Y v_k,  
\eeq
in the case of RAAR in analogy to \eqref{26-2} and
\beq
\mu_{k+1}&=& B_{k+1}^+ v_{k+1}\label{157}
\eeq
in the case of Gaussian-DRS in analogy to \eqref{57-2}.

\begin{algorithm}[h]
\caption{Two-loop method}\label{alg: suedo algorithm}
\begin{algorithmic}[1]
\STATE Input: initial mask guess ${\nu}_1$ using MPC  and random object guess $x_1$.  
\STATE Update the object estimate:
$
x_{k+1} $ is given by  \eqref{25}  with \eqref{26-2} for RAAR or with \eqref{57-2}  for Gaussian/Poisson-DRS; 
\STATE  Update the mask estimate:  
{$
\mu_{k+1} $ is given by \eqref{25'} with \eqref{27'} for RAAR or with \eqref{157} for Gaussian/Poisson-DRS;}  
\STATE Terminate if $\||B_{k+1}\mu_{k+1}|- b\| $ stagnates or is less than tolerance; otherwise, go back to step 2 with $k\rightarrow k+1.$
\end{algorithmic}
\end{algorithm}

\subsubsection{Two-loop experiments}
Following \cite{DRS-ptycho}, we refer to the two-loop version with Gaussian- or Poisson-DRS as {\em DRSAM}
which is tested next.  
 We demonstrate that even with the parameter $\rho=1$  far from the optimal value (near 0.3),  DRSAM converges geometrically under  the minimum conditions required by uniqueness, i.e. with overlap ratio slightly above $50\%$  and initial mask phase uncertainty $\delta=1/2$. 
  We let  $\delta_{k}^1$ and $\delta^2_l$ in the rank-one scheme \eqref{rank1} and 
 $\delta_{kl}^1$ and $\delta^2_{kl}$ in the full-rank scheme \eqref{rank2} to be {i.i.d. uniform random variables  over $\lb -4,4\rb$}.

The inner loops of Gaussian DRSAM become
\beqn
u_{k}^{l+1} &= &\frac{1}{2}u_{k}^{l} +\frac{1}{2}b\odot \sgn\big(R_ku_{k}^{l}\big)\\
v_{k}^{l+1}& =& \frac{1}{2} v_{k}^{l}+\frac{1}{2}b\odot \sgn{\Big(S_k v_{k}^{l}\Big)}. 
\eeqn
and the inner loops of the Poisson DRSAM become
 \beq\nn u_{k}^{l+1}& =& \frac{1}{2}u_{k}^{l} -\frac{1}{3}R_k u_{k}^{l}+\frac{1}{6} \sgn{\Big(R_k u_{k}^{l}\Big)}\odot \sqrt{|R_k u^l_k|^2+24b^2} \\
 v_{k}^{l+1} &=& \frac{1}{2} v_{k}^{l} -\frac{1}{3}S_k v_{k}^{l}+\frac{1}{6}\sgn{\Big(S_k v_{k}^{l}\Big)}\odot \sqrt{|S_k v_{k}^{l}|^2+24b^2}.\nn
\eeq
Here 
$R_k=2P_k-I$ is the  reflector corresponding to the projector $P_k:=A_kA_k^+ $ and $S_k$ is the reflector corresponding to the projector  $Q_k:=B_kB_k^+ $. 
We set  $u^1_k=u^{\infty}_{k-1}$ where $u^{\infty}_{k-1}$ is the terminal value
at epoch $k-1$ and $v^1_k=v^{\infty}_{k-1}$ where $v^{\infty}_{k-1}$ is the terminal value
at epoch $k-1$.

Figure \ref{fig:noise}(a) compares performance of four combinations of loss functions (Poisson or Gaussian) and scanning schemes (Rank 1 or full-rank) with a $60\times 60$ random mask for the test object CiB in the noiseless case.  Full-rank perturbation \eqref{rank2}  results in a faster convergence rate than rank 1 scheme \eqref{rank1}. The convergence rate of Poisson DRSAM is slightly better than Gaussian DRSAM with noiseless data. 

\begin{figure}
\centering 
\subfigure[RE vs. epoch]{\includegraphics[width=6cm]{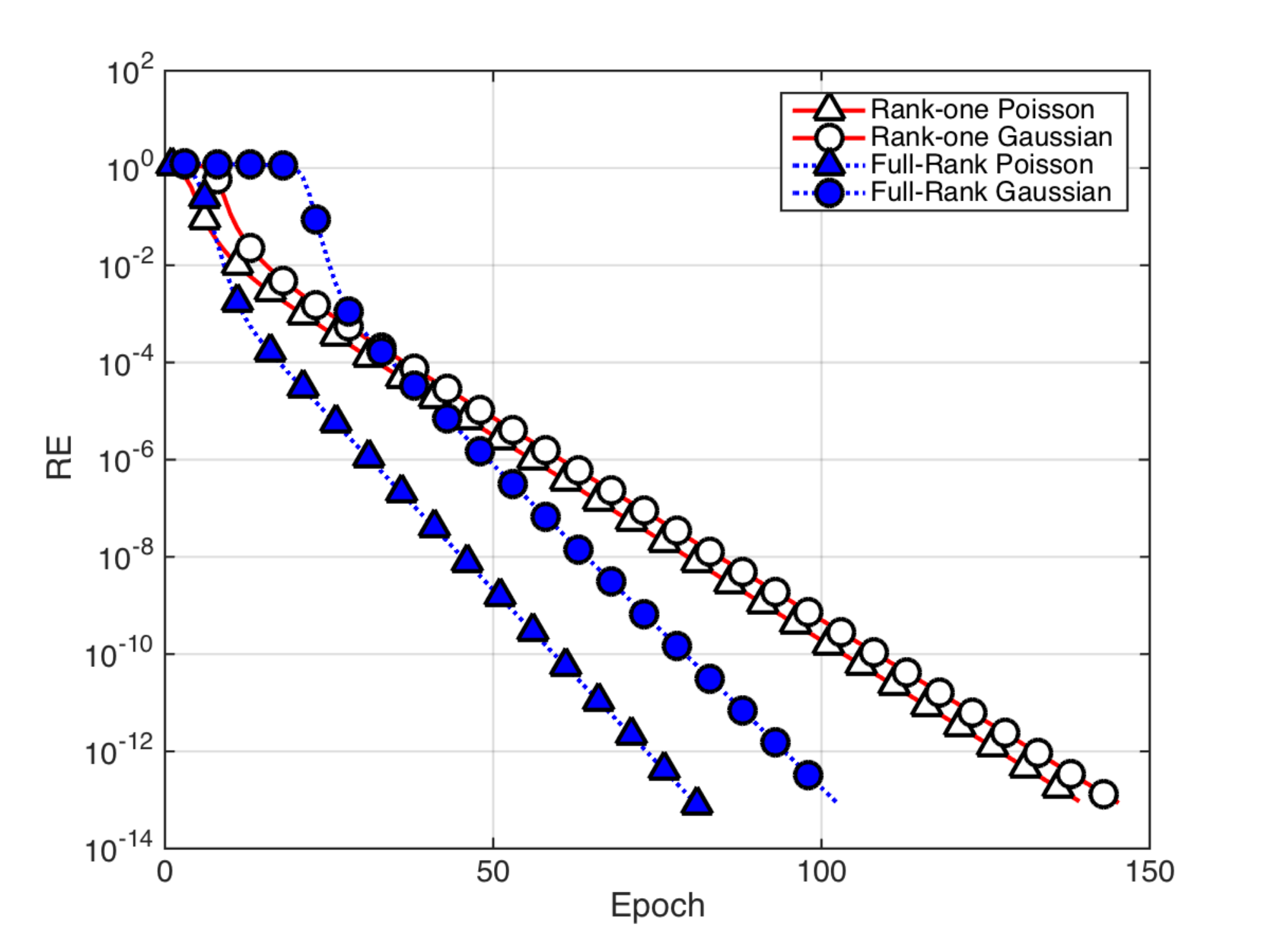}}
\subfigure[RE vs. NSR]{\includegraphics[width=6cm]{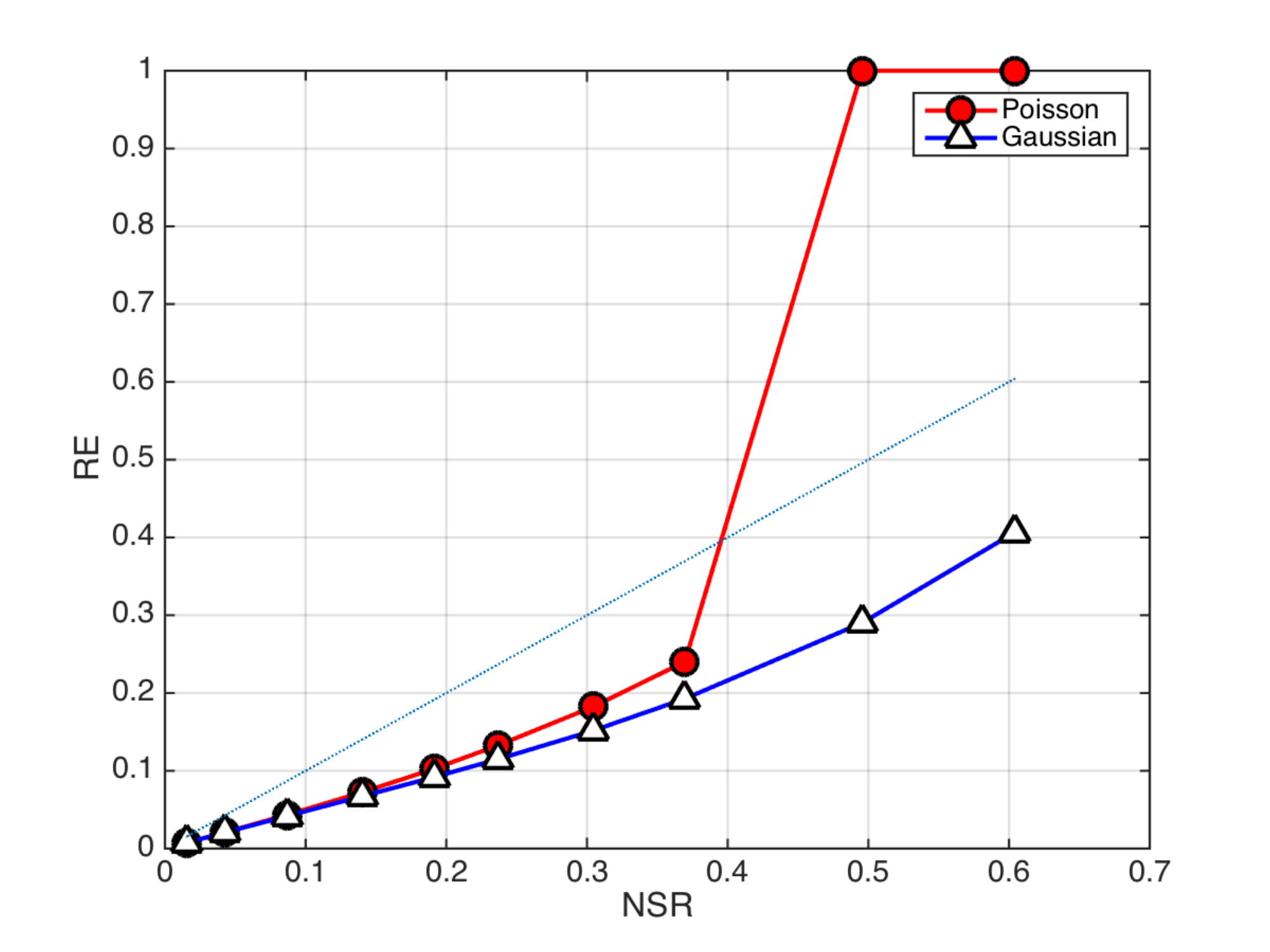}}
\caption{(a) Geometric convergence to CiB in the noiseless case at various rates for four combinations of loss functions and scanning schemes with i.i.d. mask (rank-one Poisson, $\mbox{rate}=0.8236$; rank-one Gaussian, $\mbox{rate}=0.8258$; full-rank Poisson, $\mbox{rate}=0.7205$; full-rank Gaussian, $\mbox{rate}=0.7373$) and (b) RE versus NSR for reconstruction of CiB with Poisson noise.}
\label{fig:noise}
\end{figure}

With data corrupted by by Poisson noise, Figure \ref{fig:noise}(b) shows RE versus NSR \eqref{nsr} for CiB  by Poisson-DRS and Gaussian-DRS with i.i.d. mask and the full-rank scheme. The maximum number of epoch in DRSAM is limited to $100$. The RR stabilizes  usually  after 30 epochs. The (blue) reference straight line has slope = 1. We see that the Gaussian-DRS outperforms the  Poisson-DRS, especially when  the Poisson RE becomes unstable for NSR $\ge 35\%$.  As noted in \cite{rPIE17},\cite{adaptive},\cite{AP-phasing} fast convergence (with the Poisson log-likelihood function) may introduce noisy artifacts and reduce reconstruction quality.

\section{Holographic coherent diffraction imaging}\label{s:holo}

Holography is a  lensless imaging technique that enables complex-valued image reconstruction by virtue of  placing a coherent point source at an appropriate distance from the object and having the object field interfere with the reference wave produced by this point source at the (far-field) detector plane~\cite{goodman2005introduction}. For example, adding a pinhole (corresponding to adding a delta distribution in the mathematical model) at an appropriate  position  to the sample
creates an additional wave in the far field, with a tilted phase, caused by the displacement between the pinhole and the sample. 
The far field detector now records the intensity of the Fourier transform of the sample and the reference signal (e.g., the pinhole).

The invention of holography goes back to Dennis Gabor\footnote{Gabor devoted a lot of his time and energy to overcome the initial skepticism of the community to the concept of holography  and proudly noted in a letter to Bragg  ``I have also perfected the experimental arrangement considerably, and now I can produce really pretty reproductions of the original from apparently hopelessly muddled diffraction diagrams.''~\cite{johnston2005white}.}, who in 1947 was working on improving the resolution of the recently invented
electron microscope~\cite{dennis1956improvements,gabor1948new,gabor1965reconstruction}.  In 1971, he was awarded the Nobel Prize in Physics for his invention.
 In the original scheme proposed by Gabor, called in-line holography, the reference and object waves are parallel to one another. In off-axis holography, the two waves are separated by a non-zero angle.
In classical holography, a photographic plate is used to record the spatial intensity distribution. In state-of-the-art digital holography systems  a
digital acquisition device  captures the spatial intensity distribution~\cite{seelamantula2011exact}.

We recommend~\cite{Latychevskaia19} for a recent survey on iterative algorithms in holography.
While holography leads to relatively simple algorithms for solving the phase retrieval problems, it does pose numerous challenges in the experimental practice. For a detailed discussion of various practical issues with holography, such as resolution limitations, see~\cite{duadi2011digital,latychevskaia2015practical,shechtman2015phase,saliba2016novel,Latychevskaia19}.


A compelling direction in holographic phase retrieval is to combine holography with CDI~\cite{latychevskaia2012holography,saliba2012fourier,raz2014direct}, see Figure~\ref{fig:holo} for a setup depicting holographic CDI.
This hybrid technique  ``inherits  the  benefits  of  both  techniques, namely the straightforward unambiguous recovery of the phase distribution and the visualization of a non-crystalline object at the highest possible resolution,''~\cite{latychevskaia2012holography}.
Researchers have recently successfully used holographic CDI to image proteins at the single-molecule level~\cite{Longchamp1474}.

\begin{figure}
\begin{center}
\includegraphics[height=50mm]{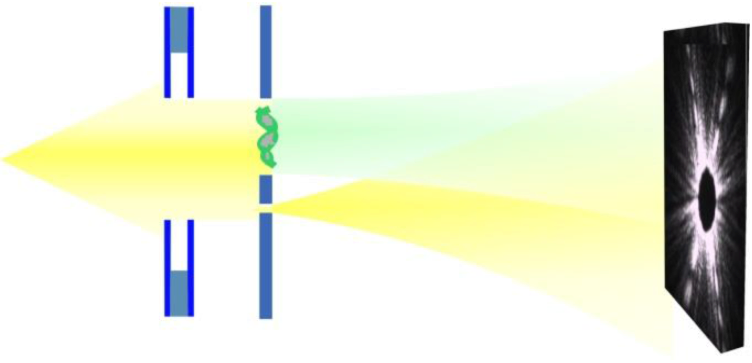}
\caption{Holographic CDI setup. Image courtesy of~\protect\cite{saliba2012fourier}.}
\label{fig:holo}
\end{center}
\end{figure}

While holographic techniques have been around for a long time, these investigations have been mainly empirical. A notable exception is the recent work~\cite{barmherzig2019holographic,barmherzig2019dual}, which contains a rigorous mathematical treatment of holographic CDI that sheds light on the reference design from an optimization viewpoint and provides a detailed error analysis. We will discuss some aspects of this work below.

From a mathematical viewpoint, the key point of holographic CDI  is that the introduction of a reference signal simplifies the phase retrieval problem
considerably, since  the computational problem of recovering the desired signal can now be expressed as a linear deconvolution problem~\cite{kikuta1972x,guizar2007holography,barmherzig2019holographic}. We discuss this insight below.

Here, we assume that our function of interest $\xo$ is an $n \times n$ image. 
We  denote the convolution of two functions $x,z$ by $x \ast z$ and define the involution { (a.k.a. the twin image) } $\check{x}$ of  $x$ as $\check{x}(t_1,t_2) = \overline{x(-t_1,-t_2)}$.
The cross correlation $\mathfrak{C}_{[x,z]}$ between the two functions $x, z$ is given by
\begin{equation}\label{crosscorr}
\mathfrak{C}_{[x,z]} := x \ast \check{z},
\end{equation}
where we use Dirichlet boundary conditions, i.e., zero-padding, outside the valid index range.
We already encountered the special case $x = z$ (although without stipulating specific boundary conditions), in form of  the autocorrelation 
\begin{equation}\label{autocorr}
\auto_{x} = x \ast \check{x}, 
\end{equation}
which is  at core of the phase retrieval problem via the relation\footnote{Arthur Lindo Patterson once asked Norbert Wiener:  ``What do you know about a function, when you know only the amplitudes of its Fourier coefficients?'' Wiener responded:  ``You know the Faltung [convolution]'', \cite{Glusker84}.}
$$F (x \ast \check{x}) = |F(x)|^2.$$

While extracting a function from its autocorrelation is a difficult quadratic problem (as exemplified by the phase retrieval problem),
extracting a function from a cross correlation is a linear problem if the other function is known, and thus much easier.
This observation is the key point of holographic CDI.
We will take full advantage of this fact by adding a reference area (in digital form represented by the signal $r$) to the specimen $\xo$.
For concreteness, we assume that the reference $r$ is placed on the right side
of $\xo$, and subject the so enlarged signal $[\xo, r]$ to the measurement process, as illustrated in Figure~\ref{fig:holo}.

For $(s_1,s_2) \in \{-(n-1),\dots,0\}\times  \{-(n-1),\dots,0\}$ we have
\begin{align}
\mathfrak{C}_{[\xo,r]}(s_1,s_2) = & (\xo \ast \check{r})(s_1,s_2) \notag \\
                                = & ([\xo,r] \ast \widecheck{[\xo,r]})(s_1,n-s_2)  =  \auto_{[x,r]}(s_1,-n+s_2).
\label{cross2auto}
\end{align}
Equation~\eqref{cross2auto} allows us to establish a linear relationship between  $\mathfrak{C}_{[\xo,  r]}$ and the measurements given by  the squared entries of $F(\auto_{[\xo,r]})$.
Most approaches in holography are based on utilizing this relationship in some way, see e.g.~\cite{seelamantula2011exact}.

Here, we take a signal processing approach and recall that the convolution of two 2-D signals with Neumann boundary conditions can be described as matrix-vector multiplication, where the matrix is given by a lower-triangular block Toeplitz matrix with lower-triangular Toeplitz blocks~\cite{gray2006toeplitz}. 
The lower-triangular property stems from the fact that the zero-padding combined with the particular index range we are considering is equivalent to applying a two-dimensional causal filter~\cite{gray2006toeplitz}. 

Let $r^{(k)}$ be the $k$-the column of the reference $r$ and let the lower triangular block-Toeplitz-Toeplitz block matrix $T(r)$ be given by
\begin{equation*}
T(r)  = \begin{bmatrix}
T_0 & 0 & \dots  &0 \\
T_1 & T_1 & 0 & \vdots  \\
\vdots & \ddots & &  \\
T_{n-1} & \dots & & T_1
\end{bmatrix},
\end{equation*}
where the first column of the lower-triangular Toeplitz matrix $T_k$ is given by $\check{r}^{(n-k-1)}$ for $k=0,\dots,n-1$.
We also define $y : =  F^{-1}(|F([\xo,r])|^2)$ and note that 
$$y =  F^{-1}(|F([\xo,r])|^2) = F^{-1}(F(\auto_{[\xo,r]}))=\auto_{[\xo,r]}.$$
Hence, with a slight abuse of notation (by considering $\xo$ also as column vector of length $n^2$ via stacking its columns)
we arrive at the following linear system of equations
\begin{equation}
\label{toepconv}
T(r) \xo = y.
\end{equation}
The $n^2 \times n^2$ matrix $T(r)$ is invertible if and only if its diagonal entries are non-zero, that is, if and only if $r_{n-1,n-1} \neq 0$.
As noted in~\cite{barmherzig2019holographic}, this condition is equivalent to the well-known holographic separation condition~\cite{guizar2007holography}, which dictates when an image is recoverable via using the reference $r$. In signal processing jargon, this separation condition prevents the occurrence of aliasing.

Let us consider the very special case of the pinhole reference. In this case $r \in \CC^{n\times n}$ is given by 
$$
r_{k,l} = \begin{cases}
1, & \text{if $k=l=n-1$,} \\
0, & \text{else.}
\end{cases}
$$
Thus $r$ acts as a delta-distribution with respect to the given digital resolution (which may be very difficult to realize in practice, and thus this is still one limiting factor in the achievable image resolution). In this particular case its diagonal entries are 
$[T(r)]_{k,k} =r_{n-1,n-1}=1$ for all $k=0,\dots,n^2-1$, and all off-diagonal entries of $T(r)$ are zero; thus $T(r)$ is simply the $n^2 \times n^2$ identity matrix. 

Other popular choices are the block reference defined by $r_{k,l}=1$ for all $k,l=0,\dots,n-1$; and the slit reference defined by 
$$
r_{k,l} = \begin{cases}
1, & \text{if $l=n-1$,} \\
0, & \text{else.}
\end{cases}
$$
In both cases the resulting matrix $T(r)$ as well as its inverse $[T(r)]^{-1}$ take a very simple form, as the interested reader may easily convince herself. 

In the noiseless case, the only difference between these references from a theoretical viewpoint is the computational complexity in solving the system~\eqref{toepconv}, which is obviously minimal for the pinhole reference. However, in the presence of noise different references have different advantages and drawbacks.
We refer to~\cite{barmherzig2019holographic} for a thorough error analysis when the measurements are corrupted by Poisson shot noise.

\bigskip

We describe some numerical experiments illustrating the effectiveness of the referenced deconvolution algorithm.
The description of these simulations and associated images are courtesy of~\cite{barmherzig2019holographic}, which also contains a number of other simulations.

\begin{figure}[!htbp] \label{fig:mimi-recov}
\begin{center}
    \subfigure[Ground-truth image]{
        \includegraphics[width=0.31\textwidth]{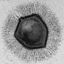}}
    \subfigure[Fourier magnitude of the groundtruth]{
        \includegraphics[width=0.31\textwidth]{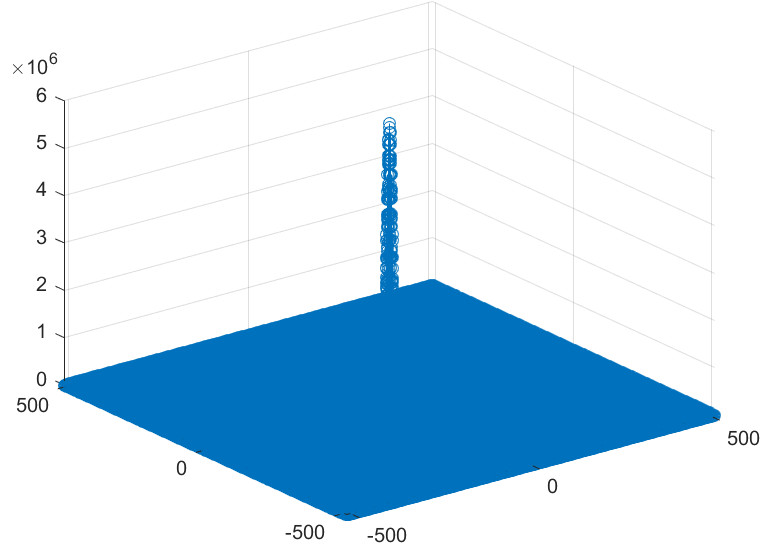}}
            \subfigure[HIO (no ref.)  $\eps = 93.794$, $\EE(\eps)$ NA]{
       \includegraphics[width=0.31\textwidth]{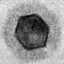}}
          \subfigure[HIO with block ref.\  $\eps = 42.813$, $\EE(\eps)$ NA]{
        \includegraphics[width=0.31\textwidth]{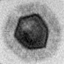}}
        \subfigure[HIO with slit ref.\ $\eps = 102.28$, $\EE(\eps)$ NA]{
        \includegraphics[width=0.31\textwidth]{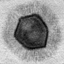}}
    \subfigure[HIO with pinhole ref.\ $\eps = 168.18$, $\EE(\eps)$ NA]{
        \includegraphics[width=0.31\textwidth]{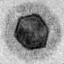}}
       \subfigure[Ref.Deconv.\ with block ref.\ $\eps = 3.703$, $\EE(\eps) = 3.795$]{
        \includegraphics[width=0.31\textwidth]{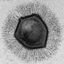}}
        \subfigure[Ref.~Deconv.\ with slit ref.\ $\eps = 5.720$, $\EE(\eps) = 5.147$]{
        \includegraphics[width=0.31\textwidth]{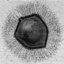}}
        \subfigure[Ref.~Deconv.\ with pinhole ref.\ $\eps = 46.97$, $\EE(\eps) = 63.84$]{
        \includegraphics[width=0.31\textwidth]{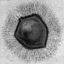}}
        \caption{Recovery result of the mimivirus image using various recovery schemes, and the corresponding relative recovery errors (all errors should be rescaled by $10^{-4}$). Referenced deconvolution clearly outperforms  HIO, both with and without the reference information enforced. Experimental and theoretical relative errors for referenced deconvolution closely match, as predicted by the theory derived in~\protect\cite{barmherzig2019holographic}. .}
    \end{center}
\end{figure}

In this experiment, the specimen $\xo$ is the mimivirus image \cite{Mimivirus}, and its spectrum mostly concentrates on very low frequencies, as shown in Figure~\ref{fig:mimi-recov}(b). The image size is $64 \times 64$, and the pixel values are normalized to $[0, 1]$. For the referenced setup, a reference $r$ of size $64 \times 64$ is placed next to $\xo$, forming a composite specimen $[\xo,r]$ of size $64 \times 128$.  Three references, i.e., the pinhole, the slit, and the block references, are considered. Note that the zero-padding introduced as boundary condition in the cross correlation function~\eqref{crosscorr} and the autocorrelation function~\eqref{autocorr} corresponds to an oversampling of the associated Fourier transform. In this experiment, the oversampled Fourier transform is taken to be of size $1024 \times 1024$, and the collected noisy data
are subject to Poisson shot noise.  We note that since the oversampling condition in  the  detector  plane  corresponds  to  zero-padding  in  the  object  plane, this requires  the  specimen  to  be surrounded  by  a  support  with  known  transmission properties.  For  instance,  when  imaging  a  biological molecule,  it  must  ideally  be  either  levitating  or  resting  on  a  homogeneous  transparent  film such  as graphene~\cite{latychevskaia2012holography}.
Thus, what is trivial to do from a mathematical viewpoint, may be rather challenging to realize in a practical experimental.

We run the referenced deconvolution algorithm and compare it to the HIO algorithm, the latter with and without enforcing the known reference for comparison. The results are presented in Figure~\ref{fig:mimi-recov}. It is evident that referenced deconvolution clearly outperforms HIO.
An inspection of the errors stated in the corresponding figure captions shows that for the referenced deconvolution schemes, the expected and empirical relative recovery errors are close for each reference, as predicted by the error analysis in~\cite{barmherzig2019holographic}.

In the example depicted in  Figure~\ref{fig:mimi-recov} the block reference gives the smallest recovery error among the tested reference schemes. However, this is not the case in general. As illustrated in~\cite{barmherzig2019holographic}
depending on the spectral decay behavior of the image under consideration, different reference schemes have different limitations.
To overcome the specific limitations of each reference, a {\em dual reference} approach has been proposed in~\cite{barmherzig2019dual}, in which 
the reference consists of two reference portions -- a pinhole portion $r_p$ and a block portion $r_b$. In this case the illuminated image takes the form
$\begin{bmatrix} \xo & r_p \\ r_b & {\bf 0} \end{bmatrix}$. The theoretical and empirical error analysis in~\cite{barmherzig2019dual} show that
this dual-reference scheme achieves a smaller recovery error than the leading single-reference schemes.

\section{Conclusion and outlook}\label{s:outlook}

In this survey we have tried to capture the state of the art  of the classical and at the same time fast-emerging field of numerical algorithms for phase retrieval.
The last decade has witnessed extensive activities in the systematic study of numerical algorithms for phase retrieval. 
Advances in convex and non-convex optimization have led to a better understanding of the benefits and limitations of various phase retrieval algorithms. The insights gained in the study of these algorithms in turn has advanced new measurement protocols, such as random illuminations. 

Some of the most challenging problems related to phase retrieval arise in blind ptychography, in imaging proteins at the single-molecule level~\cite{Longchamp1474}, and in non-crystallographic ``single-shot'' x-ray imaging~\cite{chapman2007femtosecond,loh2010cryptotomography}. In the latter problem, in addition to the phase retrieval problem one faces the major task of tomographic 3D reconstruction of the object  from the diffraction images  with unknown rotation angles -- a challenge that we also encounter in Cryo-EM~\cite{singer2018mathematics}. The review article~\cite{shechtman2015phase} contains a detailed discussion of current bottlenecks and future challenges, such as taking the CDI techniques to the regime of attosecond science. This topic remains one of the current challenges in phase retrieval.

Mathematicians sometimes develop theoretical and algorithmic  frameworks under assumptions that do not conform to current practice. It is then important to find out if these assumptions are fundamentally unrealistic, or if they actually point to new ideas that are (perhaps with considerable effort) implementable in practice and advance the field.

It is clear that much more work needs to be done and a closer dialogue between practitioners and theorists is highly desirable to create the kind of feedback loop where theory and practice drive each other forward with little temporal delay. Careful systematic numerical analysis is an essential ingredient in strengthening the bond between theory and practice.

\subsection*{Acknowledgements}

The authors are grateful to  Tatiana Latychevskaia, Emmanuel Candes, Justin Romberg, and Stefano Marchesini for
allowing us to use the  exquisite illustrations from their corresponding publications, see
~\cite{saliba2012fourier}, \cite{barmherzig2019holographic}, \cite{bahmani2016phase}, and \cite{qian2014efficient},
respectively. We thank Dr. Pengwen Chen for preparing Figures \ref{fig:init1}, \ref{fig:init2}, \ref{fig:edrs} and \ref{fig:eraar}. 
A.F.\ acknowledges support from the NSF via grant
NSF DMS-1413373 and from the Simons Foundation via grant SIMONS FDN 2019-24. 
T.~S.\ acknowledges support from the NSF via grant DMS 1620455 and from the NGA and the NSF via grant DMS
1737943.


\begin{thebibliography}{100}

\bibitem{ahmed2013blind}
Ali Ahmed, Benjamin Recht, and Justin Romberg.
\newblock Blind deconvolution using convex programming.
\newblock {\em IEEE Transactions on Information Theory}, 60(3):1711--1732,
  2013.

\bibitem{appleby2005symmetric}
D~Marcus Appleby.
\newblock Symmetric informationally complete--positive operator valued measures
  and the extended clifford group.
\newblock {\em Journal of Mathematical Physics}, 46(5):052107, 2005.

\bibitem{appleby2019tight}
Marcus Appleby, Ingemar Bengtsson, Steven Flammia, and Dardo Goyeneche.
\newblock Tight frames, hadamard matrices and zauner’s conjecture.
\newblock {\em Journal of Physics A: Mathematical and Theoretical},
  52(29):295301, 2019.

\bibitem{arridge_maass_2019}
Simon Arridge, Peter Maass, Ozan \"{O}ktem, and Carola-Bibiane Sch\"{o}nlieb.
\newblock Solving inverse problems using data-driven models.
\newblock {\em Acta Numerica}, 28:1–174, 2019.

\bibitem{bahmani2016phase}
Sohail Bahmani and Justin Romberg.
\newblock Phase retrieval meets statistical learning theory: A flexible convex
  relaxation.
\newblock {\em Electronic Journal of Statistics}, 11(2):5254--5281, 2016.

\bibitem{BBC09}
R.~Balan, B.~Bodmann, P.G. Casazza, and D.~Edidin.
\newblock Painless reconstruction from magnitudes of frame coefficients.
\newblock {\em J.\ Four.\ Anal.\ Appl.}, 15:488--501, 2009.

\bibitem{BCE07}
R.~Balan, P.G. Casazza, and D.~Edidin.
\newblock Equivalence of reconstruction from the absolute value of the frame
  coefficients to a sparse representation problem.
\newblock {\em IEEE Sig. Proc. Letters}, 14(5):341--343, 2007.

\bibitem{balan2010signal}
Radu Balan.
\newblock On signal reconstruction from its spectrogram.
\newblock In {\em 2010 44th Annual Conference on Information Sciences and
  Systems (CISS)}, pages 1--4. IEEE, 2010.

\bibitem{Balan1}
Radu Balan, Pete Casazza, and Dan Edidin.
\newblock On signal reconstruction without phase.
\newblock {\em Applied and Computational Harmonic Analysis}, 20(3):345--356,
  2006.

\bibitem{bandeira2014saving}
Afonso~S Bandeira, Jameson Cahill, Dustin~G Mixon, and Aaron~A Nelson.
\newblock Saving phase: Injectivity and stability for phase retrieval.
\newblock {\em Applied and Computational Harmonic Analysis}, 37(1):106--125,
  2014.

\bibitem{barmherzig2019dual}
David~A Barmherzig, Ju~Sun, Emmanuel~J Cand{\`e}s, TJ~Lane, and Po-Nan Li.
\newblock Dual-reference design for holographic coherent diffraction imaging.
\newblock {\em arXiv preprint arXiv:1902.02492}, 2019.

\bibitem{barmherzig2019holographic}
David~A Barmherzig, Ju~Sun, TJ~Lane, Po-Nan Li, and Emmanuel~J Cand{\`e}s.
\newblock Holographic phase retrieval and reference design.
\newblock {\em arXiv preprint arXiv:1901.06453}, 2019.

\bibitem{BCL04}
Heinz~H Bauschke, Patrick~L Combettes, and D~Russell Luke.
\newblock Finding best approximation pairs relative to two closed convex sets
  in hilbert spaces.
\newblock {\em Journal of Approximation Theory}, 127(2):178--192, 2004.

\bibitem{Beck98}
C.~Beck and R.~D'Andrea.
\newblock Computational study and comparisons of {LFT} reducibility methods.
\newblock In {\em Proceedings of the American Control Conference}, pages
  1013--1017, 1998.

\bibitem{becker2011templates}
Stephen~R Becker, Emmanuel~J Cand{\`e}s, and Michael~C Grant.
\newblock Templates for convex cone problems with applications to sparse signal
  recovery.
\newblock {\em Mathematical programming computation}, 3(3):165, 2011.

\bibitem{beinert2017sparse}
Robert Beinert and Gerlind Plonka.
\newblock Sparse phase retrieval of one-dimensional signals by {Prony's}
  method.
\newblock {\em Frontiers in Applied Mathematics and Statistics}, 3:5, 2017.

\bibitem{bendory2017fourier}
Tamir Bendory, Robert Beinert, and Yonina~C Eldar.
\newblock Fourier phase retrieval: Uniqueness and algorithms.
\newblock In {\em Compressed Sensing and its Applications}, pages 55--91.
  Springer, 2017.

\bibitem{Poisson2}
Liheng Bian, Jinli Suo, Jaebum Chung, Xiaoze Ou, Changhuei Yang, Feng Chen, and
  Qionghai Dai.
\newblock Fourier ptychographic reconstruction using poisson maximum likelihood
  and truncated wirtinger gradient.
\newblock {\em Scientific reports}, 6:27384, 2016.

\bibitem{BSV02}
G.~Bianchi, F.~Segala, and A.~Volcic.
\newblock The solution of the covariogram problem for plane {${\mathcal
  C}^2_+$} convex bodies.
\newblock {\em J.~Differential Geometry}, 60:177--198, 2002.

\bibitem{Bog08}
M.J. Bogan and et~al.
\newblock Single particle {X-ray} diffractive imaging.
\newblock {\em Nano Lett.}, 8(1):310--316, 2008.

\bibitem{BS79}
Y.M. Bruck and L.G. Sodin.
\newblock On the ambiguity of the image reconstruction problem.
\newblock {\em Opt. Comm.}, 30:304--308, 1979.

\bibitem{cai2016optimal}
T~Tony Cai, Xiaodong Li, Zongming Ma, et~al.
\newblock Optimal rates of convergence for noisy sparse phase retrieval via
  thresholded wirtinger flow.
\newblock {\em The Annals of Statistics}, 44(5):2221--2251, 2016.

\bibitem{CanTao06}
E.~J. Cand\`{e}s and T.~Tao.
\newblock Near-optimal signal recovery from random projections: Universal
  encoding strategies.
\newblock {\em IEEE Trans. on Information Theory}, 52:5406--5425, 2006.

\bibitem{CESV2013}
Emmanuel~J Cand\`{e}s, Yonina~C Eldar, Thomas Strohmer, and Vladislav
  Voroninski.
\newblock Phase retrieval via matrix completion.
\newblock {\em SIAM Journal on Imaging Sciences}, 6(1):199--225, 2013.

\bibitem{candes2015phase}
Emmanuel~J Candes, Xiaodong Li, and Mahdi Soltanolkotabi.
\newblock Phase retrieval from coded diffraction patterns.
\newblock {\em Applied and Computational Harmonic Analysis}, 39(2):277--299,
  2015.

\bibitem{CSV2013}
Emmanuel~J Cand\`{e}s, Thomas Strohmer, and Vladislav Voroninski.
\newblock Phaselift: Exact and stable signal recovery from magnitude
  measurements via convex programming.
\newblock {\em Communications on Pure and Applied Mathematics},
  66(8):1241--1274, 2013.

\bibitem{ImprovedPL}
Emmanuel~J Cand{è}s and Xiaodong Li.
\newblock Solving quadratic equations via {PhaseLift} when there are about as
  many equations as unknowns.
\newblock 14(5):1017--1026, 2014.

\bibitem{chandra2017phasepack}
Rohan Chandra, Ziyuan Zhong, Justin Hontz, Val McCulloch, Christoph Studer, and
  Tom Goldstein.
\newblock Phasepack: A phase retrieval library.
\newblock {\em Asilomar Conference on Signals, Systems, and Computers}, 2017.

\bibitem{March19}
Huibin Chang, Pablo Enfedaque, and Stefano Marchesini.
\newblock Blind ptychographic phase retrieval via convergent alternating
  direction method of multipliers.
\newblock {\em SIAM Journal on Imaging Sciences}, 12(1):153--185, 2019.

\bibitem{chapman2011femtosecond}
Henry~N Chapman, Petra Fromme, Anton Barty, Thomas~A White, Richard~A Kirian,
  Andrew Aquila, Mark~S Hunter, Joachim Schulz, Daniel~P DePonte, Uwe
  Weierstall, et~al.
\newblock Femtosecond {X-ray} protein nanocrystallography.
\newblock {\em Nature}, 470(7332):73--77, 2011.

\bibitem{chapman2007femtosecond}
Henry~N Chapman, Stefan~P Hau-Riege, Michael~J Bogan, Sa{\v{s}}a Bajt, Anton
  Barty, S{\'e}bastien Boutet, Stefano Marchesini, Matthias Frank, Bruce~W
  Woods, W~Henry Benner, et~al.
\newblock Femtosecond time-delay x-ray holography.
\newblock {\em Nature}, 448(7154):676--679, 2007.

\bibitem{CMW07}
C.C. Chen, J.~Miao, C.W. Wang, and T.K. Lee.
\newblock Application of the optimization technique to noncrystalline {X-Ray}
  diffraction microscopy: guided hybrid input-output method ({GHIO}).
\newblock {\em Phys. Rev. B.}, 76:064113, 2007.

\bibitem{DRAP-ptycho}
Pengwen Chen and Albert Fannjiang.
\newblock Coded aperture ptychography: uniqueness and reconstruction.
\newblock {\em Inverse Problems}, 34(2):025003, 2018.

\bibitem{FDR}
Pengwen Chen and Albert Fannjiang.
\newblock Fourier phase retrieval with a single mask by douglas--rachford
  algorithms.
\newblock {\em Applied and computational harmonic analysis}, 44(3):665--699,
  2018.

\bibitem{null}
Pengwen Chen, Albert Fannjiang, and Gi-Ren Liu.
\newblock Phase retrieval by linear algebra.
\newblock {\em SIAM Journal on Matrix Analysis and Applications},
  38(3):854--868, 2017.

\bibitem{AP-phasing}
Pengwen Chen, Albert Fannjiang, and Gi-Ren Liu.
\newblock Phase retrieval with one or two diffraction patterns by alternating
  projections with the null initialization.
\newblock {\em Journal of Fourier Analysis and Applications}, 24(3):719--758,
  2018.

\bibitem{chen2017solving}
Yuxin Chen and Emmanuel~J Cand{\`e}s.
\newblock Solving random quadratic systems of equations is nearly as easy as
  solving linear systems.
\newblock {\em Communications on Pure and Applied Mathematics}, 70(5):822--883,
  2017.

\bibitem{chen2019gradient}
Yuxin Chen, Yuejie Chi, Jianqing Fan, and Cong Ma.
\newblock Gradient descent with random initialization: Fast global convergence
  for nonconvex phase retrieval.
\newblock {\em Mathematical Programming}, 176(1-2):5--37, 2019.

\bibitem{Cimmino}
Gianfranco Cimmino.
\newblock Cacolo approssimato per le soluzioni dei systemi di equazioni
  lineari.
\newblock {\em La Ricerca Scientifica (Roma)}, 1:326--333, 1938.

\bibitem{conca2015algebraic}
Aldo Conca, Dan Edidin, Milena Hering, and Cynthia Vinzant.
\newblock An algebraic characterization of injectivity in phase retrieval.
\newblock {\em Applied and Computational Harmonic Analysis}, 38(2):346--356,
  2015.

\bibitem{Cor06}
J.V. Corbett.
\newblock The {P}auli problem, state reconstruction and quantum-real numbers.
\newblock {\em Rep. Math. Phys.}, 57:53--68, 2006.

\bibitem{DF87}
J.C. Dainty and J.R. Fienup.
\newblock Phase retrieval and image reconstruction for astronomy.
\newblock In H.~Stark, editor, {\em Image Recovery: Theory and Application}.
  Academic Press, New York, 1987.

\bibitem{davenport2016overview}
Mark~A Davenport and Justin Romberg.
\newblock An overview of low-rank matrix recovery from incomplete observations.
\newblock {\em IEEE Journal of Selected Topics in Signal Processing},
  10(4):608--622, 2016.

\bibitem{demanet2014stable}
Laurent Demanet and Paul Hand.
\newblock Stable optimizationless recovery from phaseless linear measurements.
\newblock {\em Journal of Fourier Analysis and Applications}, 20(1):199--221,
  2014.

\bibitem{dhifallah2017phase}
Oussama Dhifallah, Christos Thrampoulidis, and Yue~M Lu.
\newblock Phase retrieval via linear programming: Fundamental limits and
  algorithmic improvements.
\newblock In {\em 2017 55th Annual Allerton Conference on Communication,
  Control, and Computing (Allerton)}, pages 1071--1077. IEEE, 2017.

\bibitem{dierolf2010ptychographic}
Martin Dierolf, Andreas Menzel, Pierre Thibault, Philipp Schneider, Cameron~M
  Kewish, Roger Wepf, Oliver Bunk, and Franz Pfeiffer.
\newblock Ptychographic x-ray computed tomography at the nanoscale.
\newblock {\em Nature}, 467(7314):436--439, 2010.

\bibitem{doelman2018solving}
Reinier Doelman, Nguyen~H Thao, and Michel Verhaegen.
\newblock Solving large-scale general phase retrieval problems via a sequence
  of convex relaxations.
\newblock {\em JOSA A}, 35(8):1410--1419, 2018.

\bibitem{Don06}
D.~L. Donoho.
\newblock Compressed sensing.
\newblock {\em IEEE Trans. on Information Theory}, 52(4):1289--1306, 2006.

\bibitem{donoho2010message}
David~L Donoho, Arian Maleki, and Andrea Montanari.
\newblock Message passing algorithms for compressed sensing: I. motivation and
  construction.
\newblock In {\em 2010 IEEE information theory workshop on information theory
  (ITW 2010, Cairo)}, pages 1--5. IEEE, 2010.

\bibitem{dremeau2015phase}
Ang{\'e}lique Dr{\'e}meau and Florent Krzakala.
\newblock Phase recovery from a bayesian point of view: the variational
  approach.
\newblock In {\em 2015 IEEE International Conference on Acoustics, Speech and
  Signal Processing (ICASSP)}, pages 3661--3665. IEEE, 2015.

\bibitem{duadi2011digital}
Hamootal Duadi, Ofer Margalit, Vicente Mico, Jos{\'e}~A Rodrigo, Tatiana
  Alieva, Javier Garcia, and Zeev Zalevsky.
\newblock Digital holography and phase retrieval.
\newblock In {\em Holography, Research and Technologies}. InTech, 2011.

\bibitem{Dua11}
H.~Duadi~et. al.
\newblock Digital holography and phase retrieval.
\newblock In J.~Rosen, editor, {\em Source: Holography, Research and
  Technologies}. InTech, 2011.

\bibitem{eldar2014sparse}
Yonina~C Eldar, Pavel Sidorenko, Dustin~G Mixon, Shaby Barel, and Oren Cohen.
\newblock Sparse phase retrieval from short-time {F}ourier measurements.
\newblock {\em IEEE Signal Processing Letters}, 22(5):638--642, 2014.

\bibitem{elser2018benchmark}
Veit Elser, Ti-Yen Lan, and Tamir Bendory.
\newblock Benchmark problems for phase retrieval.
\newblock {\em SIAM Journal on Imaging Sciences}, 11(4):2429--2455, 2018.

\bibitem{unique}
Albert Fannjiang.
\newblock Absolute uniqueness of phase retrieval with random illumination.
\newblock {\em Inverse Problems}, 28(7):075008, 2012.

\bibitem{raster}
Albert Fannjiang.
\newblock Raster grid pathology and the cure.
\newblock {\em Multiscale Modeling \& Simulation}, 17(3):973--995, 2019.

\bibitem{blind-ptycho}
Albert Fannjiang and Pengwen Chen.
\newblock Blind ptychography: uniqueness \& ambiguities.
\newblock {\em Inverse Problems}, 36:045005.

\bibitem{DRS-ptycho}
Albert Fannjiang and Zheqing Zhang.
\newblock Fixed point analysis of douglas-rachford splitting for ptychography
  and phase retrieval.
\newblock {\em SIAM Journal on Imaging Sciences}, 2020.

\bibitem{FHP10}
A.~Faridian, D.~Hopp, G.~Pedrini, U.~Eigenthaler, M.~Hirscher, and W.~Osten.
\newblock Nanoscale imaging using deep ultraviolet digital holographic
  microscopy.
\newblock {\em Optics Express}, 18(13):14159--14164, 2010.

\bibitem{PIE104}
Helen Mary~Louise Faulkner and JM~Rodenburg.
\newblock Movable aperture lensless transmission microscopy: a novel phase
  retrieval algorithm.
\newblock {\em Physical review letters}, 93(2):023903, 2004.

\bibitem{PIE05}
Helen Mary~Louise Faulkner and John~M Rodenburg.
\newblock Error tolerance of an iterative phase retrieval algorithm for
  moveable illumination microscopy.
\newblock {\em Ultramicroscopy}, 103(2):153--164, 2005.

\bibitem{Fie78}
J.R. Fienup.
\newblock Reconstruction of an object from the modulus of its {F}ourier
  transform.
\newblock {\em Optics Letters}, 3:27--29, 1978.

\bibitem{Fie82}
J.R. Fienup.
\newblock Phase retrieval algorithms: A comparison.
\newblock {\em Applied Optics}, 21(15):2758--2768, 1982.

\bibitem{Fie86}
JR~Fienup and CC~Wackerman.
\newblock Phase-retrieval stagnation problems and solutions.
\newblock {\em JOSA A}, 3(11):1897--1907, 1986.

\bibitem{Gabay}
Michel Fortin and Roland Glowinski.
\newblock {\em Augmented Lagrangian methods: applications to the numerical
  solution of boundary-value problems}.
\newblock Elsevier, 2000.

\bibitem{FouRa13}
Simon Foucart and Holger Rauhut.
\newblock {\em A mathematical introduction to compressive sensing}.
\newblock Springer, 2013.

\bibitem{fuchs2017sic}
Christopher~A Fuchs, Michael~C Hoang, and Blake~C Stacey.
\newblock The sic question: History and state of play.
\newblock {\em Axioms}, 6(3):21, 2017.

\bibitem{gabor1965reconstruction}
D~Gabor, GW~Stroke, D~Brumm, A~Funkhouser, and A~Labeyrie.
\newblock Reconstruction of phase objects by holography.
\newblock {\em Nature}, 208(5016):1159--1162, 1965.

\bibitem{dennis1956improvements}
Dennis Gabor.
\newblock Improvements in and relating to microscopy, 1947.
\newblock Patent GB685286.

\bibitem{gabor1948new}
Dennis Gabor.
\newblock A new microscopic principle.
\newblock {\em Nature}, 161:777–778, 1948.

\bibitem{GS72}
R.W. Gerchberg and W.O. Saxton.
\newblock A practical algorithm for the determination of phase from image and
  diffraction plane pictures.
\newblock {\em Optik}, 35:237--246, 1972.

\bibitem{Mimivirus}
Eric Ghigo, J{\"u}rgen Kartenbeck, Pham Lien, Lucas Pelkmans, Christian Capo,
  Jean-Louis Mege, and Didier Raoult.
\newblock Ameobal pathogen mimivirus infects macrophages through phagocytosis.
\newblock {\em PLoS pathogens}, 4(6), 2008.

\bibitem{Boyd17}
Pontus Giselsson and Stephen Boyd.
\newblock Linear convergence and metric selection for douglas-rachford
  splitting and admm.
\newblock {\em IEEE Transactions on Automatic Control}, 62(2):532--544, 2016.

\bibitem{gladrow2019digital}
Jannes Gladrow.
\newblock Digital phase-only holography using deep conditional generative
  models.
\newblock {\em arXiv preprint arXiv:1911.00904}, 2019.

\bibitem{Glusker84}
Jenny~P. Glusker.
\newblock The patterson function.
\newblock {\em Trends in Biochemical Sciences}, 9(7):328--330, 1984.

\bibitem{noise}
Pierre Godard, Marc Allain, Virginie Chamard, and John Rodenburg.
\newblock Noise models for low counting rate coherent diffraction imaging.
\newblock {\em Optics express}, 20(23):25914--25934, 2012.

\bibitem{goldstein2018phasemax}
Tom Goldstein and Christoph Studer.
\newblock Phasemax: Convex phase retrieval via basis pursuit.
\newblock {\em IEEE Transactions on Information Theory}, 64(4):2675--2689,
  2018.

\bibitem{goodman2005introduction}
Joseph~W Goodman.
\newblock {\em Introduction to {F}ourier optics}.
\newblock Roberts and Company Publishers, 2005.

\bibitem{gray2006toeplitz}
Robert~M Gray et~al.
\newblock Toeplitz and circulant matrices: A review.
\newblock {\em Foundations and Trends{\textregistered} in Communications and
  Information Theory}, 2(3):155--239, 2006.

\bibitem{Grochenig2001}
Karlheinz Gr\"ochenig.
\newblock {\em Foundations of time-frequency analysis}.
\newblock Birkh\"auser, Boston, 2001.

\bibitem{grohs2019}
P.~Grohs, S.~Koppensteiner, and M.~Rathmair.
\newblock Phase retrieval: Uniqueness and stability.
\newblock {\em SIAM Review}, 2019.
\newblock to appear.

\bibitem{gross2011recovering}
David Gross.
\newblock Recovering low-rank matrices from few coefficients in any basis.
\newblock {\em IEEE Transactions on Information Theory}, 57(3):1548--1566,
  2011.

\bibitem{gross2015partial}
David Gross, Felix Krahmer, and Richard Kueng.
\newblock A partial derandomization of phaselift using spherical designs.
\newblock {\em Journal of Fourier Analysis and Applications}, 21(2):229--266,
  2015.

\bibitem{gross2017improved}
David Gross, Felix Krahmer, and Richard Kueng.
\newblock Improved recovery guarantees for phase retrieval from coded
  diffraction patterns.
\newblock {\em Applied and Computational Harmonic Analysis}, 42(1):37--64,
  2017.

\bibitem{guizar2007holography}
Manuel Guizar-Sicairos and James~R Fienup.
\newblock Holography with extended reference by autocorrelation linear
  differential operation.
\newblock {\em Optics express}, 15(26):17592--17612, 2007.

\bibitem{haah2017sample}
Jeongwan Haah, Aram~W Harrow, Zhengfeng Ji, Xiaodi Wu, and Nengkun Yu.
\newblock Sample-optimal tomography of quantum states.
\newblock {\em IEEE Transactions on Information Theory}, 63(9):5628--5641,
  2017.

\bibitem{hand2017phaselift}
Paul Hand.
\newblock Phaselift is robust to a constant fraction of arbitrary errors.
\newblock {\em Applied and Computational Harmonic Analysis}, 42(3):550--562,
  2017.

\bibitem{hand2018phase}
Paul Hand, Oscar Leong, and Vlad Voroninski.
\newblock Phase retrieval under a generative prior.
\newblock In {\em Advances in Neural Information Processing Systems}, pages
  9136--9146, 2018.

\bibitem{hand2016elementary}
Paul Hand and Vladislav Voroninski.
\newblock An elementary proof of convex phase retrieval in the natural
  parameter space via the linear program phasemax.
\newblock {\em arXiv preprint arXiv:1611.03935}, 2016.

\bibitem{Har93}
R.W. Harrison.
\newblock Phase problem in crystallography.
\newblock {\em J. Opt. Soc. Am. A}, 10(5):1045--1055, 1993.

\bibitem{hauptman1997}
Herbert~A Hauptman.
\newblock Shake-and-bake: An algorithm for automatic solution ab initio of
  crystal structures.
\newblock In {\em Methods in enzymology}, volume 277, pages 3--13. Elsevier,
  1997.

\bibitem{Hay82}
M.~Hayes.
\newblock The reconstruction of a multidimensional sequence from the phase or
  magnitude of its {Fourier} transform.
\newblock {\em IEEE Trans. Acoust., Speech, Signal Proc.}, 30:140--154, 1982.

\bibitem{heinosaari2013quantum}
Teiko Heinosaari, Luca Mazzarella, and Michael~M Wolf.
\newblock Quantum tomography under prior information.
\newblock {\em Communications in Mathematical Physics}, 318(2):355--374, 2013.

\bibitem{hoppe1969beugung}
Walter Hoppe.
\newblock Beugung im inhomogenen {P}rim{\"a}rstrahlwellenfeld. {I.} {P}rinzip
  einer {P}hasenmessung von {E}lektronenbeungungsinterferenzen.
\newblock {\em Acta Crystallographica Section A: Crystal Physics, Diffraction,
  Theoretical and General Crystallography}, 25(4):495--501, 1969.

\bibitem{horisaki2016single}
Ryoichi Horisaki, Riki Egami, and Jun Tanida.
\newblock Single-shot phase imaging with randomized light (spiral).
\newblock {\em Optics express}, 24(4):3765--3773, 2016.

\bibitem{horstmeyer2015solving}
Roarke Horstmeyer, Richard~Y Chen, Xiaoze Ou, Brendan Ames, Joel~A Tropp, and
  Changhuei Yang.
\newblock Solving ptychography with a convex relaxation.
\newblock {\em New journal of physics}, 17(5):053044, 2015.

\bibitem{horstmeyer2016diffraction}
Roarke Horstmeyer, Jaebum Chung, Xiaoze Ou, Guoan Zheng, and Changhuei Yang.
\newblock Diffraction tomography with {F}ourier ptychography.
\newblock {\em Optica}, 3(8):827--835, 2016.

\bibitem{huang2017solving}
Wen Huang, Kyle~A Gallivan, and Xiangxiong Zhang.
\newblock Solving phaselift by low-rank riemannian optimization methods for
  complex semidefinite constraints.
\newblock {\em SIAM Journal on Scientific Computing}, 39(5):B840--B859, 2017.

\bibitem{Hur89}
N.~Hurt.
\newblock {\em Phase Retrieval and Zero Crossings}.
\newblock Kluwer Academic Publishers, Norwell, MA, 1989.

\bibitem{iwen2017robust}
Mark Iwen, Aditya Viswanathan, and Yang Wang.
\newblock Robust sparse phase retrieval made easy.
\newblock {\em Applied and Computational Harmonic Analysis}, 42(1):135--142,
  2017.

\bibitem{iwen2016phase}
Mark~A Iwen, Brian Preskitt, Rayan Saab, and Aditya Viswanathan.
\newblock Phase retrieval from local measurements: Improved robustness via
  eigenvector-based angular synchronization.
\newblock {\em arXiv preprint arXiv:1612.01182}, 2016.

\bibitem{MCS97}
H.~N.~Chapman J.~Miao and D.~Sayre.
\newblock .
\newblock {\em Microscopy and Microanalysis 3, supplement 2}, pages 1155--1156,
  1997.

\bibitem{jaganathan2015phase}
Kishore Jaganathan, Yonina Eldar, and Babak Hassibi.
\newblock Phase retrieval with masks using convex optimization.
\newblock In {\em 2015 IEEE International Symposium on Information Theory
  (ISIT)}, pages 1655--1659. IEEE, 2015.

\bibitem{jaganathan2017sparse}
Kishore Jaganathan, Samet Oymak, and Babak Hassibi.
\newblock Sparse phase retrieval: Uniqueness guarantees and recovery
  algorithms.
\newblock {\em IEEE Transactions on Signal Processing}, 65(9):2402--2410, 2017.

\bibitem{jeong2017convergence}
Halyun Jeong and C~Sinan G{\"u}nt{\"u}rk.
\newblock Convergence of the randomized kaczmarz method for phase retrieval.
\newblock {\em arXiv preprint arXiv:1706.10291}, 2017.

\bibitem{Pfeiffer}
I.~Johnson, K.~Jefimovs, O.~Bunk, C.~David, M.~Dierolf, J.~Gray, D.~Renker, and
  F.~Pfeiffer.
\newblock Coherent diffractive imaging using phase front modifications.
\newblock {\em Phys. Rev. Lett.}, 100(15):155503, Apr 2008.

\bibitem{johnston2005white}
Sean~F Johnston.
\newblock From white elephant to nobel prize: Dennis gabor's wavefront
  reconstruction.
\newblock {\em Hist Stud Phys Biol Sci}, 36(1):35--70, 2005.

\bibitem{jung2017blind}
Peter Jung, Felix Krahmer, and Dominik St{\"o}ger.
\newblock Blind demixing and deconvolution at near-optimal rate.
\newblock {\em IEEE Transactions on Information Theory}, 64(2):704--727, 2017.

\bibitem{Kac}
S~Kaczmarz.
\newblock {Angen\"{a}herte Aufl\"{o}sung von Systemen linearer Gleichungen}.
\newblock {\em Bull. Internat. Acad. Pol. Sci. Lett. Ser. A}, 35:355--357,
  1937.

\bibitem{kikuta1972x}
S~Kikuta, S~Aoki, S~Kosaki, and K~Kohra.
\newblock X-ray holography of lensless {F}ourier-transform type.
\newblock {\em Optics Communications}, 5(2):86--89, 1972.

\bibitem{kim2019fourier}
Kyung-Su Kim and Sae-Young Chung.
\newblock Fourier phase retrieval with extended support estimation via deep
  neural network.
\newblock {\em IEEE Signal Processing Letters}, 26(10):1506--1510, 2019.

\bibitem{KST95}
M.V. Klibanov, P.E. Sacks, and A.V. Tikhonravov.
\newblock The phase retrieval problem.
\newblock {\em Inverse problems}, 11:1--28, 1995.

\bibitem{noise3}
AP~Konijnenberg, WMJ Coene, and HP~Urbach.
\newblock Model-independent noise-robust extension of ptychography.
\newblock {\em Optics express}, 26(5):5857--5874, 2018.

\bibitem{krahmer2019complex}
Felix Krahmer and Dominik St{\"o}ger.
\newblock Complex phase retrieval from subgaussian measurements.
\newblock {\em arXiv preprint arXiv:1906.08385}, 2019.

\bibitem{kueng2017low}
Richard Kueng, Holger Rauhut, and Ulrich Terstiege.
\newblock Low rank matrix recovery from rank one measurements.
\newblock {\em Applied and Computational Harmonic Analysis}, 42(1):88--116,
  2017.

\bibitem{kueng2016low}
Richard Kueng, Huangjun Zhu, and David Gross.
\newblock Low rank matrix recovery from {Clifford} orbits.
\newblock {\em arXiv preprint arXiv:1610.08070}, 2016.

\bibitem{kumar2014fiber}
Shiva Kumar and M~Jamal Deen.
\newblock {\em Fiber optic communications: fundamentals and applications}.
\newblock John Wiley \& Sons, 2014.

\bibitem{Latychevskaia19}
Tatiana Latychevskaia.
\newblock Iterative phase retrieval for digital holography: tutorial.
\newblock {\em J. Opt. Soc. Am. A}, 36(12):D31--D40, Dec 2019.

\bibitem{latychevskaia2015practical}
Tatiana Latychevskaia and Hans-Werner Fink.
\newblock Practical algorithms for simulation and reconstruction of digital
  in-line holograms.
\newblock {\em Applied optics}, 54(9):2424--2434, 2015.

\bibitem{latychevskaia2012holography}
Tatiana Latychevskaia, Jean-Nicolas Longchamp, and Hans-Werner Fink.
\newblock When holography meets coherent diffraction imaging.
\newblock {\em Optics express}, 20(27):28871--28892, 2012.

\bibitem{li2018nett}
Housen Li, Johannes Schwab, Stephan Antholzer, and Markus Haltmeier.
\newblock {NETT}: Solving inverse problems with deep neural networks.
\newblock {\em Preprint, arXiv: 1803.00092}, 2018.

\bibitem{raar17}
Ji~Li and Tie Zhou.
\newblock On relaxed averaged alternating reflections (raar) algorithm for
  phase retrieval with structured illumination.
\newblock {\em Inverse Problems}, 33(2):025012, 2017.

\bibitem{li2019rapid}
Xiaodong Li, Shuyang Ling, Thomas Strohmer, and Ke~Wei.
\newblock Rapid, robust, and reliable blind deconvolution via nonconvex
  optimization.
\newblock {\em Applied and computational harmonic analysis}, 47(3):893--934,
  2019.

\bibitem{li2013sparse}
Xiaodong Li and Vladislav Voroninski.
\newblock Sparse signal recovery from quadratic measurements via convex
  programming.
\newblock {\em SIAM Journal on Mathematical Analysis}, 45(5):3019--3033, 2013.

\bibitem{li2016identifiability}
Yanjun Li, Kiryung Lee, and Yoram Bresler.
\newblock Identifiability in blind deconvolution with subspace or sparsity
  constraints.
\newblock {\em IEEE Transactions on information Theory}, 62(7):4266--4275,
  2016.

\bibitem{ling2015self}
Shuyang Ling and Thomas Strohmer.
\newblock Self-calibration and biconvex compressive sensing.
\newblock {\em Inverse Problems}, 31(11):115002, 2015.

\bibitem{ling2017blind}
Shuyang Ling and Thomas Strohmer.
\newblock Blind deconvolution meets blind demixing: Algorithms and performance
  bounds.
\newblock {\em IEEE Transactions on Information Theory}, 63(7):4497--4520,
  2017.

\bibitem{ling2019regularized}
Shuyang Ling and Thomas Strohmer.
\newblock Regularized gradient descent: a non-convex recipe for fast joint
  blind deconvolution and demixing.
\newblock {\em Information and Inference: A Journal of the IMA}, 8(1):1--49,
  2019.

\bibitem{Liu08}
Y.J Liu and et~al.
\newblock Phase retrieval in x-ray imaging based on using structured
  illumination.
\newblock {\em Phys. Rev. A}, 78:023817, 2008.

\bibitem{LP97}
E.G. Loewen and E.~Popov.
\newblock {\em Diffraction Gratings and Applications}.
\newblock Marcel Dekker, 1997.

\bibitem{loh2010cryptotomography}
ND~Loh, Michael~J Bogan, Veit Elser, Anton Barty, S{\'e}bastien Boutet,
  Sa{\v{s}}a Bajt, Janos Hajdu, Tomas Ekeberg, Filipe~RNC Maia, Joachim Schulz,
  et~al.
\newblock Cryptotomography: reconstructing {3D Fourier} intensities from
  randomly oriented single-shot diffraction patterns.
\newblock {\em Physical review letters}, 104(22):225501, 2010.

\bibitem{Longchamp1474}
Jean-Nicolas Longchamp, Stephan Rauschenbach, Sabine Abb, Conrad Escher,
  Tatiana Latychevskaia, Klaus Kern, and Hans-Werner Fink.
\newblock Imaging proteins at the single-molecule level.
\newblock {\em Proceedings of the National Academy of Sciences},
  114(7):1474--1479, 2017.

\bibitem{Lu17}
Yue~M Lu and Gen Li.
\newblock Phase transitions of spectral initialization for high-dimensional
  nonconvex estimation.
\newblock {\em arXiv preprint arXiv:1702.06435}, 2017.

\bibitem{Luke}
D~Russell Luke.
\newblock Relaxed averaged alternating reflections for diffraction imaging.
\newblock {\em Inverse problems}, 21(1):37, 2004.

\bibitem{Luke2}
D~Russell Luke.
\newblock Finding best approximation pairs relative to a convex and
  prox-regular set in a hilbert space.
\newblock {\em SIAM Journal on Optimization}, 19(2):714--739, 2008.

\bibitem{luke2017phase}
D~Russell Luke.
\newblock Phase retrieval, what’s new.
\newblock {\em SIAG/OPT Views and News}, 25(1):1--5, 2017.

\bibitem{LBL02}
D.R. Luke, J.V. Burke, and R.G. Lyon.
\newblock Optical wavefront reconstruction: {T}heory and numerical methods.
\newblock {\em SIAM Rev.}, 44(2):169--224, 2002.

\bibitem{luo2019optimal}
Wangyu Luo, Wael Alghamdi, and Yue~M Lu.
\newblock Optimal spectral initialization for signal recovery with applications
  to phase retrieval.
\newblock {\em IEEE Transactions on Signal Processing}, 67(9):2347--2356, 2019.

\bibitem{ma2018implicit}
Cong Ma, Kaizheng Wang, Yuejie Chi, and Yuxin Chen.
\newblock Implicit regularization in nonconvex statistical estimation: Gradient
  descent converges linearly for phase retrieval, matrix completion, and blind
  deconvolution.
\newblock {\em Foundations of Computational Mathematics}, pages 1--182, 2018.

\bibitem{ptycho-rpi}
AM~Maiden, GR~Morrison, B~Kaulich, A~Gianoncelli, and JM~Rodenburg.
\newblock Soft x-ray spectromicroscopy using ptychography with randomly phased
  illumination.
\newblock {\em Nature communications}, 4(1):1--6, 2013.

\bibitem{rPIE17}
Andrew Maiden, Daniel Johnson, and Peng Li.
\newblock Further improvements to the ptychographical iterative engine.
\newblock {\em Optica}, 4(7):736--745, 2017.

\bibitem{ePIE09}
Andrew~M Maiden and John~M Rodenburg.
\newblock An improved ptychographical phase retrieval algorithm for diffractive
  imaging.
\newblock {\em Ultramicroscopy}, 109(10):1256--1262, 2009.

\bibitem{Mar07}
S.~Marchesini.
\newblock A unified evaluation of iterative projection algorithms for phase
  retrieval.
\newblock {\em Rev. Sci. Inst.}, 78:011301 11--10, 2007.

\bibitem{Marchesini2016SHARP}
Stefano Marchesini, Hari Krishnan, Benedikt~J Daurer, David~A Shapiro, Talita
  Perciano, James~A Sethian, and Filipe~RNC Maia.
\newblock {SHARP}: a distributed {GPU}-based ptychographic solver.
\newblock {\em Journal of applied crystallography}, 49(4):1245--1252, 2016.

\bibitem{marchesini2019shaping}
Stefano Marchesini and Anne Sakdinawat.
\newblock Shaping coherent x-rays with binary optics.
\newblock {\em Optics express}, 27(2):907--917, 2019.

\bibitem{Mesbahi97}
M.~Mesbahi and G.~P. Papavassilopoulos.
\newblock On the rank minimization problem over a positive semidefinite linear
  matrix inequality.
\newblock {\em IEEE Transactions on Automatic Control}, 42(2):239--243, 1997.

\bibitem{metzler2016bm3d}
Christopher~A Metzler, Arian Maleki, and Richard~G Baraniuk.
\newblock Bm3d-prgamp: Compressive phase retrieval based on bm3d denoising.
\newblock In {\em 2016 IEEE International Conference on Image Processing
  (ICIP)}, pages 2504--2508. IEEE, 2016.

\bibitem{metzler2018prdeep}
Christopher~A Metzler, Philip Schniter, Ashok Veeraraghavan, and Richard~G
  Baraniuk.
\newblock prdeep: Robust phase retrieval with a flexible deep network.
\newblock {\em arXiv preprint arXiv:1803.00212}, 2018.

\bibitem{metzler2017coherent}
Christopher~A Metzler, Manoj~K Sharma, Sudarshan Nagesh, Richard~G Baraniuk,
  Oliver Cossairt, and Ashok Veeraraghavan.
\newblock Coherent inverse scattering via transmission matrices: Efficient
  phase retrieval algorithms and a public dataset.
\newblock In {\em 2017 IEEE International Conference on Computational
  Photography (ICCP)}, pages 1--16. IEEE, 2017.

\bibitem{MIS08}
J.~Miao, T.~Ishikawa, Q.~Shen, and T.~Earnest.
\newblock Extending {X-Ray} crystallography to allow the imaging of
  noncrystalline materials, cells and single protein complexes.
\newblock {\em Annu.\ Rev.\ Phys.\ Chem.}, 59:387--410, 2008.

\bibitem{miao1999extending}
Jianwei Miao, Pambos Charalambous, Janos Kirz, and David Sayre.
\newblock Extending the methodology of x-ray crystallography to allow imaging
  of micrometre-sized non-crystalline specimens.
\newblock {\em Nature}, 400(6742):342, 1999.

\bibitem{Miao2}
Jianwei Miao, J~Kirz, and D~Sayre.
\newblock The oversampling phasing method.
\newblock {\em Acta Crystallographica Section D: Biological Crystallography},
  56(10):1312--1315, 2000.

\bibitem{Miao}
Jianwei Miao, David Sayre, and HN~Chapman.
\newblock Phase retrieval from the magnitude of the {F}ourier transforms of
  nonperiodic objects.
\newblock {\em JOSA A}, 15(6):1662--1669, 1998.

\bibitem{Mil90}
R.P. Millane.
\newblock Phase retrieval in crystallography and optics.
\newblock {\em J. Opt. Soc. Am. A.}, 7:394–--411, 1990.

\bibitem{Mil06}
R.P. Millane.
\newblock Recent advances in phase retrieval.
\newblock In P.J. Bones, M.A. Fiddy, and R.P. Millane, editors, {\em Image
  Reconstruction from Incomplete Data IV}, volume 6316 of {\em Proc. SPIE},
  pages 63160E/1--11, 2006.

\bibitem{Mis73}
D.L. Misell.
\newblock A method for the solution of the phase problem in electron
  microscopy.
\newblock {\em J. Phys. D: App. Phy.}, 6(1):L6--L9, 1973.

\bibitem{Mont}
Marco Mondelli and Andrea Montanari.
\newblock Fundamental limits of weak recovery with applications to phase
  retrieval.
\newblock {\em Foundations of Computational Mathematics}, 19(3):703--773, 2019.

\bibitem{monteiro1997primal}
Renato~DC Monteiro.
\newblock Primal--dual path-following algorithms for semidefinite programming.
\newblock {\em SIAM Journal on Optimization}, 7(3):663--678, 1997.

\bibitem{nawab1983signal}
S~Nawab, T~Quatieri, and Jae Lim.
\newblock Signal reconstruction from short-time {F}ourier transform magnitude.
\newblock {\em IEEE Transactions on Acoustics, Speech, and Signal Processing},
  31(4):986--998, 1983.

\bibitem{NesterovBook}
Y.~Nesterov.
\newblock {\em Introductory Lectures on Convex Optimization: A Basic Course},
  volume~87 of {\em Applied Optimization}.
\newblock Kluwer, Boston, 2004.

\bibitem{neutze2000potential}
Richard Neutze, Remco Wouts, David Van~der Spoel, Edgar Weckert, and Janos
  Hajdu.
\newblock Potential for biomolecular imaging with femtosecond x-ray pulses.
\newblock {\em Nature}, 406(6797):752--757, 2000.

\bibitem{ohlsson2011compressive}
Henrik Ohlsson, Allen~Y Yang, Roy Dong, and S~Shankar Sastry.
\newblock Compressive phase retrieval from squared output measurements via
  semidefinite programming.
\newblock {\em IFAC Proceedings}, 45(16):89--94, 2012.

\bibitem{paris2004quantum}
Matteo Paris and Jaroslav Rehacek.
\newblock {\em Quantum state estimation}, volume 649.
\newblock Springer Science \& Business Media, 2004.

\bibitem{random-aperture}
Xiaopeng Peng, Garreth~J Ruane, Marco~B Quadrelli, and Grover~A Swartzlander.
\newblock Randomized apertures: high resolution imaging in far field.
\newblock {\em Optics express}, 25(15):18296--18313, 2017.

\bibitem{pfander2019robust}
Götz~E Pfander and Palina Salanevich.
\newblock Robust phase retrieval algorithm for time-frequency structured
  measurements.
\newblock {\em SIAM Journal on Imaging Sciences}, 12(2):736--761, 2019.

\bibitem{pfeiffer2018x}
Franz Pfeiffer.
\newblock X-ray ptychography.
\newblock {\em Nature Photonics}, 12(1):9--17, 2018.

\bibitem{Boche15}
V.~Pohl, F.~Yang, and H.~Boche.
\newblock Phase retrieval from low-rate samples.
\newblock {\em Sampling Theory in Signal and Image Processing}, 14(1):71--99,
  Jan 2015.

\bibitem{qian2014efficient}
Jianliang Qian, Chao Yang, A~Schirotzek, F~Maia, and S~Marchesini.
\newblock Efficient algorithms for ptychographic phase retrieval.
\newblock {\em Inverse Problems and Applications, Contemp. Math}, 615:261--280,
  2014.

\bibitem{rauhut2017low}
Holger Rauhut, Reinhold Schneider, and {\v{Z}}eljka Stojanac.
\newblock Low rank tensor recovery via iterative hard thresholding.
\newblock {\em Linear Algebra and its Applications}, 523:220--262, 2017.

\bibitem{raz2014direct}
Oren Raz, Ben Leshem, Jianwei Miao, Boaz Nadler, Dan Oron, and Nirit Dudovich.
\newblock Direct phase retrieval in double blind fourier holography.
\newblock {\em Optics express}, 22(21):24935--24950, 2014.

\bibitem{recht2010guaranteed}
Benjamin Recht, Maryam Fazel, and Pablo~A Parrilo.
\newblock Guaranteed minimum-rank solutions of linear matrix equations via
  nuclear norm minimization.
\newblock {\em SIAM review}, 52(3):471--501, 2010.

\bibitem{Rei44}
H.~Reichenbach.
\newblock {\em Philosophic Foundations of Quantum Mechanics}.
\newblock University of California Press, Berkeley, 1944.

\bibitem{rivenson2018phase}
Yair Rivenson, Yibo Zhang, Harun G{\"u}nayd{\i}n, Da~Teng, and Aydogan Ozcan.
\newblock Phase recovery and holographic image reconstruction using deep
  learning in neural networks.
\newblock {\em Light: Science \& Applications}, 7(2):17141--17141, 2018.

\bibitem{Rod08}
J.M. Rodenburg.
\newblock Ptychography and related diffractive imaging methods.
\newblock {\em Advances in Imaging and Electron Physics, vol. 150},
  150:87--184, 2008.

\bibitem{PIE204}
John~M Rodenburg and Helen~ML Faulkner.
\newblock A phase retrieval algorithm for shifting illumination.
\newblock {\em Applied physics letters}, 85(20):4795--4797, 2004.

\bibitem{saliba2016novel}
M~Saliba, J~Bosgra, AD~Parsons, UH~Wagner, C~Rau, and P~Thibault.
\newblock Novel methods for hard x-ray holographic lensless imaging.
\newblock {\em Microscopy and Microanalysis}, 22(S3):110--111, 2016.

\bibitem{saliba2012fourier}
M~Saliba, T~Latychevskaia, J~Longchamp, and H~Fink.
\newblock Fourier transform holography: a lensless non-destructive imaging
  technique.
\newblock {\em Microscopy and Microanalysis}, 18(S2):564--565, 2012.

\bibitem{San85}
J.L.C. Sanz.
\newblock Mathematical considerations for the problem of {Fourier} transform
  phase retrieval frommagnitude.
\newblock {\em SIAM Journal on Applied Mathematics}, 45(4):651--664, 1985.

\bibitem{Sca06}
G.~Scapin.
\newblock Structural biology and drug discovery.
\newblock {\em Current Pharmaceutical Design}, 12:2087--2097, 2006.

\bibitem{schniter2014compressive}
Philip Schniter and Sundeep Rangan.
\newblock Compressive phase retrieval via generalized approximate message
  passing.
\newblock {\em IEEE Transactions on Signal Processing}, 63(4):1043--1055, 2014.

\bibitem{schniter2015message}
Philip Schniter and Sundeep Rangan.
\newblock A message-passing approach to phase retrieval of sparse signals.
\newblock In {\em Excursions in Harmonic Analysis, Volume 4}, pages 177--204.
  Springer, 2015.

\bibitem{Schwarz}
Hermann~Amandus Schwarz.
\newblock {\em {Ueber einen Grenz\"{u}bergang durch alternirendes Verfahren}},
  volume~15.
\newblock 1870.

\bibitem{scott2010symmetric}
Andrew~James Scott and Markus Grassl.
\newblock Symmetric informationally complete positive-operator-valued measures:
  A new computer study.
\newblock {\em Journal of Mathematical Physics}, 51(4):042203, 2010.

\bibitem{seaberg2015coherent}
Matthew~H Seaberg, Alexandre d'Aspremont, and Joshua~J Turner.
\newblock Coherent diffractive imaging using randomly coded masks.
\newblock {\em Applied Physics Letters}, 107(23):231103, 2015.

\bibitem{seelamantula2011exact}
Chandra~Sekhar Seelamantula, Nicolas Pavillon, Christian Depeursinge, and
  Michael Unser.
\newblock Exact complex-wave reconstruction in digital holography.
\newblock {\em JOSA A}, 28(6):983--992, 2011.

\bibitem{shechtman2014gespar}
Yoav Shechtman, Amir Beck, and Yonina~C Eldar.
\newblock {GESPAR}: Efficient phase retrieval of sparse signals.
\newblock {\em IEEE transactions on signal processing}, 62(4):928--938, 2014.

\bibitem{shechtman2015phase}
Yoav Shechtman, Yonina~C Eldar, Oren Cohen, Henry~Nicholas Chapman, Jianwei
  Miao, and Mordechai Segev.
\newblock Phase retrieval with application to optical imaging: a contemporary
  overview.
\newblock {\em IEEE signal processing magazine}, 32(3):87--109, 2015.

\bibitem{singer2018mathematics}
Amit Singer.
\newblock Mathematics for cryo-electron microscopy.
\newblock {\em Proceedings of the International Congress of Mathematicians},
  2018.

\bibitem{SH03}
Thomas Strohmer and Robert Heath.
\newblock Grassmannian frames with applications to coding and communication.
\newblock {\em arXiv preprint math/0301135}, 2003.

\bibitem{strohmer2009randomized}
Thomas Strohmer and Roman Vershynin.
\newblock A randomized kaczmarz algorithm with exponential convergence.
\newblock {\em Journal of Fourier Analysis and Applications}, 15(2):262, 2009.

\bibitem{sun2018geometric}
Ju~Sun, Qing Qu, and John Wright.
\newblock A geometric analysis of phase retrieval.
\newblock {\em Foundations of Computational Mathematics}, 18(5):1131--1198,
  2018.

\bibitem{sun2016guaranteed}
Ruoyu Sun and Zhi-Quan Luo.
\newblock Guaranteed matrix completion via non-convex factorization.
\newblock {\em IEEE Transactions on Information Theory}, 62(11):6535--6579,
  2016.

\bibitem{tan2019phase}
Yan~Shuo Tan and Roman Vershynin.
\newblock Phase retrieval via randomized kaczmarz: Theoretical guarantees.
\newblock {\em Information and Inference: A Journal of the IMA}, 8(1):97--123,
  2019.

\bibitem{TDB09}
P.~Thibault, M.~Dierolf, O.~Bunk, A.~Menzel, and F.~Pfeiffer.
\newblock Probe retrieval in ptychographic coherent diffractive imaging.
\newblock {\em Ultramicroscopy}, 109:338--343, 2009.

\bibitem{ML12}
P~Thibault and M~Guizar-Sicairos.
\newblock Maximum-likelihood refinement for coherent diffractive imaging.
\newblock {\em New Journal of Physics}, 14(6):063004, 2012.

\bibitem{DM08}
Pierre Thibault, Martin Dierolf, Andreas Menzel, Oliver Bunk, Christian David,
  and Franz Pfeiffer.
\newblock High-resolution scanning x-ray diffraction microscopy.
\newblock {\em Science}, 321(5887):379--382, 2008.

\bibitem{tillmann2016dolphin}
Andreas~M Tillmann, Yonina~C Eldar, and Julien Mairal.
\newblock {DOLPHIn} - {D}ictionary learning for phase retrieval.
\newblock {\em IEEE Transactions on Signal Processing}, 64(24):6485--6500,
  2016.

\bibitem{toh1999sdpt3}
Kim-Chuan Toh, Michael~J Todd, and Reha~H T{\"u}t{\"u}nc{\"u}.
\newblock {SDPT3—a MATLAB} software package for semidefinite programming,
  version 1.3.
\newblock {\em Optimization methods and software}, 11(1-4):545--581, 1999.

\bibitem{tu2015low}
Stephen Tu, Ross Boczar, Max Simchowitz, Mahdi Soltanolkotabi, and Benjamin
  Recht.
\newblock Low-rank solutions of linear matrix equations via {P}rocrustes flow.
\newblock {\em arXiv preprint arXiv:1507.03566}, 2015.

\bibitem{vinzant2015small}
Cynthia Vinzant.
\newblock A small frame and a certificate of its injectivity.
\newblock In {\em 2015 International Conference on Sampling Theory and
  Applications (SampTA)}, pages 197--200. IEEE, 2015.

\bibitem{Neuman}
John Von~Neumann.
\newblock {\em Functional operators: Measures and integrals}, volume~1.
\newblock Princeton University Press, 1950.

\bibitem{waldspurger2015phase}
Ir{\`e}ne Waldspurger, Alexandre d’Aspremont, and St{\'e}phane Mallat.
\newblock Phase recovery, maxcut and complex semidefinite programming.
\newblock {\em Mathematical Programming}, 149(1-2):47--81, 2015.

\bibitem{Wal63}
A.~Walther.
\newblock The question of phase retrieval in optics.
\newblock {\em Opt. Acta}, 10:41--49, 1963.

\bibitem{TAF18}
Gang Wang, Georgios~B Giannakis, and Yonina~C Eldar.
\newblock Solving systems of random quadratic equations via truncated amplitude
  flow.
\newblock {\em IEEE Transactions on Information Theory}, 64(2):773--794, 2018.

\bibitem{wang2017sparse}
Gang Wang, Liang Zhang, Georgios~B Giannakis, Mehmet Ak{\c{c}}akaya, and Jie
  Chen.
\newblock Sparse phase retrieval via truncated amplitude flow.
\newblock {\em IEEE Transactions on Signal Processing}, 66(2):479--491, 2017.

\bibitem{wei2015solving}
Ke~Wei.
\newblock Solving systems of phaseless equations via kaczmarz methods: A proof
  of concept study.
\newblock {\em Inverse Problems}, 31(12):125008, 2015.

\bibitem{Wie32}
Norbert Wiener.
\newblock Tauberian theorems.
\newblock {\em Ann. of Math., (2)}, 33(1):1--100, 1932.

\bibitem{Waller15}
Li-Hao Yeh, Jonathan Dong, Jingshan Zhong, Lei Tian, Michael Chen, Gongguo
  Tang, Mahdi Soltanolkotabi, and Laura Waller.
\newblock Experimental robustness of {F}ourier ptychography phase retrieval
  algorithms.
\newblock {\em Optics express}, 23(26):33214--33240, 2015.

\bibitem{yuan2019phase}
Ziyang Yuan, Hongxia Wang, and Qi~Wang.
\newblock Phase retrieval via sparse {Wirtinger} flow.
\newblock {\em Journal of Computational and Applied Mathematics}, 355:162--173,
  2019.

\bibitem{zaunerquantendesigns}
Gerhard Zauner.
\newblock basics of a non-commutative design theory.
\newblock {\em Ph. D. dissertation, PhD thesis}, 1999.

\bibitem{zhang2016phase}
Fucai Zhang, Bo~Chen, Graeme~R Morrison, Joan Vila-Comamala, Manuel
  Guizar-Sicairos, and Ian~K Robinson.
\newblock Phase retrieval by coherent modulation imaging.
\newblock {\em Nature communications}, 7(1):1--8, 2016.

\bibitem{zhang2018fast}
Gong Zhang, Tian Guan, Zhiyuan Shen, Xiangnan Wang, Tao Hu, Delai Wang,
  Yonghong He, and Ni~Xie.
\newblock Fast phase retrieval in off-axis digital holographic microscopy
  through deep learning.
\newblock {\em Optics express}, 26(15):19388--19405, 2018.

\bibitem{noise2}
Yongbing Zhang, Pengming Song, and Qionghai Dai.
\newblock Fourier ptychographic microscopy using a generalized {A}nscombe
  transform approximation of the mixed poisson-gaussian likelihood.
\newblock {\em Optics express}, 25(1):168--179, 2017.

\bibitem{adaptive}
Chao Zuo, Jiasong Sun, and Qian Chen.
\newblock Adaptive step-size strategy for noise-robust {F}ourier ptychographic
  microscopy.
\newblock {\em Optics express}, 24(18):20724--20744, 2016.

\end{thebibliography}

\end{document}